\documentclass[11pt]{article}
\usepackage[utf8]{inputenc}
\usepackage{dsfont}
\usepackage{graphicx}% Include figure files
\usepackage{xcolor}
\usepackage{amsmath,amsthm,amssymb}
\usepackage[margin=1in]{geometry}
\usepackage[pdftex,colorlinks=true,linkcolor=blue,citecolor=blue,urlcolor=purple]{hyperref}
\usepackage{comment}
\usepackage[normalem]{ulem}
\usepackage{cite}
\usepackage{tcolorbox}
\usepackage[boxsize=0.4em]{ytableau}

\def\ba{\begin{array}}
\def\ea{\end{array}}

\def\0{{\bf 0}}

\def\f{{\bf f}}

\def\s{{\bf s}}

\def\G{{\bf G}}
\def\U{{\mathbb U}}

\newcommand{\be}{\begin{equation}}
\newcommand{\ee}{\end{equation}}
\newcommand{\bea}{\begin{eqnarray}}
\newcommand{\eea}{\end{eqnarray}}
\newcommand{\bes}{\begin{equation*}}
\newcommand{\ees}{\end{equation*}}
\newcommand{\beas}{\begin{eqnarray*}}
\newcommand{\eeas}{\end{eqnarray*}}

\makeatletter
\newtheorem*{rep@theorem}{\rep@title}
\newcommand{\newreptheorem}[2]{%
\newenvironment{rep#1}[1]{%
 \def\rep@title{#2 \ref{##1} (restated)}%
 \begin{rep@theorem}}%
 {\end{rep@theorem}}}
\makeatother

\allowdisplaybreaks

\newtheorem{thm}{Theorem}
\newtheorem*{thm*}{Theorem}

\newtheorem{lem}[thm]{Lemma}
\newtheorem*{lem*}{Lemma}

\newtheorem{prop}[thm]{Proposition}
\newtheorem{defn}[thm]{Definition}
\newtheorem{rem}[thm]{Remark}

\newtheorem{dfn}{Definition}
\newreptheorem{thm}{Theorem}
\newreptheorem{lem}{Lemma}
\newreptheorem{cor}{Corollary}

\usepackage{pgfplots}

\usepackage[title]{appendix}
\usepackage{authblk}

\title{Joint moments of higher order derivatives of CUE characteristic polynomials II: Structures, recursive relations, and applications }

\author[1]{Jonathan P. Keating\thanks{keating@maths.ox.ac.uk}}
\author[2]{Fei Wei\thanks{weif@mail.tsinghua.edu.cn
}}
\affil[1]{Mathematical Institute,
University of Oxford, Oxford OX2 6GG, UK}
\affil[2]{Yau Mathematical Sciences Center, Tsinghua University}

\date{\today}

\begin{document}

\maketitle

\begin{abstract}

In a companion paper \cite{jon-fei}, we established asymptotic formulae for the joint moments of higher order derivatives of the characteristic polynomials of CUE random matrices. The leading order coefficients of these asymptotic formulae are expressed as partition sums of derivatives of determinants of Hankel matrices involving I-Bessel functions, with column indices shifted by Young diagrams. In this paper, we continue the study of these joint moments and establish more properties for their leading order coefficients, including structure theorems and recursive relations. We also build a connection to a solution of the $\sigma$-Painlev\'{e} III$'$ equation. In the process, we give recursive formulae for the Taylor coefficients of the Hankel determinants formed from I-Bessel functions that appear and find differential equations that these determinants satisfy. The approach we establish is applicable to determinants of general Hankel matrices whose columns are shifted by Young diagrams.

\end{abstract}


\section{Introduction}

There are many deep connections between the theory of Painlev\'{e} equations and random matrix theory. For example,
Tracy and Widom \cite{tracy1994fredholm} expressed the limit distribution of the largest eigenvalue, suitably normalized, of GUE (Gaussian Unitary Ensemble) random matrices in terms of a solution of the Painlev\'{e} II differential equation. In a series of works \cite{forrester2001application,forrester2002application}, Forrester and Witte applied the $\tau$-function theory of Painlev\'{e} equations to study certain averages with respect to the probability density functions of various random matrix ensembles, including the LUE (Laguerre Unitary Ensemble), JUE (Jacobi Unitary Ensemble) and CUE (Circular Unitary Ensemble -- the space of unitary matrices of a given size, endowed with the Haar measure). 
Recently in \cite{basor2019representation}, Basor, et al.~used the Riemann-Hilbert method to establish connections between solutions of the $\sigma$-Painlev\'{e} III$'$ and V equations and the joint moments of the characteristic polynomials from the CUE and its first order derivative. Using these connections, they established recursive formulae, structure results, and extensions of the joint moments.  For more on the roles of solutions of the Painlev\'{e} equations in many other aspects of random matrix theory, we refer readers to (e.g., \cite{forrester2015painleve,ItsAlexander}).

As a continuation of our previous work \cite{jon-fei}, the motivation in the present paper is to use the theory of Painlev\'{e} equations to study joint moments of higher order derivatives of characteristic polynomials from the CUE when the matrix size goes to infinity. Specifically, we give explicit recursive formulae for the joint moments and study their properties. Our methods are combinatorial, and so are quite different from the Riemann-Hilbert method used in \cite{basor2019representation} for the first-order derivative. We also give applications of our methods, e.g., we give recursive formulae for the Taylor coefficients of a Hankel determinant defined in terms of I-Bessel functions. This determinant is known to be associated with a solution of the  $\sigma$-Painlev\'{e} III$'$ equation.  

We refer readers to the introduction to the companion paper \cite{jon-fei} for further context and a more extensive review of the previous literature.

\subsection{Main results}

Let $A\in \U(N)$ be taken from the Circular Unitary Ensemble (CUE) of random matrices. Let $\Lambda_{A}(s)$ be the characteristic polynomial of $A$ given by
\[
\Lambda_A(s) = \prod_{n=1}^N (1-se^{-\i\theta_n}),
\]
where we set $\i^2=-1$ to avoid confusion with the index $i$ and  $e^{\i\theta_1}, \ldots,e^{\i\theta_{N}}$ are the eigenvalues of $A$. Define
\[
Z_{A}(s) := e^{-\pi \i N/2}e^{\i \sum_{n=1}^{N}\theta_{n}/2}
s^{- N/2}\Lambda_{A}(s),
\]
where for $s^{-N/2}$, when $N$ is an odd integer, we use the branch of the square-root function that is positive for positive real $s$.
$Z_A(s)$ satisfies a functional equation 
$Z_A(s) = (-1)^N Z_{A^*}(1/s)$, where $A^*$ is the conjugate transpose of $A$. This implies that $Z_A(e^{\i \theta})$ is real when $\theta$ is real. 

In \cite{jon-fei}, we give explicit formulae for the leading term in the asymptotic expression of for $\int_{\U(N)} |Z_A^{(n_{1})} (1)|^{2M} |Z_A^ {(n_{2})}(1)|^{2k-2M} dA_N$ for arbitrary non-negative integers $n_1,n_2$ and $k,M$ with $M\leq k$. In particular, the leading order coefficient of the asymptotic formula given in \cite[Theorem 24]{jon-fei} is 
expressed as partition sums of derivatives of determinants of Hankel matrices of I-Bessel functions whose columns are shifted
by Young diagrams. To better understand the asymptotic formula, it is necessary to investigate the structures of these determinants. We focus below mainly on the representative case $n_1=2,n_2=0$. Our methods extend directly to the general cases (see Section \ref{section:Generalization}). 

\begin{prop}[Theorem 3 of \cite{jon-fei}]
\label{intro:prop1}
Let $k\geq 1$, $0\leq M\leq k$ be integers. Then we have
\bea
F_2(M,k) &:=& \lim_{N\rightarrow \infty} \frac{1}{N^{k^2+4M}}
\int_{\U(N)} |Z_A''(1)|^{2M} |Z_A(1)|^{2k-2M} dA_N \nonumber \\
&=& (-1)^{\frac{k(k-1)}{2}} 
\sum_{l=0}^{2M} \binom{2M}{l}
\left( \frac{d}{dx} \right)^{4M-2l}  \left( e^{-\frac{x}{2}} x^{-l-\frac{k^2}{2} } f_{l}(x) \right) \Bigg|_{x= 0},
\label{expression of bkm}
\eea
where 
\bea
\label{410definition of f_{l}}
f_{l}(x)= \sum_{\substack{l_1+\cdots+l_k=l \\ l_1\geq 0,\ldots,l_k\geq 0}} \binom{l}{l_1,\ldots,l_k}\det \Big(I_{i+j+1+2l_{j+1}}(2\sqrt{x})\Big)_{i,j=0,\ldots,k-1},
\eea
and $I_{n}(x)=(x/2)^{n}\sum_{j=0}^{\infty}\frac{x^{2j}}{2^{2j}(n+j)!j!}$ is the modified Bessel function of the first kind.
\end{prop}

In this paper, we explore more intrinsic properties of $f_l(x)$. Let
\bea\label{definition of tauk}
\tau_{k}(x)=\det\Big(I_{i+j+1}(2\sqrt{x})\Big)_{i,j=0,1,\ldots,k-1}.
\eea
It is known that $\tau_k(x)$ is closely related to the $\tau$-function of a certain $\sigma$-Painlev\'{e} III$'$ equation \cite[(4.20)]{forrester2002application}. Specifically, it was shown in \cite{forrester2006boundary} that
\be
x^{-k^2/2}  \tau_k(x) = (-1)^{k(k-1)/2} \frac{G^2(k+1)}{G(2k+1)} e^x \exp\left( - \int_0^{4x} \frac{\sigma_{\text{III}'}(s) + k^2}{s} ds \right),
\label{tauk-solution-relation}
\ee
where $\sigma_{\text{III}'}(s)$ satisfies the particular $\sigma$-Painlev\'{e} III$'$ equation
\be
\label{painleve equation}
(s \sigma_{\text{III}'}'')^2 + \sigma_{\text{III}'}' (4\sigma_{\text{III}'}'-1) (\sigma_{\text{III}'} -s \sigma_{\text{III}'}') - \frac{k^2}{16} = 0
\ee
with boundary condition $\sigma_{\text{III}'}(s) \sim -k^2+\frac{s}{8} + O(s^2)$ when $s \rightarrow 0$, and $G$ is the Barnes G-function \cite{barnes1900theory}.

Our first main result expresses $f_l(x)$ in terms of derivatives of $\tau_k(x)$.

\begin{thm}\label{theorem2 in note 4}
Let $l\geq 1$ be an integer, then
\bea\label{expression of fl}
f_{l}(x) = \frac{1}{x^l} \sum_{m=0}^l x^m P_{m}(x) \frac{d^m \tau_k(x)}{dx^m} ,
\eea
where $P_{m}(x)=\sum_{j=0}^{l-m} c_{j,m}^{(l)} (k) x^{j}$, and $c_{j,m}^{(l)}(k)$ are polynomials of $k$ of degree at most $3l-2(m+j)$ with coefficients depending on $j, l, m$.
\end{thm}

In the following, we explain how to use properties of solutions of the $\sigma$-Painleve III$'$ equation to compute $F_2(M,k)$. Recall that in the first order derivative case (see \cite[Theorem 1.1]{bailey2019mixed}), 
\bea
F_1(M,k) &:=& \lim_{N\rightarrow \infty} \frac{1}{N^{k^2 + 2M}}
\int_{\U(N)} |Z_A'(1)|^{2M} |Z_A(1)|^{2k-2M} dA_N \nonumber \\
&=& (-1)^{M+ \frac{k(k-1)}{2}} 
\left( \frac{d}{dx}\right)^{2M} \Big( e^{-x/2 } x^{-k^2/2} \tau_k(x) \Big) \Bigg|_{x=0}.
\label{F1Mk}
\eea
So for fixed $M$ with $0\leq M\leq k$, $F_1(M,k)$ can be computed from the linear combination of the first $2M$ coefficients of the Taylor expansion of $x^{-k^2/2}\tau_k(x)$ at $x=0$. Substituting Theorem \ref{theorem2 in note 4} into Proposition \ref{intro:prop1}, one can see that $F_2(M,k)$ may be computed by the linear combination of the first $4M$ coefficients of the Taylor expansion of $x^{-k^2/2}\tau_k(x)$ at $x=0$. The Taylor coefficients of $x^{-k^2/2}\tau_k(x)$ are determined by the Taylor expansion of $\sigma_{\text{III}'}(s)$, and so can be deduced recursively from the differential equation it satisfies.

Our second main result establishes a recursive relation for $f_l(x)$ and so provides a recursive relation for the coefficients $c_{j,m}^{(l)}(k)$ of $P_m(x)$ in Theorem \ref{theorem2 in note 4}.
Before stating the result, we introduce some matrices that will be used throughout this paper.
Let $l\geq 3,k\geq 1$ be integers. 
Let $B^{(l)}=(b_{i,j}^{(l)})_{i,j=1,\ldots,l}$ be an $l\times l$-matrix satisfying 
\bea\label{definition of B}
b_{ij} =\begin{cases}
(-1)^{i+j-1}/j(j+1) & j\geq i; \\
-1/i & j=i-1; \\
0  & j<i-1; \\
(-1)^{i-1}/l & j=l.
\end{cases}
\eea
Let 
$C_1^{(l)}= (c_{i,j}^{(1)})_{\substack{i=1,\ldots,l
\\j=1,\ldots,l-1}}$ be an $l \times (l-1)$-matrix satisfying
\bea\label{definitionofC1}
c_{i,j}^{(1)}=
\begin{cases}\label{definitionofC1}
(-1)^{i+j}\big(\frac{l-1-j}{j+1}+\frac{2k-j+1}{j(j+1)}\big) & \text{if } i\leq j\leq l-1;\\
 -\frac{1}{j+1}\big(j(2k-j)+l-1\big)& \text{if } j=i-1; \\
  0 & \text{if } j<i-1.
\end{cases}
\eea
Let 
$C_2^{(l)}= (c_{i,j}^{(2)})_{\substack{i=1,\ldots,l \\j=1,\ldots,l-2}}$ be an $l\times (l-2)$-matrix satisfying
\bea\label{definitionofC2}
c_{i,j}^{(2)} = 
\begin{cases}
 (-1)^{i+j}  \frac{k+2}{(j+1)(j+2)} & i-1\leq j \leq l-2; \\
\frac{i-k-2}{i} & j=i-2; \\
 0 & j<i-2.
\end{cases}
\eea
Let $C_3^{(l)}= (c_{i,j}^{(3)})_{\substack{i=1,\ldots,l \\j=1,\ldots,l+1}}$ be an $l\times (l+1)$-matrix satisfying
\be\label{definition of C3}
c_{i,j}^{(3)}
=\begin{cases}
(-1)^{j} & i=1,j=1,2; \\
\frac{(-1)^{i+j}}{(j-1)(j-2)} & j>i+1;\\
\frac{j}{j-1} & j=i+1,i\neq 1; \\
0 & j\leq i, i\neq 1.
\end{cases}
\ee

\begin{thm}\label{theorem 1 in note 4}
Let $k\geq 1, l\geq 3$ be integers. Let $f_{l}(x)$ be as given in (\ref{410definition of f_{l}}). Let $B^{(l)}$, $C_1^{(l)}, C_2^{(l)}, C_3^{(l)}$ be as given in (\ref{definition of B}), (\ref{definitionofC1}),  (\ref{definitionofC2}) and (\ref{definition of C3}), respectively. 
Denote
\[
\f_{l}^{(i)} = \begin{pmatrix}
f_{l,1}^{(i)}(x) &
\cdots &
f_{l,l}^{(i)}(x) 
\end{pmatrix}^T, \quad
\hat{\f}_{l}^{(i)} = \begin{pmatrix}
f_{l,2}^{(i)}(x) &
\cdots &
f_{l,l}^{(i)}(x) 
\end{pmatrix}^T.
\]
Then, for $i\geq 0$,
\be\label{expression of fi}
f_{i+1}(x) = f^{(i)}_{2,1}(x) -  f^{(i)}_{2,2}(x) ,
\ee
where $f^{(i)}_{2,1}(x), f^{(i)}_{2,2}(x)$ satisfy the following recursive relation
\bea\label{recursive formula for fjqi}
\f_{l}^{(i)}
&=&
-\, B^{(l)} \Big(\sqrt{x}\frac{d}{d x}-\frac{k^2+l-1-2i}{2\sqrt{x}} \Big)
\begin{pmatrix}
\f_{l-1}^{(i)} \\
0
\end{pmatrix} 
- \frac{1}{\sqrt{x}} C_1^{(l)}
\f_{l-1}^{(i)}
+ C_2^{(l)}
\f_{l-2}^{(i)} +\frac{2i}{\sqrt{x}} C_3^{(l)}
\f_{l+1}^{(i-1)}  \nonumber \\
&&
+ \, B^{(l)} \Big(2i\frac{d}{dx}-\frac{i(k^2+l)+2i(i-1)}{x}\Big)
\begin{pmatrix}
\hat{\f}_{l}^{(i-1)} \\
0
\end{pmatrix}.
\eea
The initial conditions for recursive formula (\ref{recursive formula for fjqi})
are given as follows.
\bea\label{initial values for the recursive formulas}
f_{0}(x)&=&\tau_{k}(x), \nonumber \\
f_{1,1}^{(i)}(x)&=&
\sqrt{x} \frac{d}{dx} f_i - \frac{1}{2\sqrt{x}} k^2 f_i - \frac{i}{\sqrt{x}} f_i, \nonumber \\
f_{2,1}^{(i)}(x)&=&\frac{1}{2}
kf_i   - \frac{i}{\sqrt{x}} (f^{(i-1)}_{3,1}-f^{(i-1)}_{3,2} +f^{(i-1)}_{3,3}) \nonumber  \\
&& +\, 
\frac{1}{2}\Big(x \frac{d^2}{dx^2} f_i - (k^2+2k+2i) \frac{d}{dx} f_i + \frac{(k^2+2i)(k^2+4k+2i+2)}{4x} f_i\Big), \nonumber \\
f_{2,2}^{(i)}(x)&=& -\frac{1}{2}
kf_i   + \frac{i}{\sqrt{x}} (f^{(i-1)}_{3,1}-f^{(i-1)}_{3,2} +f^{(i-1)}_{3,3}) \nonumber \\
&& +\,
\frac{1}{2}\Big(x \frac{d^2}{dx^2} f_i - (k^2-2k+2i) \frac{d}{dx} f_i + \frac{(k^2+2i)(k^2-4k+2i+2)}{4x} f_i\Big). 
\eea
\end{thm}

In (\ref{initial values for the recursive formulas}), when $i=0$, $f^{(i-1)}_{3,1}, f^{(i-1)}_{3,2}, f^{(i-1)}_{3,3}$ are viewed as 0.
{We will explain the meaning of $f_{j,k}^{(i)}$ in Section \ref{proof of main results}}.
From the above recursive formula (\ref{recursive formula for fjqi}), for a fixed $i$ to compute $\f_{l}^{(i)}$ for any $l$ we need the information about 
$\f_{l-1}^{(i)}, \f_{l-2}^{(i)}, \f_{l+1}^{(i-1)}, \f_{l}^{(i-1)}$. 
So, when using the recursive formula (\ref{recursive formula for fjqi}), we first iterate over $i$ (namely, we compute $\f_{l'}^{(i-1)}$ for all $l'\leq l+1$ and store them), 
then we recursively use (\ref{recursive formula for fjqi}) to compute $\f_{l}^{(i)}$ from $\f_{l-1}^{(i)}, \f_{l-2}^{(i)}$ and the stored information of $\f_{l+1}^{(i-1)}, \f_{l}^{(i-1)}$. 
There are only finitely many $l,i$ that need to be considered, so the recursive approach is a linear process with polynomial complexity.

In practice, we need explicit formulae for $f_{1}(x),\ldots,f_{2k}(x)$ as form of (\ref{expression of fl}). 
We now briefly explain how to use (\ref{recursive formula for fjqi}) to obtain these.
Suppose we are in the $i$-th step for some $0\leq i\leq 2k-1$ and already have explicit formulae for $f_{0}(x),\ldots,f_{i}(x)$. Firstly, we use (\ref{initial values for the recursive formulas}) to update the initial values $f_{1,1}^{(i)}(x)$, $f_{2,1}^{(i)}(x)$, $f_{2,2}^{(i)}(x)$. Secondly, we use (\ref{expression of fi}) to calculate $f_{i+1}(x)$.
Thirdly, we use  (\ref{recursive formula for fjqi}) to calculate $\f_{l,1}^{(i)}(x)$ for $3\leq l\leq 2k+1-i$, from which we can compute $f_{1,1}^{(i+1)}(x)$, $f_{2,1}^{(i+1)}(x)$, $f_{2,2}^{(i+1)}(x)$ based on (\ref{initial values for the recursive formulas}). Continuing the above process, we can compute $f_i(x)$ for all $1 \leq i \leq 2k$.

Our third main result is about the structure of $F_2(M,k)$.

\begin{prop}
\label{intro:prop}
For any given integers $k\geq 1$ and any integer $M$ with $0\leq M\leq k$, we have
\be
F_2(M,k) = \frac{G^2(k+1)}{G(2k+1)} R_M(k),
\label{intro:decomposition}
\ee
where $G$ is the Barnes G-function, $R_M(k)$ is a rational function which is analytic when ${\rm Re}(k)>M-1/2$.
\end{prop}

It was demonstrated in \cite{hughes2001characteristic,dehaye2008joint,basor2019representation} that $F_1(M,k)$ also equals $\frac{G^2(k+1)}{G(2k+1)}$ multiplying a rational function. The factor $\frac{G^2(k+1)}{G(2k+1)}$ first appeared in the $2k$-th moment of CUE characteristic polynomials in \cite{keating2000random1}. Namely, it equals $F_1(0,k)$ and $F_2(0,k)$. Moreover, from our results in Section \ref{section:Generalization}, the joint moments
\[
\lim_{N\rightarrow \infty} \frac{1}{N^{k^2+2Mn}} \int_{\U(N)} |Z_A^{(n)}(1)|^{2M} |Z_A(1)|^{2k-2M} dA_N,
\]
for any $n\geq 1$, all have a similar structure to (\ref{intro:decomposition}).

We list some examples as an illustration of Proposition \ref{intro:prop}:
\beas
R_1(k) &=& \frac{1}{16(2k-1)(2k+3)}, \\
R_2(k) &=& \frac{16k^4+64k^3+40k^2-32k-99}{256(2k-3)(2k-1)^2(2k+1)^2(2k+3)(2k+5)(2k+7)}, \\
R_3(k) &=& (512k^9+5376k^8+14336k^7-13824k^6-102080k^5-66912k^4+188608k^3 \\
&& -239232k^2-225318k+463545)
\Big/ 4096(2k-5)(2k-3)^2(2k-1)^3(2k+1)^3 \\
&& (2k+3)^2(2k+5)(2k+7)(2k+9)(2k+11) , \\
R_4(k) &=& (4096k^{12} + 81920k^{11} +509952 k^{10} +233472 k^9 -8833280k^8-25065472k^7 \\
&& +30041856k^6+155091456k^5-18354704k^4+2414144k^3-800470200k^2 \\
&& -1962813360k+6148319625) \Big/ 
65536 (2k-7)(2k-5)^2(2k-3)^2(2k-1)^3 \\
&& (2k+1)^3 (2k+3)^2(2k+5)^2 (2k+7)(2k+9)(2k+11)(2k+13)(2k+15).
\eeas

As an application of our method to prove Theorem \ref{theorem 1 in note 4}, we provide a recursive relation for the coefficients of the Taylor expansion of $x^{-\frac{k^2}{2}}\tau_k(x)$ at $x=0$.

\begin{thm}
\label{intro:thm5}
For any given $k\geq 1$, assume that 
\be
\label{expression of tauk}
\tau_k(x) = (-1)^{\frac{k(k-1)}{2}} \frac{G^2(k+1)}{G(2k+1)} x^{\frac{k^2}{2}} \sum_{j=0}^\infty a_j x^j,
\ee
then
for any $i\geq 1$,
\be
\label{thm:recursive formula for ai-eq1}
a_i D_{k+1,0}(i) = -\sum_{q=1}^{\min(i, \lfloor \frac{k+1}{2}\rfloor)} a_{i-q} D_{k+1,q}(i),
\ee
where $a_0=1$, $D_{k+1,0}(i)\neq 0$,
and $D_{k+1,q}(i)$ can be computed via the following recursive formulae: for
$l\geq 3$
\beas
D_{l,q}(n) &=& \frac{n+(l-1)(2k-l+1)}{l} D_{l-1,q}(n) + \frac{l-k-2}{l} D_{l-2,q-1}(n-1), 
\\
D_{2,0}(n) &=& \frac{n(2k-1+n)}{2}, \quad 
D_{2,1}(n) = -\frac{k}{2}, \quad D_{1,0}(n) = n.
\eeas
Moreover, $D_{l,q}(n)=0$ when $q> \lfloor l/2\rfloor$.
\end{thm}

From the connection (\ref{tauk-solution-relation}) between $\tau_k(x)$ and the solution $\sigma_{\text{III}'}(s)$ of the $\sigma$-Painlev\'{e} III$'$equation, it is standard to use the Taylor series of $\sigma_{\text{III}'}(s)$ to obtain the coefficients of the Taylor expansion of $\tau_k(x)$ at $x=0$, e.g., see \cite{forrester2006boundary, basor2019representation}. However, the differential equation (\ref{painleve equation}) does not uniquely determine $\sigma_{\text{III}'}(s)$, even when one is provided with boundary conditions like  $\sigma_{\text{III}'}(0)=-k^2, \sigma_{\text{III}'}'(0)=1/8$. For any given $k\geq 1$, the recursive relations obtained from the differential equation (\ref{painleve equation}) can determine the first $2k$ Taylor coefficients $a_1,\ldots,a_{2k}$, but not the $(2k+1)$-th coefficient (see (\ref{eta}) for more details). Different values of this coefficient determine different solutions, i.e., the differential equation (\ref{painleve equation}) has a one-parameter solution. For the moments of the first-order derivative of characteristic polynomials from CUE (i.e., (\ref{F1Mk}) with $M=k$), the first $2k$ Taylor coefficients are enough. 
According to our result, the first $4k$ Taylor coefficients are required for moments of second order derivative (i.e., (\ref{expression of bkm}) with $M=k$) (see Section \ref{section:Truncated case}). Theorem \ref{intro:thm5} provides a recursive relation for finding all the Taylor coefficients of $\tau_k(x)$. In this process, we used a different differential equation that $\tau_k(x)$ satisfies, rather than its connection to the $\sigma$-Painlev\'{e} III$'$ equation. This new differential equation is $f_{k+1,k+1}^{(0)} = 0$, where $f_{k+1,k+1}^{(0)}$ is given in Theorem \ref{theorem 1 in note 4}. Generally, for any $l \geq q \geq k+1$, using the recursive relation in Theorem \ref{theorem 1 in note 4} with $i=0$, we can prove that $f_{l,q}^{(0)}$ has the form 
$ x^{-l/2} \sum_{s=0}^{l} \frac{d^s \tau_k}{dx^s} x^s \sum_{j=0}^{\lfloor\frac{l-s}{2}\rfloor} a_{j,s,l,q} (k) x^j,$. This can be proved to be 0 (see Remark \ref{natural extension}). So this gives a differential equation of order $l$ that $\tau_k(x)$ satisfies.

% To better illustrate our main idea of establishing the connection, we first consider some simple special cases. 
% Let
% $A=\det\begin{pmatrix}
% a_1 & a_2 \\
% a_2 & a_3
% \end{pmatrix}$
% be a Wronksian determinant, i.e., $a_{n+1}=a_{n}'$. When shifting the second column by 1, we obtain determinant $B=\det\begin{pmatrix}
% a_1 & a_3 \\
% a_2 & a_4
% \end{pmatrix}.$
% The connection of $A,B$ is described by the following equation
% \[
% B = \begin{pmatrix}
% a_2 & a_3 \\
% a_3 & a_4
% \end{pmatrix} \cdot \begin{pmatrix}
% F_{11} & F_{12} \\
% F_{21} & F_{22}
% \end{pmatrix} 
% =\begin{pmatrix}
% a_1' & a_2' \\
% a_2' & a_3'
% \end{pmatrix} \cdot \begin{pmatrix}
% F_{11} & F_{12} \\
% F_{21} & F_{22}
% \end{pmatrix} = A',
% \]
\subsection{Summary of ideas and methods}
In this section, we briefly explain our main ideas and methods to prove Theorems \ref{theorem2 in note 4} and \ref{theorem 1 in note 4}.

From the expression (\ref{410definition of f_{l}}) for $f_l(x)$, our starting point is to examine the connection between $\det \Big(I_{i+j+1+2l_{j+1}}(2\sqrt{x})\Big)_{i,j=0,\ldots,k-1}$ and $\tau_k(x)=\det \Big(I_{i+j+1}(2\sqrt{x})\Big)_{i,j=0,\ldots,k-1}$. The latter is known to be closely related to a solution of the $\sigma$-Painlev\'{e} III$'$ equation. The former has a similar structure to $\tau_k(x)$, except that the columns are shifted by integers. 
Through some column permutations, we can rewrite $\det \Big(I_{i+j+1+2l_{j+1}}(2\sqrt{x})\Big)_{i,j=0,\ldots,k-1}$ as $\det \Big(I_{i+j+1+t_{k-j}}(2\sqrt{x})\Big)_{i,j=0,\ldots,k-1}$ such that $t_1\geq \cdots\geq t_s > t_{s+1}=\cdots =t_k= 0$. We will view this sequence of integers as a Young diagram $Y=(t_1,\ldots,t_s)$, denote the determinant as $\tau_{k,Y}$, and refer to it as a determinant shifted by the Young diagram $Y$. The length of $Y$ is defined as $|Y|:=t_1+\cdots+t_s$. Our goal is to explore connections between $\tau_{k,Y}$ and $\tau_k$. 

To better illustrate the main idea underpinning the connection, we first consider some simple special cases. If $Y=(1)$, $\tau_{2,Y} = \begin{pmatrix}
I_1(2\sqrt{x}) & I_3(2\sqrt{x})  \\
I_2(2\sqrt{x}) & I_4(2\sqrt{x})
\end{pmatrix}$. Then from the relation $I_{\beta+1}(2\sqrt{x}) = \sqrt{x} I'_\beta(2\sqrt{x}) - \frac{\beta}{2\sqrt{x}} I_\beta(2\sqrt{x})$, we obtain the following relation $\tau_{2,Y} = \sqrt{x} \tau_2' - \frac{2}{\sqrt{x}}  \tau_2$. For general $k$, one can prove similarly that $\tau_{k,Y} = \sqrt{x} \tau_k' - \frac{k^2}{2\sqrt{x}}  \tau_k$.

From the above example, it is possible to express $\tau_{k,Y}$ in terms of derivatives of $\tau_k$.
To this end, we consider determinants of general $k\times k$ Hankel matrices whose columns are shifted by Young diagrams $Y=(t_1,\ldots,t_s)$, and denote it as $H_{k,Y}$. The corresponding matrix is denoted as $M_{k,Y}=(a_{i+j+t_{k-j}+1})_{i,j=0,\ldots,k-1}$. When $Y=\emptyset$ is an empty Young diagram, we denote $H_{k,Y}$ as $H_k$ for simplicity.
We introduce a translation operator $T_h$ to study properties of $H_{k,Y}$. $T_h H_{k,Y}$ is the inner product between $M_{k,T_hY}$ and the cofactor matrix of $M_{k,Y}$, where $T_hY=(t_1+h,\ldots,t_k+h)$. We want to explore the connections between $H_{k,Y_1}$ and $T_hH_{k,Y_2}$. From the definition, $T_hH_{k,Y_2}$ is a linear combination of $H_{k,Y_1}$ for some Young diagrams $Y_1$ of length $|Y_2|+h$. Conversely, given a $H_{k,Y_1}$, we want to know if it can be written as a linear combination of $T_hH_{k,Y_2}$ for some $Y_2$.

An obvious approach towards this goal is to seek an invertible linear system of equations from the linear expression of $T_hH_{k,Y_2}$ in terms of $H_{k,Y_1}$. The main difficulty in achieving this is that the number of Young diagrams of a fixed length is exponentially large, so it is hard in practice to find such a linear system.

Therefore, we develop an alternative approach, which starts from hook Young diagrams, i.e., diagrams of the form $Y_{l,j}=(l-j+1,1,\ldots,1)$ with $(j-1)$ 1's. For a hook diagram $Y_{i,j}$, $T_hH_{k,Y_{l,j}}$ is a linear combination of some $H_{k,Y}$, where some $Y$ are hooked and others are not. We establish a method to eliminate non-hook Young diagrams so that a linear combination of $T_hH_{k,Y_{l,j}}$ is a linear combination of some $H_{k,Y_{l',j'}}$ (see Theorem \ref{thm1onHankel deteminants}). Consequently, we obtain a linear system among hook diagrams. Fortunately, this linear system has a nice structure whose coefficient matrix is lower Hessenberg with an explicit inversion (see Theorem \ref{gerenel linear systems}). This means that, somewhat unexpectedly, hook diagrams are sufficient for our purposes. In addition, for general Young diagrams, we also show how to write $H_{k,Y_1}$ as a linear combination of $T_hH_{k,Y_2}$ by generalizing the results and methods for hook diagrams (see the end of Section \ref{rfHsy}).

When the matrix elements $a_n$ of $M_{k,Y}$ satisfy some recursive relations, we establish recursive relations for $H_{k,Y_{l,j}}$. For example, consider when $a_n$ are Bessel functions, which have a recursive relation
$I_{\beta+2}(2\sqrt{x})=I_{\beta}(2\sqrt{x})-\frac{\beta +1}{\sqrt{x}}I_{\beta+1}(2\sqrt{x})$. Let $Y_{l-h,j-h}=(t_1,\ldots,t_{j-h})$ be a hook diagram. We set $t_{j-h+1}=\cdots=t_k=0$.
Substituting the relation for Bessel functions into $T_h H_{k,Y_{l-h,j-h}}$, we then obtain from the definition that
\[
T_h H_{k,Y_{l-h,j-h}}
=(I_{i+j+h-2+t_{k-j}}) \cdot (F_{ij})
-\frac{1}{\sqrt{x}} ((i+j+h-1+t_{k-j})I_{i+j+h-1+t_{k-j}}) \cdot (F_{ij}),
\]
where $F_{ij}$ is the $(i,j)$-th cofactor of $M_{k,Y_{l-h,j-h}}$.
The first term is $T_{h-2} H_{k,Y_{l-h,j-h}}$, which can be handled similarly using the approach described above. Regarding the second term, we introduce a weighted translation operator $S_{h-1}$ so that it is $\frac{1}{\sqrt{x}} S_{h-1} H_{k,Y_{l-h,j-h}}$. We further decompose this term into two parts
$\frac{1}{\sqrt{x}} ((j+h-1+t_{k-j})I_{i+j+h-1+t_{k-j}}) \cdot (F_{ij})
+ \frac{1}{\sqrt{x}} (i I_{i+j+h-1+t_{k-j}}) \cdot (F_{ij})$, where the weights are put on the columns and rows respectively. The first part is handled using a similar approach to the one described above, by a Laplace expansion on columns, see Lemma \ref{lemmaonS1}. The handling process for the second part is much more complicated. Roughly speaking, we extend $H_{k,Y}$ to $H_{k,\{X,Y\}}$ with two Young diagrams $X,Y$, where rows are shifted by $X$ and columns are shifted by $Y$. We then show that there is an invertible linear system between $\{H_{k,\{X,Y_{l,j}\}}:j=1,2,\ldots,l \}$ and $\{\sum_{h=1}^{j-1} (-1)^h T_hH_{k,\{X,Y_{l-h,j-h}\}}:j=2,\ldots,l \} \cup \{T_lH_{k,X}\}$, see Theorem \ref{thm1onHankel deteminants}. 
Repeatedly using this linear property with some different and appropriate choices of $X$, we can express the second term as a linear combination of determinants of Hankel matrices shifted by hook diagrams, see Lemma \ref{lemmaonS2}. In the end, we obtain a recursive relation for $H_{k,Y_{l,j}}$ involving $H_{k,Y_{l-1,j_1}}$ and $H_{k,Y_{l-2,j_2}}$, see Proposition \ref{recursive formula for tau0}.

From the results described above, we obtain a general description of the basic structure of $\tau_{k,Y}(x)$:
\be
\label{intro:tauk}
\tau_{k,Y} = x^{-m/2} \sum_{s=0}^{m} \frac{d^s \tau_k}{dx^s} x^s \sum_{j=0}^{\lfloor\frac{m-s}{2}\rfloor} a_{j,s, Y} (k) x^j,
\ee
where $m=|Y|$, and $a_{j,s,Y} (k)$ are certain polynomials of $k$ with degree at most $2(m-s-j)$ and with coefficients depending on $j, s, Y$.

Recall that our overarching goal is to build recursive relations for $f_l(x)$, defined in (\ref{410definition of f_{l}}), which is a combinatorial sum of $\tau_{k,Y}$. For example, $f_1(x) = \tau_{k,(2)} - \tau_{k,(1,1)}$ and $f_2(x) = \tau_{k,(4)} - \tau_{k,(1,3)}-\tau_{k,(1,1,2)} + 2 \tau_{k,(2,2)} +\tau_{k,(1,1,1,1)}$. From (\ref{intro:tauk}), we can see that the highest differential order of $\tau_{k}$ appearing in the expression of $\tau_{k,Y}$ is $m$, the length of $Y$. However, one sees that the highest differential order of $\tau_k$ in $f_1$ is 1 and in $f_2$ is $2$; that is, they are actually half the lengths of Young diagrams in their decomposition (see Remark \ref{newremark} for an explanation for the general case). This implies that there are some cancellations among the $\tau_{k,Y}$ appearing in the expression for $f_l(x)$. So, in order to obtain a more effective recursive relation for $f_l$, we cannot use the recursive relations satisfied by $\tau_{k,Y}$ without further consideration of this point. 
Our way around this is firstly to express $f_l$ as a partial derivative of a new Hankel matrix, 
\[
f_l(x) = \frac{\partial^l}{\partial t^l} G_k(x,t) \Big|_{t=0},
\]
where $G_k(x,t)=\det(g_{i+j+1}(x,t))_{i,j=0,\ldots,k-1}$ and $g_{\beta}(x,t) = \sum_{n=0}^\infty \frac{t^n}{n!} I_{2n+\beta}(2\sqrt{x})$. Notice that $G_k(x,0)=\tau_k(x)$. We next show that there is a connection between the partial derivative of $G_k$ with respect to $t$ and the translation operator $T_2$ acting on $G_k$, see Proposition \ref{derivative and translations}. This is further connected to the Hankel determinants shifted by Young diagrams. Namely, 
\[
\frac{\partial}{\partial t} G_k(x,t) = T_2 G_k(x,t) = G_{k,Y_{2,1}} - G_{k,Y_{2,2}}.
\]
From this, we can see that to compute $f_l(x)$, it suffices to build recursive relations for $ (\frac{\partial^{l-1}}{\partial t^{l-1}}G_{k,Y_{2,1}}|_{t=0}$,  $\frac{\partial^{l-1}}{\partial t^{l-1}}G_{k,Y_{2,2}}|_{t=0}) $.
Substituting the recursive relation $g_{\beta+2}=g_{\beta} - \frac{\beta+1}{\sqrt{x}} g_{\beta+1} - \frac{2t}{\sqrt{x}}g_{\beta+3}$ into $T_2 G_k(x,t)$, we obtain 
\beas
G_{k,Y_{2,1}} - G_{k,Y_{2,2}} = \frac{\partial}{\partial t} G_k(x,t) = k G_k(x,t) - \frac{2k}{\sqrt{x}} G_{k,Y_{1,1}} - \frac{2t}{\sqrt{x}} (G_{k,Y_{3,1}} - G_{k,Y_{3,2}} + G_{k,Y_{3,3}}).
\eeas
Taking $(l-1)$-th partial derivative with respect to $t$ on both sides at $t=0$ yields
\beas
&& \frac{\partial^{l-1}}{\partial t^{l-1}} (G_{k,Y_{2,1}} - G_{k,Y_{2,2}}) \Big|_{t=0} \\
&=& k  \frac{\partial^{l-1}}{\partial t^{l-1}} 
G_k\Big|_{t=0} - \frac{2k}{\sqrt{x}} \frac{\partial^{l-1}}{\partial t^{l-1}} 
G_{k,Y_{1,1}}\Big|_{t=0} - \frac{2(l-1)}{\sqrt{x}} \frac{\partial^{l-2}}{\partial t^{l-2}} (G_{k,Y_{3,1}} - G_{k,Y_{3,2}} + G_{k,Y_{3,3}}) \Big|_{t=0}.
\eeas
The above is our general method for building a recursive formula for $f_l$. 
For the first two terms, the lengths of hook diagrams are decreasing which shows the existence of a recursive relation. For the third term, the length is increasing, but fortunately the order of partial derivative with respect to $t$ is decreasing. About this term, using the methods we established previously for Hankel determinants 
shifted by hook diagrams (where we choose $a_n=g_n$ in $M_{k,Y}$), we can establish similar recursive relations with the property that apart from some ``good" terms like the first two described above, the length of Young diagrams is increasing, while the order of partial derivative with respect to $t$ is decreasing. We end up with a zero-th order derivative of some determinants shifted by hook diagrams, which can be handled using the recursive relations we established for Hankel determinants of Bessel functions shifted by hook diagrams mentioned previously.

The above is an overview of the key ideas involved in our proof of Theorem \ref{theorem 1 in note 4}. To conclude, we explain our main idea for proving Theorem \ref{theorem2 in note 4}.
From the recursive relations in Theorem \ref{theorem 1 in note 4}, we obtain basic descriptions of $f_{j,q}^{(i)}(x) :=\frac{\partial^{i}}{\partial t^{i}}G_{k,Y_{j,q}}|_{t=0}$. These have similar structures to (\ref{intro:tauk}), which are polynomial combinations of derivatives of $\tau_k$. For $\frac{\partial^{i}}{\partial t^{i}}G_{k,Y_{j,q}}|_{t=0}$, the highest order of the derivative of $\tau_k$ is $i+j$. In particular, for $\frac{\partial^{l-1}}{\partial t^{l-1}}G_{k,Y_{2,1}}|_{t=0}, \frac{\partial^{l-1}}{\partial t^{l-1}}G_{k,Y_{2,2}}|_{t=0}$, the highest differential order is $l+1$.
From the analysis below (\ref{intro:tauk}), the highest differential order $\tau_k$ in $f_l$ should be $l$. So it is not straightforward to obtain a precise description of the structure of $f_l$ from the expression
$f_l(x) = (\frac{\partial^{l-1}}{\partial t^{l-1}}G_{k,Y_{2,1}}- \frac{\partial^{l-1}}{\partial t^{l-1}}G_{k,Y_{2,2}})|_{t=0}$, which required lots of accurate computations about the polynomial coefficients in $\frac{\partial^{l-1}}{\partial t^{l-1}}G_{k,Y_{2,1}}|_{t=0}, \frac{\partial^{l-1}}{\partial t^{l-1}}G_{k,Y_{2,2}}|_{t=0}$. Our idea is to expand $\frac{\partial}{\partial t} G_k(x,t)$ in powers of $t$ up to degree $l$ repeatedly using the recursive formula of $g_\beta$. This yields
\[
\frac{\partial G_{k}}{\partial t}(x,t)=\sum_{i=0}^{l-1}(-1)^{i+1}
(\frac{1}{\sqrt{x}}S_{i+1}G_{k} -T_{i}G_{k})
(\frac{2t}{\sqrt{x}})^{i}+(-1)^{l}T_{l+2}G_{k}(\frac{2t}{\sqrt{x}})^{l}.
\]
Taking $(l-1)$-th partial derivative with respect to $t$ at $t=0$, we obtain
\bea
f_{l}(x) &=& kf_{l-1}(x)-\frac{2k}{\sqrt{x}}f_{1,1}^{(l-1)}(x)-\sum_{j=1}^{l-1}(-1)^{j+1}j!\binom{l-1}{j}(\frac{2}{\sqrt{x}})^{j}\sum_{q=1}^{j}(-1)^{q-1} f_{j,q}^{(l-1-j)}(x) \nonumber\\
&& + \, \frac{1}{2}\sum_{j=1}^{l-1}(-1)^{j+1}j!\binom{l-1}{j}(\frac{2}{\sqrt{x}})^{j+1}\Big(\sum_{q=1}^{j+1}(-1)^{q-1}(2k-2q+j+2)f_{j+1,q}^{(l-1-j)}(x)\Big).
\label{intro:fl}
\eea
As discussed above, the highest differential order of $\tau_k(x)$ in the structure expression of $f_{j,q}^{(i)}(x)$ is $i+j$, so one can see that the highest differential order of $\tau_k(x)$ in the right-hand side of (\ref{intro:fl}) is $l$. As a result, we can deduce the correct highest differential order of $\tau_k(x)$ in $f_l$ without any complicated calculations requiring the intricate cancellations mentioned previously.  We mainly use (\ref{intro:fl}) and suitable modifications for the initial conditions of the recursive relation obtained in Theorem \ref{theorem 1 in note 4} to prove Theorem \ref{theorem2 in note 4}.

\subsection{Notation}

For a matrix $M$, we use $M^T$ to denote its transpose. 
The {\em standard inner product} of two matrices $A=(a_{ij})_{k\times k}, B=(b_{ij})_{k\times k}$ is defined as $A \cdot B = \sum_{i,j=1}^k a_{ij} b_{ij} = {\rm Tr}(A^T B)$.

The {\em Barnes G-function} is formally defined as 
\[
G(1+z)=(2\pi)^{z/2} \exp\left(- \frac{z+z^2(1+\gamma)}{2} \right) \, \prod_{k=1}^\infty \left\{ \left(1+\frac{z}{k}\right)^k \exp\left(\frac{z^2}{2k}-z\right) \right\}.
\]
In particular, $G(0)=0$, $G(1)=1$. For $n\geq 2$ a positive integer, we have
$G(n) = \prod_{j=0}^{n-2} j!$.

\subsection{Acknowledgements}

This research was carried while F.W. was visiting the University of Oxford. F.W. would like to thank the Mathematical Institute, University of Oxford for generous hospitality. F.W. was supported by the fellowship of China Postdoctoral Science Foundation 2020M670273 and Jinxin Xue's grant NSFC (project No.20221300605).  JPK is pleased to acknowledge support from ERC Advanced Grant 740900 (LogCorRM).
The authors would like to thank Alexander Its for his valuable comments.

\section{Determinants of Hankel matrices shifted by Young diagrams}
\label{section:Determinants of Hankel matrices}

In this section, we introduce some notion and prove some preliminary results that will be used later in the paper.

A Young diagram is a finite collection of boxes arranged in left-justified rows, with the row lengths in non-increasing order. The total number of boxes is called the {\em length}. Listing the number of boxes in each row gives a {\em partition} of the length. 
A Young diagram uniquely corresponds to a partition. In this paper, we will write Young diagrams in the form of partitions. Namely, let $t_1,\ldots,t_s$ be  positive integers such that $t_{1}\geq t_{2}\geq \cdots \geq t_{s}\geq 1$, then $Y=(t_1,t_2,\ldots,t_s)$ defines a Young diagram. An empty Young diagram is denoted as $\emptyset$. A Young diagram of the form $(i-j+1,1,\ldots,1)$, where the number of 1's is $j-1$, is called a {\em hook diagram}. Hook diagrams will play an important role in our calculations.  They will be denoted $Y_{i,j}$.

Let $\{a_{n}\}_{n\in \mathbb{C}}$ be a sequence of complex numbers, and $\beta_0,\ldots,\beta_{k-1}$ be $k$ distinct real numbers. Let $Y=(t_1,t_2,\ldots,t_s)$ be a Young diagram with $s\leq k$. We define
\be
\label{def:M-k}
M_k(\beta_0,\ldots,\beta_{k-1};Y)=
\begin{pmatrix}
a_{\beta_0} & a_{\beta_0+1} & \cdots & a_{\beta_0+k-s+t_s} & \cdots & a_{\beta_0+k-1+t_1} \\
a_{\beta_1} & a_{\beta_1+1} & \cdots & a_{\beta_1+k-s+t_s} & \cdots & a_{\beta_1+k-1+t_1} \\
\vdots & \vdots &  \vdots & \vdots & \vdots & \vdots \\
a_{\beta_{k-1}} & a_{\beta_{k-1}+1} & \cdots & a_{\beta_{k-1}+k-s+t_s} & \cdots & a_{\beta_{k-1}+k-1+t_1}\\
\end{pmatrix}.
\ee
In the above, the first $k-s$ columns are $(a_{\beta_0+i}, a_{\beta_1+i}, \cdots, a_{\beta_{k-1}+i})^T,i=0,\ldots,k-s-1$ and the last $s$ columns are $(a_{\beta_0+k-j+t_j}, a_{\beta_1+k-j+t_j}, \cdots, a_{\beta_{k-1}+k-j+t_j})^T,j=s,\ldots,1$.
A related matrix is the following, where the entries are endowed with weights equaling the sub-indices
\bea
\label{def:tilde-M-k}
&& \widetilde{M}_k(\beta_0,\ldots,\beta_{k-1};Y) \nonumber \\
&=&
\begin{pmatrix}
\beta_0a_{\beta_0}  & \cdots &(\beta_0+k-s+t_s) a_{\beta_0+k-s+t_s} & \cdots & (\beta_0+k-1+t_1)a_{\beta_0+k-1+t_1} \\
\beta_1g_{\beta_1}  & \cdots & (\beta_1+k-s+t_s)a_{\beta_1+k-s+t_s} & \cdots & (\beta_1+k-1+t_1)a_{\beta_1+k-1+t_1} \\
\vdots  &  \vdots & \vdots & \vdots & \vdots \\
\beta_{k-1}a_{\beta_{k-1}}  & \cdots & (\beta_{k-1}+k-s+t_s) a_{\beta_{k-1}+k-s+t_s} & \cdots & (\beta_{k-1}+k-1+t_1) a_{\beta_{k-1}+k-1+t_1}  \\
\end{pmatrix}. \nonumber \\
&& 
\eea
We also define
\be\label{definition of Dk}
D_k(\beta_0,\ldots,\beta_{k-1};Y)=
\det M_k(\beta_0,\ldots,\beta_{k-1};Y).
\ee
In this paper, as a convention we set 
\be \label{previous convention}
M_k,\widetilde{M}_k, D_k=0, \quad \text{if } s>k.
\ee
The reason for this will be explained in Remark \ref{natural extension}.

For $\alpha \in \mathbb{C}$,
we define Hankel determinant 
\bea\label{Hk}
H_{k}=\det(a_{\alpha+i+j})_{i,j=0,\ldots,k-1}.
\eea
This coincides with $D_{k}(\alpha,\alpha+1,\ldots,\alpha+k-1;\emptyset)$. 

\begin{dfn}\label{definition of shifted Hankel}
Let $\alpha \in \mathbb{C}$. Let $X=(l_{1},\ldots,l_{h})$ and $Y$ be Young diagrams. Set $(\beta_{0},\ldots,\beta_{k-1})=(\alpha,\alpha+1,\ldots,\alpha+k-h-1,\alpha+k-h+l_{h},\ldots,\alpha+k-1+l_{1})$.
We define the determinant of the Hankel matrix whose rows are shifted by Young diagram $X$ and columns are shifted by Young diagram $Y$ by
\be
\label{def of H-k-X-Y}
H_{k,\{X,Y\}}:=D_{k}(\beta_{0},\ldots,\beta_{k-1};Y).
\ee
When $X=Y=\emptyset$, $H_{k,\{\emptyset;\emptyset\}}$ is $H_{k}$ as defined in (\ref{Hk}). When $X=\emptyset$, we simply write $H_{k,\{\emptyset;Y\}}$ as $H_{k,Y}$ for short.
\end{dfn}

\begin{rem}{\rm
The definitions of $M_{k}$, $\widetilde{M}_{k}$, $D_{k}$, $H_{k}$ and $H_{k,\{X,Y\}}$ depend on the set $\{a_{\alpha}:\alpha \in \mathbb{C}\}$. In the following, without specific indication, our results hold for any general $\{a_{\alpha}:\alpha \in \mathbb{C}\}$.
}\end{rem}

\begin{rem}{\rm 
The quantities $D_k(\beta_{0},\beta_{1},\ldots,\beta_{k-1}; Y), H_{k,\{X;Y\}}$ are also well-defined when $X,Y$ are not non-decreasing sequences: when there are no two common columns, then via some column permutations, up to a $\pm1$ sign they can be changed into equivalent quantities $D_k(\beta_{0},\beta_{1},\ldots,\beta_{k-1}; Y')$, $H_{k,\{X';Y'\}}$ such that $X', Y'$ are non-decreasing (see Proposition \ref{prop2} below).
}\end{rem}

\begin{defn}[Operators $T_h$ and $S_h$]
\label{defn of translation operator}
Let $k\geq 1$ and $h\geq 0$ be integers. Let $Y=(t_{1},\ldots,t_{s})$ be a Young diagram with $s\leq k$. Define $T_{h}Y$ as the following Young diagram 
\[
T_{h}Y := (t_{1}+h,\ldots,t_{s}+h, h,\ldots,h) \in \mathbb{N}^k.
\]
Define $T_{h}D_k(\beta_0,\ldots,\beta_{k-1};Y)$ to be the standard inner product between $M_k(\beta_0,\ldots,\beta_{k-1};T_{h}Y)$ and the cofactor matrix of $M_k(\beta_0,\ldots,\beta_{k-1};Y)$, and define $S_{h}D_k(\beta_0,\ldots,\beta_{k-1};Y)$ to be the standard inner product between $\widetilde{M}_k(\beta_0,\ldots,\beta_{k-1};T_{h}Y)$ and the cofactor matrix of $M_k(\beta_0,\ldots,\beta_{k-1};Y)$.
\end{defn}

We remark that the operators $T_h, S_h$ are essentially defined on Young diagrams. Namely, if $Y_1=Y_2$ are two equal Young diagrams, then $T_hD_k(\beta_0,\ldots,\beta_{k-1};Y_1) = T_h D_k(\beta_0,\ldots,\beta_{k-1};Y_2)$.

It is easy to see that
$H_{k,\{X;T_{l}Y\}}=H_{k,\{T_{l}X;Y\}}$.
We can also linearly extend the operators $T_{h}$ and $S_{h}$ by
\beas
T_{h}(H_{k,\{X_{1};Y_{1}\}}+H_{k,\{X_{2};Y_{2}\}}) &:=& T_{h}H_{k,\{X_{1};Y_{1}\}}+T_{h}H_{k,\{X_{2};Y_{2}\}}, \\
S_{h}(H_{k,\{X_{1};Y_{1}\}}+H_{k,\{X_{2};Y_{2}\}}) &:=& S_{h}H_{k,\{X_{1};Y_{1}\}}+S_{h}H_{k,\{X_{2};Y_{2}\}}.
\eeas
In the following, we list more basic properties of $D_{k}$ and $H_{k,\{X,Y\}}$.
By definition, we have the following proposition.

\begin{prop}\label{prop1}
For integer $l\geq 0$, 
\[T_{l}D_k(\beta_0,\ldots,\beta_{k-1};Y)=\sum_{i=0}^{k-1}D_k(\beta_0,\ldots,\beta_{i-1},\beta_{i}+l,\beta_{i+1},\ldots,\beta_{k-1};Y).\]
\end{prop}

\begin{prop}
\label{prop2}
Let $\alpha \in \mathbb C$, $h\geq 0$ be an integer, and let $l_{1},\ldots,l_{k}$ be integers with $l_{1}\geq l_{2}\geq \dots \geq l_{k}\geq 0$.
Suppose $\beta_i = \alpha+i+l_{k-i}$ for $i=0,1,\ldots,k-1$. 
If any two of $\{\beta_{0}, \cdots, \beta_{i-1}, \beta_{i}+h, \cdots, \beta_{k-1}\}$ are distinct, we then can reformulate it to be a new set $\{\beta_0, \cdots,\beta_{i-1}, \tilde{\beta}_i, \cdots, \tilde{\beta}_{k-1}\}$ with $\tilde{\beta}_{j} = \alpha + j + \tilde{l}_{k-j}$ for $j=i,i+1,\ldots,k-1$, where
\be
\tilde{l}_{1} \geq \tilde{l}_{2} \geq \cdots \geq \tilde{l}_{k-i}\geq 1,
\quad \sum_{j=1}^{k-i} \tilde{l}_{j} = h+ \sum_{j=1}^{k-i} l_{j}.
\ee
\end{prop}

\begin{proof}
Note that $\beta_i+h= \alpha+i+l_{k-i}+h$. So if $l_{k-i}+h\leq l_{k-i-1}$, then we do not need to rewrite $\beta_i$. We now assume that $l_{k-i}+h\geq l_{k-i-1} + 1$. Note that $\beta_{i}+h\neq\beta_{i+1}$, so $l_{k-i}+h>l_{k-i-1} + 1$, we set $\tilde{\beta}_i=\alpha+i+l_{k-i-1} + 1$ and $\tilde{\beta}_{i+1}=\alpha+i + 1+l_{k-i}-1+h$. 
Namely, $\tilde{l}_{i}=l_{k-i-1}+1$ and $\tilde{l}_{i+1}=l_{k-i}-1+h$. 
We continue this process by considering if $l_{k-i} - 1+h \leq l_{k-i-2}$ or not, and set appropriate $\tilde{\beta}_{i+1}, \tilde{\beta}_{i+2}$.  This process will terminate after a finite number of steps, since $k$ is a finite number. 
\end{proof}

\begin{prop}
\label{prop4}
Let $Y=(1)$ be a Young diagram with 1 box, then
\be
\label{eq2}
D_k(\beta_0,\ldots,\beta_{k-1};Y) = \sum_{i=0}^{k-1}
D_k(\beta_0,\ldots,\beta_{i-1},\beta_i+1, \beta_{i+1}, \ldots,\beta_{k-1};\emptyset).
\ee
\end{prop}

\begin{proof}
Let $F_{i,j}$ be the $(i,j)$-cofactor of $M_k(\beta_0,\ldots,\beta_{k-1};\emptyset)$. We consider the inner product between $M_{k}(\beta_{0},\ldots,\beta_{k-1};T_{1}\emptyset)$ and $(F_{i,j})_{i,j=1,\ldots,k}$. By the Laplace expansion along the $j$-th column, we obtain the left-hand side of (\ref{eq2}). Again, by the Laplace expansion along the $i$-th row, we obtain the right-hand side of (\ref{eq2}).
\end{proof}

Recall that a {\em hook diagram} of length $l$ is a Young diagram of the form $Y_{l,j}:=(l-j+1,{1,1,\ldots,1})$ with $(j-1)$ 1's. The following result follows from Propositions \ref{prop1} and \ref{prop2}.

\begin{prop}
\label{513prop5}
Let $s\geq 0$ and $Y=(t_1,\ldots,t_s)$ be a Young diagram. Let $l,k\geq 1, 1\leq j\leq l$. Let $Y_{l,j}$ be a hook diagram. Then
\[T_{l}H_{k,Y}=\sum_{j=1}^{l}(-1)^{j-1}H_{k,\{Y_{l,j};Y\}}.\]
Here when $s>k$ or $j>k$, following our previous convention $H_{k,Y}, H_{k,\{Y_{l,j};Y\}}$ are zero.
\end{prop}

\begin{prop}\label{513prop6}
Let $l,k\geq 1, 1\leq j\leq l$, and let $Y_{l,j}$ be a hook diagram. Then
\[S_{l} H_k=\sum_{j=1}^{l} (-1)^{j-1} (2k-2j+l+\alpha) H_{k,Y_{l,j}}.\]
\end{prop}
\begin{proof}
Let $F_{i,j}$ be the $(i,j)$-cofactor of $(a_{\alpha+i+j})_{i,j=0,\ldots,k-1}$.
By the definition of $S_{l}$, for the case $l\leq k$,
\beas
S_{l}H_{k}&=&\Big((\alpha+i+j+l)a_{\alpha+i+j+l}\Big)_{i,j=0,1,\ldots,k-1}\cdot \Big(F_{i,j}\Big)_{i,j=0,1,\ldots,k-1}\\
&=&\Big((\alpha+j+l)a_{\alpha+i+j+l}\Big)_{i,j=0,1,\ldots,k-1}\cdot \Big(F_{i,j}\Big)_{i,j=0,1,\ldots,k-1}\\
&&+\,
\Big(ia_{\alpha+i+j+l}\Big)_{i,j=0,1,\ldots,k-1}\cdot \Big(F_{i,j}\Big)_{i,j=0,1,\ldots,k-1}\\
&=&\sum_{j=k-l}^{k-1}(\alpha+j+l)(-1)^{k-1-j}H_{k,Y_{l,k-j}}+\sum_{i=k-l}^{k-1}i(-1)^{k-1-i}H_{k,Y_{l,k-i}}\\
&=&\sum_{j=1}^{l} (-1)^{j-1} (\alpha+2k-2j+l) H_{k,Y_{l,j}}.
\eeas
Using a similar argument to the above, we can prove the case for $l>k$.
\end{proof}

\begin{rem}
\label{rem:627}
{\rm 
In the above two propositions, without considering our convention that $H_{k,Y}, H_{k,\{Y_{l,j};Y\}}$ are zero, then the equalities also hold and the summation over $j$ is from 1 to $\min(l,k)$.
}
\end{rem}

\section{Theorems and Propositions on Hankel determinants shifted by Young diagrams}

In this section, we prove some results for determinants of Hankel matrices whose entries are shifted by Young diagrams that we will make extensive use of later on in the paper.

\begin{thm} \label{thm1onHankel deteminants}
Let $l,k\geq 1, i\geq s\geq 1$. Let $Y_{i,s}$ be a hook diagram. Then for $j=2,\ldots,l$,
\bea
&& \sum_{h=1}^{j-1} (-1)^h T_{h} D_{k}(\beta_0,\ldots,\beta_{k-1};Y_{l-h,j-h}) \nonumber \\
&=& \sum_{h=1}^{j-1} (-1)^h D_{k}(\beta_0,\ldots,\beta_{k-1};Y_{l,j-h})
-(j-1) D_{k}(\beta_0,\ldots,\beta_{k-1};Y_{l,j})\label{513formula1}.
\label{eq:thm2}
\eea
\end{thm}
\begin{proof}
We first show that (\ref{513formula1}) holds when $j=2,\ldots,\min(l,k)$.
Let $Z_{h,q}$ and $W_{h,q}$ be two $k$-tuples satisfying that for any $0\leq i\leq (k-1)$,
\[
W_{h,q}(i)=\begin{cases}
   h+1& \text{if } i=q;\\
 1 & \text{if } k-j+h\leq i\leq k-2 \text{~and~} i\neq q;\\
  l-j+1& \text{if } i=k-1;\\
  0 & \text{otherwise}.
\end{cases}
\]
and
\[Z_{h,q}(i)=
\begin{cases}
  h & \text{if } i=q;\\
  1 & \text{if } k-j+h\leq i\leq k-2 \text{~and~} i\neq q;\\
  l-j+1& \text{if } i=k-1;\\
  0 & \text{otherwise},
\end{cases}
\]
For convenience, we use $x_{h,q}$ and $y_{h,q}$ to denote 
$D_{k}(\beta_{0},\ldots,\beta_{k-1};W_{h,q})$ and $D_{k}(\beta_{0},\ldots,\beta_{k-1};Z_{h,q})$, respectively. It is not hard to check that, for any $q$ with $k-j+h \leq q \leq k-2$, 
\bea\label{513formula6}
x_{h,q}=(-1)^{q-k+j-h}y_{j+q-k+1,k-(j-h)}.
\eea
By the definition of $T_{h}$, we have 
\beas
T_{h}D_{k}(\beta_{0},\ldots,\beta_{k-1};Y_{l-h,j-h})
&=& (-1)^{h-1} D_{k}(\beta_0,\ldots,\beta_{k-1};Y_{l,j}) + \sum_{q=k-h}^{k-1-(j-h)}y_{h,q}\\
&&+\sum_{q=\max(k-j+h,k-h-1)}^{k-2}x_{h,q}+D_{k}(\beta_0,\ldots,\beta_{k-1};Y_{l,j-h}).
\eeas
If $j$ is an even integer, then
\bea
&& \sum_{h=1}^{j-1} (-1)^h T_{h}D_{k}(\beta_0,\ldots,\beta_{k-1};Y_{l-h,j-h})\nonumber\\
&&- \sum_{h=1}^{j-1} (-1)^h D_{k}(\beta_0,\ldots,\beta_{k-1};Y_{l,j-h}) + (j-1) D_{k}(\beta_0,\ldots,\beta_{k-1};Y_{l,j}) \nonumber\\
&=& \sum_{h=1}^{j/2-1} (-1)^h \sum_{q=k-h-1}^{k-2} x_{h,q} + \sum_{h=j/2+1}^{j-1} (-1)^h \sum_{q=k-h}^{k-1-(j-h)} y_{h,q}
+\sum_{h=j/2}^{j-1} (-1)^h \sum_{q=k-(j-h)}^{k-2} x_{h,q} \label{513formula3}.
\eea
By (\ref{513formula6}),  the above formula
\bea
&=& \sum_{h=1}^{j/2-1} (-1)^{j} \sum_{q=k-h-1}^{k-2} (-1)^{q-k} y_{j+q-k+1, k-j+h} + \sum_{h=j/2+1}^{j-1} (-1)^h \sum_{q=k-h}^{k-1-(j-h)} y_{h,q} \nonumber \\
&& +\sum_{h=j/2}^{j-1} (-1)^{j} \sum_{q=k-(j-h)}^{k-2} (-1)^{q-k} y_{j+q-k+1,k-j+h} .
\label{513afternoon1}
\eea
By changing the variables and exchanging the order of summation, 
\beas
(\ref{513afternoon1})&=&
\sum_{h=1}^{j/2-1} \sum_{w=j-h}^{j-1} (-1)^{w-1} y_{w, k-j+h} + \sum_{h=j/2+1}^{j-1} (-1)^h \sum_{q=k-h}^{k-1-(j-h)} y_{h,q}\\
&&+\sum_{h=j/2}^{j-1} \sum_{w=h+1}^{j-1} (-1)^{w-1} y_{w,k-j+h}\\
% &=&\sum_{w=j/2+1}^{j-1} \sum_{h=j-w}^{j/2-1} (-1)^{w-1} y_{w, k-j+h}
% + \sum_{h=j/2+1}^{j-1} (-1)^h \sum_{q=k-h}^{k-1-(j-h)} y_{h,q}\\
% &&+\sum_{w=j/2+1}^{j-1}\sum_{h=j/2}^{w-1} (-1)^{w-1} y_{w,k-j+h}\\
&=&\sum_{w=j/2+1}^{j-1} \sum_{s=k-w}^{k-j/2-1} (-1)^{w-1} y_{w, s}+\sum_{h=j/2+1}^{j-1} (-1)^h \sum_{q=k-h}^{k-1-(j-h)} y_{h,q}\\
&&+\sum_{w=j/2+1}^{j-1}\sum_{h=k-j/2}^{k-j+w-1} (-1)^{w-1} y_{w,s}=0.
\eeas
By (\ref{513formula3}), we have that (\ref{eq:thm2}) holds for even integers $j=2,\ldots,\min(l,k)$. A similar argument leads to the verification of (\ref{eq:thm2}) for odd integers $j=2,\ldots,\min(l,k)$.

In the following, we shall establish (\ref{eq:thm2}) when $l\geq k+1$ and $j=k+1,\ldots,l$. The argument is indeed very similar. According to our previous convention (\ref{previous convention}), (\ref{eq:thm2}) is equivalent to the following equality:
\bea
\sum_{h=j-k}^{j-1} (-1)^h T_{h} D_{k}(\beta_0,\ldots,\beta_{k-1};Y_{l-h,j-h}) 
= \sum_{h=j-k}^{j-1} (-1)^h D_{k}(\beta_0,\ldots,\beta_{k-1};Y_{l,j-h})
\label{514formula4}.
\eea
Let $\widetilde{W}_{h,q}$ and $\widetilde{Z}_{h,q}$ be two $k$-tuples satisfying that for any $i$ with $0\leq i\leq (k-1)$,
\[
\widetilde{W}_{h,q}(i)=\begin{cases}
   h_{0}+h+1& \text{if } i=q;\\
 1 & \text{if } h-1\leq i\leq k-2 \text{~and~} i\neq q\\
  l-h_{0}-k& \text{if } i=k-1;\\
  0 & \text{otherwise}.
\end{cases}
\]
and
\[\widetilde{Z}_{h,q}(i)=
\begin{cases}
  h_{0}+h & \text{if } i=q;\\
  1 & \text{if } h-1\leq i\leq k-2 \text{ and } i\neq q;\\
  l-h_{0}-k& \text{if } i=k-1;\\
  0 & \text{otherwise}.
\end{cases}
\]
Denote $\widetilde{x}_{h,q}$ and $\widetilde{y}_{h,q}$ as 
$D_{k}(\beta_{0},\ldots,\beta_{k-1};\widetilde{W}_{h,q})$ and $D_{k}(\beta_{0},\ldots,\beta_{k-1};\widetilde{Z}_{h,q})$, respectively. 
It is not hard to check that for any $q$ with $h-1\leq q \leq k-2$, 
\bea\label{514formula67}
\widetilde{x}_{h,q}=(-1)^{q-h+1}\widetilde{y}_{q+2,h-1}.
\eea
Denote $h_{0}=j-k-1$. Then
\beas
&&T_{h_{0}+h}D_{k}(\beta_{0},\ldots,\beta_{k-1};Y_{l-h_{0}-h,k+1-h})\\
&=& \sum_{q=\max(k-h-h_{0},0)}^{h-2}\widetilde{y}_{h,q}+\sum_{q=\max(k-1-h_{0}-h,h-1)}^{k-2}\widetilde{x}_{h,q}+D_{k}(\beta_0,\ldots,\beta_{k-1};Y_{l,k+1-h}).
\eeas
Using the above equality, 
\beas
&&\sum_{h=1}^{k} (-1)^h T_{h_{0}+h}D_{k}(\beta_0,\ldots,\beta_{k-1};Y_{l-h_{0}-h,k+1-h})-\sum_{h=1}^{k} (-1)^hD_{k}(\beta_0,\ldots,\beta_{k-1};Y_{l,k+1-h})\\
&=&\sum_{h=1}^{k-h_{0}}(-1)^{h}\sum_{q=k-h-h_{0}}^{h-2}\widetilde{y}_{h,q}+\sum_{h=1}^{k-h_{0}}(-1)^{h}\sum_{q=\max(k-1-h_{0}-h,h-1)}^{k-2}\widetilde{x}_{h,q}\\
&& +\, \sum_{h=k-h_{0}+1}^{k}(-1)^{h}\sum_{q=0}^{h-2}\widetilde{y}_{h,q}+\sum_{h=k-h_{0}+1}^{k-1}(-1)^{h}\sum_{q=h-1}^{k-2}\widetilde{x}_{h,q}.
\eeas
If $k-h_{0}=2w$, then by the above and (\ref{514formula67}) with noting that the first term is nontrivial only when $h\geq w+1$,
we have
\beas
&&\sum_{h=1}^{k} (-1)^h T_{h_{0}+h}D_{k}(\beta_0,\ldots,\beta_{k-1};Y_{l-h_{0}-h,k+1-h})-\sum_{h=1}^{k} (-1)^hD_{k}(\beta_0,\ldots,\beta_{k-1};Y_{l,k+1-h})\\
=&&\sum_{h=w+1}^{2w}(-1)^{h}\sum_{q=2w-h}^{h-2}\widetilde{y}_{h,q}+\sum_{h=1}^{w}(-1)^{h}\sum_{q=2w-h-1}^{k-2}\widetilde{x}_{h,q}+\sum_{h=w+1}^{k-1}(-1)^{h}\sum_{q=h-1}^{k-2}\widetilde{x}_{h,q}\\
&&+\sum_{h=2w+1}^{k}(-1)^{h}\sum_{q=0}^{h-2}\widetilde{y}_{h,q}\\
=&&\sum_{h=1}^{w}\sum_{q=2w-h-1}^{k-2} (-1)^{q+1} \widetilde{y}_{q+2,h-1}+\sum_{h=w+1}^{k-1}\sum_{q=h-1}^{k-2} (-1)^{q+1} \widetilde{y}_{q+2,h-1}\\
&&+\sum_{h=w+1}^{2w}(-1)^{h}\sum_{q=2w-h}^{h-2}\widetilde{y}_{h,q}
+\sum_{h=2w+1}^{k}(-1)^{h}\sum_{q=0}^{h-2}\widetilde{y}_{h,q}.
\eeas
By changing the variables and exchanging the order of summation, we obtain that the above is $0$. This also holds if $k-h_{0}$ is an odd integer by a similar argument.
Hence, we have 
\beas
\sum_{h=1}^{k} (-1)^h T_{h_{0}+h}D_{k}(\beta_0,\ldots,\beta_{k-1};Y_{l-h_{0}-h,k+1-h})=\sum_{h=1}^{k} (-1)^hD_{k}(\beta_0,\ldots,\beta_{k-1};Y_{l,k+1-h}).
\eeas
This implies
\beas
\sum_{h=j-k}^{j-1} (-1)^h T_{h}D_{k}(\beta_0,\ldots,\beta_{k-1};Y_{l-h,j-h})&=&(-1)^{j-k-1}\sum_{h=1}^{k} (-1)^hD_{k}(\beta_0,\ldots,\beta_{k-1};Y_{l,k+1-h})\\
&=&\sum_{h=j-k}^{j-1} (-1)^h D_{k}(\beta_0,\ldots,\beta_{k-1};Y_{l,j-h}).
\eeas
The above gives the desired result (\ref{514formula4}).
This completes the proof.
\end{proof}

The proof of the following lemma is straightforward. 

\begin{lem}\label{matrix}
Let $l\geq 1$ be an integer, and let $B^{(l)}$ be as given in
(\ref{definition of B}), then $A^{(l)} := (B^{(l)})^{-1} =(a_{ij}^{(l)})_{i,j=1,..,l}$ is a lower Hessenberg matrix satisfying 
\be
a_{ij}^{(l)} =\begin{cases}
    (-1)^{i-j+1} & j\leq i\leq l-1; \\
    -i & j=i+1; \\
    (-1)^{j-1} &i=l;\\
    0  & j>i+1
\end{cases}
\label{definition of A}
\ee
and  $\det(A)=l!$. 
\end{lem}

% \begin{lem}\label{matrix}
% Let $m\geq 1$ be an integer. Let $A^{(m)}=(a_{ij})_{i,j=1..m}$ be a lower Hessenberg matrix satisfying 
% \be
% a_{ij} =\begin{cases}
%     (-1)^{i-j+1} & j\leq i\leq m-1; \\
%     -i & j=i+1; \\
%     (-1)^{j-1} &i=m;\\
%     0  & j>i+1
% \end{cases}\label{definition of A}
% \ee
% Then $\det(A)=m!$. Moreover $B^{(m)}:=(A^{(m)})^{-1}=(b_{ij})_{i,j=1..m}$ is a upper Hessenberg matrix satisfying 
% \be
% b_{ij} =\begin{cases}
% (-1)^{i-j+1}/j(j+1) & i\leq j\leq m-1; \\
% -1/i & j=i-1;\\
% (-1)^{i-1}/m & j=m;\\
% 0  & j<i-1.
% \end{cases}
% \label{definition of B}
% \ee
% \end{lem}

\begin{thm}
\label{gerenel linear systems}
Let $l,k\geq 1, i\geq j\geq 1$. Let $Y_{i,j}$ be a hook diagram. Then for any Young diagram $X$,
\be
\label{linear system 1}
\begin{pmatrix}
H_{k,\{X;Y_{l,1}\}} \\
\vdots \\
\vdots \\
H_{k,\{X;Y_{l,l}\}} \\
\end{pmatrix}
=B^{(l)}
\begin{pmatrix}
\vdots \\
\sum_{h=1}^{j-1} (-1)^h T_{h} H_{k,\{X;Y_{l-h,j-h}\}} \\
\vdots \\
T_{l}H_{k,X} \\
\end{pmatrix}_{j=2,\ldots,l},
\ee
where $B^{(l)}$ is the $(l\times l)$-matrix given in (\ref{definition of B}).
\end{thm}

\begin{proof}
Suppose that $X=(l_{1},\ldots,l_{h})$, then (\ref{linear system 1}) follows from Proposition \ref{513prop5}, Theorem \ref{thm1onHankel deteminants} with $(\beta_{0},\ldots, \beta_{k-1})=(\alpha,\alpha+1,\ldots, \alpha+k-h+1,\alpha+k-h+l_{h},\ldots,\alpha+k-1+l_{1})$ and Lemma \ref{matrix}.
\end{proof}

As a consequence, we show that any $H_{k,\{X,Y_{l,j}\}}$, which we recall is the determinant of a Hankel matrix whose rows and columns are both shifted by Young diagrams, can be expressed by a linear combination of determinants of Hankel matrices with only rows or columns shifted by Young diagrams.

\begin{prop}
\label{320eveprop6}
Let $q\geq 1$, and $X$ be a Young diagram $(l_1,\ldots,l_{q})$ of length $m$. Let $Y_{l,j}$ be a hook diagram. Then for any $j=1,2,\ldots,l$, $H_{k,\{X;Y_{l,j}\}}$ is a linear combination of $H_{k,X_n}, n=1,\ldots,C_{X,l,j}$ for some constant $C_{X,l,j}$ depending on $X,l,j$, where $X_n$ is a Young diagram of length $m+l$. 
\end{prop}

\begin{proof}
We prove the claim by induction. For $l=1$, the claim follows from Proposition \ref{prop1} because
$H_{k,\{\widetilde{X},Y_{1,1}\}} = T_1 H_{k,\widetilde{X}}$ for any $\widetilde{X}$. By induction, for any hook diagram $Y_{i,j}$ of length $i\leq l-1$ and any Young diagram $\widetilde{X}$ of length $\tilde{m}$, we assume that $H_{k, \{\widetilde{X}; Y_{i,j}\}}, j=1,2,\ldots,i$ are linear combinations of $H_{k,\widetilde{X}_n}$ for some Young diagrams $\widetilde{X}_n$ having length $\tilde{m}+i$. For convenience, we below set $l_{q+1}=\cdots=l_k=0$ when $q\leq k-1$.
By Proposition \ref{prop1}, 
\be
\label{new label}
T_{h}H_{k,\{X;Y_{l-h,j-h}\}} = \sum_{s=1}^{k}H_{k,\{(l_{1},\ldots,l_{s}+h,\ldots,l_{k}); Y_{l-h,j-h}\}},
\quad 
T_{l}H_{k,X}=\sum_{s=1}^{k} H_{k,(l_1,\ldots,l_{s}+l,\ldots,l_{k})}.
\ee
By Theorem \ref{gerenel linear systems}, there is an invertible matrix $B^{(l)}$ such that 
\[
\begin{pmatrix}
H_{k,\{X;Y_{l,1}\}} \\
\vdots \\
\vdots \\
H_{k,\{X;Y_{l,l}\}} \\
\end{pmatrix}
=B^{(l)}
\begin{pmatrix}
\vdots \\
\sum_{h=1}^{j-1} (-1)^h T_{h} H_{k,\{X;Y_{l-h,j-h}\}} \\
\vdots \\
T_{l}H_{k,X} \\
\end{pmatrix}_{j=2,..,l}.
\]
By (\ref{new label}) and the inductive assumption, we obtain the claimed result.
\end{proof}

\begin{rem}
\label{natural extension}
{\rm
We explain now why we set $M_{k},\widetilde{M}_{k}, D_{k}$ to be $0$ when $s>k$ in Section \ref{section:Determinants of Hankel matrices}. We initially considered only the  truncated case 
\be\label{truncated case}
\begin{pmatrix}
H_{k,\{X;Y_{l,1}\}} \\
\vdots \\
\vdots \\
H_{k,\{X;Y_{l,l_{0}}\}} \\
\end{pmatrix}
=B^{(l_{0})}
\begin{pmatrix}
\vdots \\
\sum_{h=1}^{j-1} (-1)^h T_{h} H_{k,\{X;Y_{l-h,j-h}\}} \\
\vdots \\
T_{l}H_{k,X} \\
\end{pmatrix}_{j=2,\ldots,l_{0}}
\ee
where $l_{0}=\min(l,k)$, because from (\ref{def of H-k-X-Y}), $H_{k,\{X;Y_{l,j}\}}$ is a determinant of a matrix of $k$ columns, which is well-defined when $j\leq k$. However, $l_{0}=\min(l,k)$ removes the information relating to $k$ when $k>l$. Obviously, when $k>l$, $l_0=l$ and so is independent of $k$.
To obtain Theorem \ref{theorem2 in note 4}, which is a combination of certain orders of derivatives of $H_k$ whose coefficients are polynomial in $k$, we have to remove this restriction to recover the dependence on $k$. To this end, we need to extend (\ref{truncated case}). One natural way is replacing $B^{(l_{0})}$ with $B^{(l)}$ in (\ref{truncated case}). This gives (\ref{linear system 1}).
Interestingly, when $l>k$, we have $H_{k,\{X;Y_{l,k+1}\}}=\dots=H_{k,\{X;Y_{l,l}\}}=0$ in (\ref{linear system 1}). 
Indeed, by (\ref{linear system 1}), $H_{k,\{X;Y_{k+1,k+1}\}} = -\frac{1}{k+1} \sum_{h=1}^k (-1)^h T_h H_{k,\{X;Y_{k+1-h,k+1-h}\}} + \frac{(-1)^k}{k+1} T_{k+1} H_{k,X}$. 
Then by (\ref{514formula4}) with $j=k+1$, we have $\sum_{h=1}^k (-1)^h T_h H_{k,\{X;Y_{k+1-h,k+1-h}\}} = (-1)^k \sum_{h=1}^k (-1)^{h-1} H_{k,\{X;Y_{k+1,h}\}}$. By Remark \ref{rem:627}, this further equals $(-1)^k T_{k+1} H_{k,X}$. So $H_{k,\{X;Y_{k+1,k+1}\}}=0$. 
For $k+1 \leq j \leq l$,
by (\ref{linear system 1}), we have $H_{k,\{X;Y_{l,j}\}} = \sum_{q=1}^{l-1} (B^{(l)})_{j,q} \sum_{h=1}^q (-1)^h T_h H_{k,\{X;Y_{l-h,q+1-h}\}} + (B^{(l)})_{j,l} T_l H_{k,X}$. By the inductive assumption that for any $k+1\leq l'\leq l-1$, $H_{k,\{X;Y_{l',j}\}} = 0 $ for $k+1 \leq j \leq l'$. Using  (\ref{514formula4}) and Remark \ref{rem:627}, when $q \geq j-1\geq k$, we have $\sum_{h=1}^q (-1)^h T_h H_{k,\{X;Y_{l-h,q+1-h}\}} = (-1)^{q} \sum_{h=1}^k (-1)^{h-1} H_{k,\{X;Y_{l,h}\}} = (-1)^q T_l H_{k,X}$, so $H_{k,\{X;Y_{l,j}\}} = 0$.

Moreover, when $l>k$, the first $k$ entries in the left hand side of (\ref{linear system 1}) are the same as those from (\ref{truncated case}). This is the reason why we set $M_{k},\widetilde{M}_{k}, D_{k}$ to be $0$ when $s>k$.
In the latter part of this paper (see Sections \ref{rfHsy}, \ref{LiftrfHsY} and \ref{section:Generalization}), we will expand the right-hand side of (\ref{linear system 1}) by making use of recursive formulae for some specific sequences (such as the modified Bessel function of the first kind). This will represent $H_{k,Y_{l,j}}, (j=1,\ldots,l)$ as certain polynomials of derivatives of $H_{k}$. So $H_{k,\{X;Y_{l,j}\}}$ does not appear as 0 formally even if $l>k$. But it is an expression that is ultimately zero. This will lead to some differential equations, e.g., see (\ref{diff equation}) below.

}\end{rem}

\begin{prop}\label{411proponS}
Let $l,k\geq 1, i\geq j\geq 1$. Let $Y_{i,j}$ be a hook diagram. Then for any $j=3,\ldots,l$, we have 
\[\sum_{h=2}^{j-1}(-1)^{h}S_{h-1}H_{k,Y_{l-h,j-h}}=\sum_{q=1}^{l-1}d_{q}H_{k,Y_{l-1,q}},\]
where
\[d_{q}=
\begin{cases}
(-1)^{j-q}(\alpha+2k-1-2q+l) & \text{if } 1\leq q\leq j-2,\\
 (j-2)(\alpha+2k-j+1)& \text{if } q=j-1, \\
  0 & \text{if } j\leq q \leq l-1.
\end{cases}
\]
\end{prop}

Before proving the above proposition, we need some further preparations. 

\begin{lem}\label{lemmaonS1}
Making the same assumptions as in Proposition \ref{411proponS}, let 
\[
\widetilde{M}_k^{(1)}=\Big((\alpha+j+h-1+t_{k-j})a_{\alpha+i+j+h-1+t_{k-j}}\Big)_{i,j=0,\ldots,k-1},
\]
and let $S_{h-1}^{(1)}H_{k,Y_{l-h,j-h}}$ be the standard inner product between $\widetilde{M}_k^{(1)}$ and the cofactor matrix of the matrix in defining $H_{k,Y_{l-h,j-h}}$. Then
\[
\sum_{h=2}^{j-1}(-1)^{h}S_{h-1}^{(1)}H_{k,Y_{l-h,j-h}}=\sum_{h=2}^{j-1}(-1)^{h}(\alpha+k-1+l-j+h) H_{k,Y_{l-1,j-h}}+H_{k,Y_{l-1,j-1}} \sum_{h=2}^{j-1}(\alpha+k-j+h).
\]
\end{lem}

\begin{proof}
We use a similar trick to that in the proof of Theorem \ref{thm1onHankel deteminants}. The difference is that now we are considering weighted sums.
We first assume that $j\leq k+1$. For fixed $j\geq 3, 2 \leq h \leq j-1, 0 \leq q \leq k-2$,
let $W_{h,q}$ and $Z_{h,q}$ be two $k$-tuples satisfying that for any $i$ with $0\leq i\leq k-1$,
\[
W_{h,q}(i)=\begin{cases}
   h & \text{if } i=q;\\
 1 & \text{if } k-j+h\leq i\leq k-2 \text{~and~} i\neq q\\
  l-j+1& \text{if } i=k-1;\\
  0 & \text{otherwise}.
\end{cases}
\]
and 
\[
Z_{h,q}(i)=
\begin{cases}
  h-1 & \text{if } i=q;\\
  1 & \text{if } k-j+h\leq i\leq k-2 \text{~and~} i \neq q;\\
  l-j+1& \text{if } i=k-1;\\
  0 & \text{otherwise},
\end{cases}
\]
We use $x_{h,q}$ and $y_{h,q}$ to denote 
$H_{k,W_{h,q}}$ and $H_{k,Z_{h,q}}$, respectively. It is not hard to check that, for any $q$ with $k-j+h\leq q \leq k-2$, 
\bea\label{411formula6}
x_{h,q}=(-1)^{q-k+j-h}y_{j+q-k+1,k-(j-h)}.
\eea
Using the notation introduced above and the definition of $S_{h-1}^{(1)} H_{k, Y_{l-h,j-h}}$, we have
\bea
 S_{h-1}^{(1)} H_{k, Y_{l-h,j-h}}
&=& (\alpha+k-1+l-j+h)H_{k,Y_{l-1,j-h}} 
+ (\alpha+k-j+h) (-1)^{h-2} H_{k,Y_{l-1,j-1}} \nonumber\\
&&+ \sum_{q=k-h+1}^{k-1-(j-h)}(\alpha+q+h-1)y_{h,q}+\sum_{q=\max(k-h,k-(j-h))}^{k-2} (\alpha+q+h)x_{h,q}.\label{411formula7} 
\eea
If $j$ is an even integer, then
\bea
&&  \sum_{h=2}^{j-1} (-1)^{h} S_{h-1}^{(1)} H_{k, Y_{l-h,j-h}} \nonumber \\
&&-\sum_{h=2}^{j-1} (-1)^h (\alpha+k-1+l-j+h) H_{k,Y_{l-1,j-h}} - \sum_{h=2}^{j-1} (\alpha+k-j+h) H_{k,Y_{l-1,j-1}}\nonumber \\
&=& \sum_{h=2}^{j/2} (-1)^h \sum_{q=k-h}^{k-2} (\alpha+q+h) x_{h,q}
+ \sum_{h=j/2+1}^{j-1} (-1)^h \sum_{q=k-h+1}^{k-1-(j-h)} (\alpha+q+h-1) y_{h,q} \nonumber \\ 
&& + \sum_{h=j/2+1}^{j-1} (-1)^h \sum_{q=k-j+h}^{k-2} (\alpha+q+h) x_{h,q}
\label{411eq11}
\eea
By relation (\ref{411formula6}),
\bea\label{411eq12}
(\ref{411eq11})&=&  \sum_{h=2}^{j/2} (-1)^h \sum_{q=k-h}^{k-2} (\alpha+q+h) (-1)^{q-k+j-h} y_{j+q-k+1,k-j+h}\nonumber \\
&& + \sum_{h=j/2+1}^{j-1} (-1)^h \sum_{q=k-h+1}^{k-1-(j-h)} (\alpha+q+h-1) y_{h,q} \nonumber \\ 
&& + \sum_{h=j/2+1}^{j-1} (-1)^h \sum_{q=k-j+h}^{k-2} (\alpha+q+h) (-1)^{q-k+j-h} y_{j+q-k+1,k-(j-h)}.
\eea
Changing variables and exchanging the order of the summation, we have
\bea\label{411eq13}
(\ref{411eq12})&=&  - \sum_{h_1=j/2+1}^{j-1}(-1)^{h_1} \sum_{q=k-h_1+1}^{k-j/2} (\alpha+h_1-1+q) y_{h_1,q}\nonumber \\
&& + \sum_{h=j/2+1}^{j-1} (-1)^h \sum_{q=k-h+1}^{k-1-(j-h)} (\alpha+q+h-1) y_{h,q} \nonumber \\ 
&& - \sum_{h_1=j/2+2}^{j-1}  (-1)^{h_1} \sum_{q=k-j/2+1}^{k-j+h_1-1} (\alpha+h_1-1+q) y_{h_1,q}=0.
\eea
Substituting (\ref{411eq12}), (\ref{411eq13}) into  (\ref{411eq11}), we obtain
\bea\label{411eq14}
\sum_{h=2}^{j-1} (-1)^{h} S_{h-1}^{(1)} H_{k, Y_{l-h,j-h}}&=&\sum_{h=2}^{j-1} (-1)^h (\alpha+k-1+l-j+h) H_{k,Y_{l-1,j-h}}\nonumber\\ 
&& + \sum_{h=2}^{j-1} (\alpha+k-j+h)H_{k,Y_{l-1,j-1}}.
\eea
By a similar argument, we conclude that (\ref{411eq14}) holds for odd integer $j$. 

When $j\geq k+2$,  by a similar argument to that relating to (\ref{514formula4}) and the case $j\leq k+1$ above, we have
\[
\sum_{h=2}^{j-1}(-1)^{h}S_{h-1}^{(1)}H_{k,Y_{l-h,j-h}}=(-1)^{j}\sum_{h=1}^{k}(-1)^{h}(\alpha+k-1+l-h)H_{k,Y_{l-1,h}}.
\]
This is as claimed in the lemma, therefore completing the proof. 
\end{proof}

The following lemma follows directly from the definition of $B^{(m)}$.

\begin{lem}\label{410lemma1}
Let $m\geq 1$. Let $B^{(m)}=(b_{i,j})_{i,j=1,..,m}$ be given in (\ref{definition of B}) with $l$ replaced with $m$.
Then 
$\sum_{i=1}^m (-1)^{i-1} (k-i) b_{i,j}=(-1)^j/2$ 
for any $j$ with $1 \leq j \leq m-1$, and
$\sum_{i=1}^m (-1)^{i-1} (k-i) b_{i,j} = k-\frac{m+1}{2}$ 
for $j=m$.
\end{lem}

\begin{lem}\label{lemmaonS2}
Making the same assumptions as in Proposition \ref{411proponS}. 
Let 
\[\widetilde{M}_k^{(2)}=\Big(ia_{\alpha+i+j+h-1+t_{k-j}}\Big)_{i,j=0,\ldots,k-1},\]
and $S_{h-1}^{(2)}H_{k,Y_{l-h,j-h}}$ be the standard inner product between $\widetilde{M}_k^{(2)}$ and the cofactor matrix of the matrix in defining $H_{k,Y_{l-h,j-h}}$.
Let $(b_{1},\ldots,b_{l-1})$ be an $(l-1)$-tuple with $b_{q}=(-1)^{j-q}(k-q)$ when $1\leq q\leq (j-2)$, $b_{j-1}=(k-\frac{j-1}{2})(j-2)$, and $b_q=0$ when $j\leq q \leq l-1$. Then
\[\sum_{h=2}^{j-1}(-1)^{h}S_{h-1}^{(2)}H_{k,Y_{l-h,j-h}}=\sum_{q=1}^{l-1}b_{q}H_{k,Y_{l-1,q}}.\]

\end{lem}

\begin{proof}

It is not hard to check that
\bea\label{410formula2}
S_{h-1}^{(2)}H_{k,Y_{l-h,j-h}}=\sum_{s=1}^{h-1}(-1)^{s-1}(k-s)H_{k,\{Y_{h-1,s};Y_{l-h,j-h}\}}.
\eea
By Theorem \ref{gerenel linear systems}, Lemma \ref{410lemma1} and (\ref{410formula2}), we have
\bea\label{410formula3}
&&S_{h-1}^{(2)}H_{k,Y_{l-h,j-h}}-(k-\frac{h}{2}) T_{h-1} H_{k,Y_{l-h,j-h}}\nonumber\\
&=& \frac{1}{2} \sum_{j_1=2}^{h-1} (-1)^{j_1-1} \sum_{h_1=1}^{j_1-1} (-1)^{h_1} T_{h_1} H_{k,\{Y_{l-h,j-h}; Y_{h-1-h_1,j_1-h_1}\}}
\eea
So for $j=3,\ldots,l$,
\bea\label{410formula5}
&& \sum_{h=2}^{j-1} (-1)^h S_{h-1}^{(2)} H_{k,Y_{l-h,j-h}}-\sum_{h=2}^{j-1} (-1)^h (k-\frac{h}{2}) T_{h-1} H_{k,Y_{l-h,j-h}}\nonumber  \\
&=& \frac{1}{2} \sum_{h=2}^{j-1} (-1)^h \sum_{j_1=2}^{h-1} (-1)^{j_1-1} \sum_{h_1=1}^{j_1-1} (-1)^{h_1} T_{h_1} H_{k,\{Y_{l-h,j-h}; Y_{h-1-h_1,j_1-h_1}\}} \nonumber\\
&=& \frac{1}{2} \sum_{h=2}^{j-2} (-1)^h \sum_{j_1=2}^{h-1} (-1)^{j_1-1} \sum_{h_1=1}^{j_1-1} (-1)^{h_1} T_{h_1} H_{k,\{Y_{l-h,j-h}; Y_{h-1-h_1,j_1-h_1}\}} \nonumber\\
&& +\, \frac{1}{2} (-1)^{j-1} \sum_{j_1=2}^{j-2} (-1)^{j_1-1} \sum_{h_1=2}^{j_1-1} (-1)^{h_1} T_{h_1} H_{k,\{Y_{l-j+1,1}; Y_{j-2-h_1,j_1-h_1}\}} \nonumber\\
&& +\, \frac{1}{2} (-1)^{j-1} \sum_{j_1=2}^{j-2} (-1)^{j_1}T_{1} H_{k,\{Y_{l-j+1,1}; Y_{j-3,j_1-1}\}}.
\eea
From the second equality, we can see that for (\ref{410formula5}) the summation over $j_{1}$ constrains the summation of $h$, starting from 3 in the first term, and the summation over $h_{1}$ constrains the summation of $j_{1}$, starting from 3 in the second term. 
So by exchanging the order of the summation, we have the following expression.
\bea\label{410formula6}
(\ref{410formula5})&=& \frac{1}{2} \sum_{h_1'=1}^{j-4}\sum_{h=h_1'+2}^{j-2} \sum_{j_1'=h_1'+1}^{h-1} (-1)^{j_1'+h+h_1'-1} T_{h_1'} H_{k,\{Y_{l-h,j-h}; Y_{h-1-h_1',j_1'-h_1'}\}} \nonumber\\
&& +\, \frac{1}{2} \sum_{h_1=2}^{j-3 } \sum_{j_1=h_1+1}^{j-2} (-1)^{j+h_1+j_1} T_{h_1}H_{k,\{Y_{l-j+1,1}; Y_{j-2-h_{1},j_1-h_1}\}}\nonumber \\
&& +\, \frac{1}{2} (-1)^{j-1} \sum_{j_1=2}^{j-2} (-1)^{j_1} T_{1} H_{k,\{Y_{l-j+1,1}; Y_{j-3,j_1-1}\}}.
\eea
By changing variables $h-1-h_{1}'=j-2-h_{1}$ and $j_{1}'-h_{1}'=j_{1}-h_{1}$ in the first term of the above, we have 
\bea\label{410formula7}
(\ref{410formula6})&=& \frac{1}{2} \sum_{h_{1}=2}^{j-3}\sum_{j_{1}=h_{1}+1}^{j-2}\Big(\sum_{h_{1}'=1}^{h_{1}-1}(-1)^{j+j_{1}+h_{1}'}T_{h_{1}'}
H_{k,\{Y_{l-j+1+h_{1}-h_{1}',h_{1}-h_{1}'+1}; Y_{j-2-h_{1},j_{1}-h_{1}}\}} \nonumber\\
&&+\,
(-1)^{h_{1}+j_{1}+j}T_{h_{1}}H_{k,\{Y_{l-j+1,1};Y_{j-2-h_{1},j_{1}-h_{1}}\}}
\Big)\nonumber\\
&&+\, \frac{1}{2} (-1)^{j-1} \sum_{j_1=2}^{j-2} (-1)^{j_1} T_{1} H_{k,\{Y_{l-j+1,1}; Y_{j-3,j_1-1}\}}\nonumber\\
&=&\frac{(-1)^{j}}{2}\sum_{h_{1}=2}^{j-3}\sum_{j_{1}=h_{1}+1}^{j-2}(-1)^{j_{1}}\sum_{h_{1}'=1}^{h_{1}}(-1)^{h_{1}'}T_{h_{1}'}
H_{k,\{Y_{l-j+1+h_{1}-h_{1}',h_{1}-h_{1}'+1}; Y_{j-2-h_{1},j_{1}-h_{1}}\}}\nonumber\\
&&+\, \frac{1}{2} (-1)^{j-1} \sum_{j_1=2}^{j-2} (-1)^{j_1} T_{1} H_{k,\{Y_{l-j+1,1}; Y_{j-3,j_1-1}\}}.
\eea
By Theorem \ref{thm1onHankel deteminants}, we have
\bea\label{410formula8}
&&\sum_{h_{1}'=1}^{h_{1}}(-1)^{h_{1}'}T_{h_{1}'}
H_{k,\{Y_{l-j+1+h_{1}-h_{1}',h_{1}-h_{1}'+1}; Y_{j-2-h_{1},j_{1}-h_{1}}\}}\nonumber\\
&=&
\sum_{s=1}^{h_{1}}(-1)^{h_{1}-s+1}H_{k,\{Y_{j-2-h_1,j_1-h_1} ; Y_{l-j+1+h_1,s}\}}-h_{1}H_{k,\{Y_{j-2-h_1,j_1-h_1} ; Y_{l-j+1+h_1,h_1+1}\}}.
\eea
By Proposition \ref{513prop5} and (\ref{410formula8}),
\bea\label{410formula10}
&&\sum_{j_{1}=h_{1}+1}^{j-2}(-1)^{j_{1}}\sum_{h_{1}'=1}^{h_{1}}(-1)^{h_{1}'}T_{h_{1}'}
H_{k,\{Y_{l-j+1+h_{1}-h_{1}',h_{1}-h_{1}'+1}; Y_{j-2-h_{1},j_{1}-h_{1}}\}}\nonumber\\
&=& \sum_{s=1}^{h_1} (-1)^{s+1} \sum_{j_1=1}^{j-2-h_1} (-1)^{j_1}H_{k, \{Y_{j-2-h_1,j_1} ; Y_{l-j+1+h_1,s}\}} \nonumber\\
&& - \,(-1)^{h_1} h_1  \sum_{j_1=1}^{j-2-h_1} (-1)^{j_1} H_{k,\{Y_{j-2-h_1,j_1} ; Y_{l-j+1+h_1,h_1+1}\}}\nonumber \\
&=& \sum_{s=1}^{h_1} (-1)^{s} T_{j-2-h_1}H_{k, Y_{l-j+1+h_1,s}} + h_1 (-1)^{h_1} T_{j-2-h_1} H_{k,Y_{l-j+1+h_1,h_1+1}}.
\eea
Now by formulae (\ref{410formula5})-(\ref{410formula7}) and (\ref{410formula10}),
\bea
\sum_{h=2}^{j-1} (-1)^{h} S_{h-1}^{(2)} H_{k, Y_{l-h,j-h}} 
&=& \frac{(-1)^j}{2} \sum_{h_1=2}^{j-3} \sum_{s=1}^{h_1} (-1)^s T_{j-2-h_1}H_{k,Y_{l-j+1+h_1,s}} \nonumber \\
&& +\, \frac{(-1)^j}{2} \sum_{h_1=2}^{j-3} h_1 (-1)^{h_1} T_{j-2-h_1} H_{k,Y_{l-j+1+h_1,h_1+1}} \nonumber  \\
&& +\, \frac{(-1)^{j-1}}{2} \sum_{j_1=2}^{j-2} (-1)^{j_1} T_{1}H_{k,\{Y_{l-j+1,1} ; Y_{j-3,j_1-1}\}} \nonumber \\
&& +\, \sum_{h=2}^{j-1}(-1)^{h} (k-\frac{h}{2}) T_{h-1}H_{k,Y_{l-h,j-h}}. \label{eq31}
\eea
By Propositions \ref{prop1} and \ref{513prop5},
\beas
&&\frac{(-1)^{j-1}}{2} \sum_{j_1=2}^{j-2} (-1)^{j_1} T_{1}H_{k,\{Y_{l-j+1,1} ; Y_{j-3,j_1-1}\}}\\
&=&
\frac{(-1)^{j-1}}{2} \sum_{j_1=2}^{j-2} (-1)^{j_1} H_{k,\{Y_{l-j+2,1} ; Y_{j-3,j_1-1}\}}+ \, \frac{(-1)^{j-1}}{2} \sum_{j_1=2}^{j-2} (-1)^{j_1} H_{k,\{Y_{l-j+2,2} ; Y_{j-3,j_1-1}\}} \\
&=& \frac{(-1)^{j-1}}{2} T_{j-3} (H_{k,Y_{l-j+2,1}} + H_{k,Y_{l-j+2,2}}).
\eeas
So by the above and (\ref{eq31}),
\bea\label{0410eq12}
&& \sum_{h=2}^{j-1} (-1)^{h} S_{h-1}^{(2)} H_{k, Y_{l-h,j-h}}\nonumber  \\
&=& \frac{(-1)^j}{2} \sum_{h_1=2}^{j-3} \sum_{s=1}^{h_1} (-1)^s T_{j-2-h_1} H_{k,Y_{l-j+1+h_1,s}}+\, \sum_{h=2}^{j-1} (-1)^h (k-\frac{h}{2}) T_{h-1}H_{k,Y_{l-h,j-h}}\nonumber\\
&& -\, \frac{1}{2} \sum_{h=2}^{j-2} (j-1-h) (-1)^{h} T_{h-1} H_{k,Y_{l-h,j-h}}+\frac{(-1)^{j-1}}{2} T_{j-3} H_{k,Y_{l-j+2,1}}\nonumber \\
&=&\frac{(-1)^j}{2} \sum_{h_1=2}^{j-3} \sum_{s=1}^{h_1} (-1)^s T_{j-2-h_1} H_{k,Y_{l-j+1+h_1,s}}\nonumber\\
&& +(k-\frac{j-1}{2}) \sum_{h=2}^{j-1} (-1)^{h} T_{h-1} H_{k,Y_{l-h,j-h}}+\frac{(-1)^{j-1}}{2} T_{j-3} H_{k,Y_{l-j+2,1}} \label{411formula1}
\eea
Observe that
\bea
&&\frac{(-1)^j}{2} \sum_{h_1=2}^{j-3} \sum_{s=1}^{h_1} (-1)^s T_{j-2-h_1} H_{k,Y_{l-j+1+h_1,s}}\nonumber\\
&=& \frac{(-1)^j}{2} \sum_{h_1=2}^{j-3} \sum_{s=1}^{j-2-h_{1}} (-1)^s T_{h_{1}-1}H_{k,Y_{l-h_{1},s}}- \frac{1}{2}\sum_{h=2}^{j-3}(-1)^{h}T_{h-1}H_{k,Y_{l-h,j-1-h}}\nonumber\\
&=& \frac{(-1)^j}{2} \sum_{h=1}^{j-4} (-1)^h \sum_{s_1=h+1}^{j-3} (-1)^{s_1} T_{h}H_{k,Y_{l-1-h,s_1-h}}
-\frac{1}{2}\sum_{h=2}^{j-3}(-1)^{h}T_{h-1}H_{k,Y_{l-h,j-1-h}}\nonumber  \\
&=& \frac{(-1)^j}{2} \sum_{s_1=2}^{j-3} (-1)^{s_1} \sum_{h=1}^{s_1-1} (-1)^h  T_{h}H_{k, Y_{l-1-h,s_1-h}}
-\frac{1}{2}\sum_{h=2}^{j-3}(-1)^{h}T_{h-1}H_{k,Y_{l-h,j-1-h}}\label{411formula2}.
\eea
By formulae (\ref{411formula1}) and (\ref{411formula2}), we have
\bea
&&\sum_{h=2}^{j-1} (-1)^{h} S_{h-1}^{(2)} H_{k, Y_{l-h,j-h}}\nonumber\\
&=&\frac{(-1)^j}{2} \sum_{s_1=2}^{j-2} (-1)^{s_1} \sum_{h=1}^{s_1-1} (-1)^h  T_{h} H_{k,Y_{l-1-h,s_1-h}}
 -(k-\frac{j-1}{2})\sum_{h=1}^{j-2} (-1)^h  T_{h}H_{k,Y_{l-1-h,j-1-h}}\label{411formula3} .
\eea
Let $A^{(l-1)} = (a_{i,j})_{i,j=1,\ldots,k}$ be the matrix given in (\ref{definition of A}) with $l$ replaced by $l-1$. By Theorem \ref{gerenel linear systems} with $X=\emptyset$ and (\ref{411formula3}),
\bea
\sum_{h=2}^{j-1} (-1)^{h} S_{h-1}^{(2)} H_{k, Y_{l-h,j-h}} &=&
\frac{(-1)^j}{2} \sum_{s_1=2}^{j-2} (-1)^{s_1} \sum_{q=1}^{l-1} a_{s_{1}-1,q}  H_{k,Y_{l-1,q}} \nonumber \\
&& -\, (k-\frac{j-1}{2}) \sum_{q=1}^{l-1} a_{j-2,q} H_{k,Y_{l-1,q}} . \label{411formula4}
\eea
Then by (\ref{411formula4}) and the definition of $A^{(l-1)}$, we obtain the claim in the lemma.
\end{proof}

\begin{proof}[Proof of Proposition \ref{411proponS}]
This is an immediate consequence of Lemmas \ref{lemmaonS1} and \ref{lemmaonS2} and the fact that 
$
S_{h-1}H_{k,Y_{l-h,j-h}}=S_{h-1}^{(1)}H_{k,Y_{l-h,j-h}}+S_{h-1}^{(2)}H_{k,Y_{l-h,j-h}}.$
\end{proof}

\section{Recursive formulae for $\tau_{k,Y}(x)$}
\label{rfHsy}

In this section, we establish general recursive formulae for Hankel determinants of I-Bessel functions shifted by Young diagrams.

Let $I_{\beta}(x)$ be the Bessel function of the first kind with power series expansion
\[I_{\beta}(x) =(x/2)^{\beta}\sum_{j=0}^{\infty}\frac{x^{2j}}{2^{2j}j! \Gamma(\beta+j+1)},\]
where $\beta$ is a complex number and $\Gamma(z)$ is the Gamma function. This satisfies the following recursive relations:
\bea
\frac{d}{dx} I_{\beta}(2\sqrt{x}) &=& 
\frac{I_{\beta+1}(2\sqrt{x})}{\sqrt{x}}  + \frac{\beta}{2x}I_{\beta}(2\sqrt{x}) , \nonumber \\
\frac{d}{dx} I_{\beta}(2\sqrt{x}) &=& 
\frac{I_{\beta-1}(2\sqrt{x})}{\sqrt{x}}  - \frac{\beta}{2x}I_{\beta}(2\sqrt{x}). \label{recursiveformulasforbessel}
\eea
By the above,
we have 
\bea\label{recursiveformula2}
I_{\beta+2}(2\sqrt{x})=I_{\beta}(2\sqrt{x})-\frac{\beta +1}{\sqrt{x}}I_{\beta+1}(2\sqrt{x}).
\eea
Let $h\geq 1$ and $Y=(l_{1},\ldots,l_{h})$ be a Young diagram with $l_{1}\geq \dots \geq l_{h}\geq 1$. For convenience of writing,  when $h<k$ we set $l_{j}=0$ for $h+1\leq j\leq k$.
Define $\tau_{k,Y}(x)$ as $D_{k}(1,2,\ldots,k;Y)$ in (\ref{definition of Dk}) with $a_{\beta}$  replaced by $I_{\beta}(2\sqrt{x})$. That is 
\bea\label{definition of tau}
\tau_{k,Y}(x): = \det(I_{i+j+l_{k-j}+1}(2\sqrt{x}))_{i,j=0,\ldots,k-1}.
\eea
This is a special case of $H_{k,\{X,Y\}}$ obtained by by setting $\alpha=1, X=\emptyset, a_\beta = I_\beta(2\sqrt{x})$ in Definition \ref{definition of shifted Hankel}. When $Y=\emptyset$ is an empty Young diagram, we denote $\tau_{k,\emptyset}$ as $\tau_k$ for simplicity.

\begin{prop}
\label{lemmaontranslationby1}
Let $k\geq 1,l\geq 0$ and $Y$ be a Young diagram of length $l$. Let $\tau_{k,Y}$ be given in (\ref{definition of tau}). Then
\bea
T_{1}\tau_{k,Y}(x) = \sqrt{x}\frac{d}{dx} \tau_{k,Y}(x) - \frac{k^2+l}{2\sqrt{x}}\tau_{k,Y},
\eea
In particular, if $Y=\emptyset$, then  
\bea
\tau_{k,(1)}(x) = \sqrt{x}\frac{d}{dx} \tau_{k}(x) - \frac{k^2}{2\sqrt{x}}\tau_{k} .
\eea
\end{prop}

\begin{proof}
Let $(F_{ij})_{i,j=0,\ldots,k-1}$ be the cofactor matrix of $(I_{i+j+1+l_{k-j}}(2\sqrt{x}))_{i,j=0,\ldots,k-1}$. Write $Y$ as $(l_{1},\ldots,l_{h})$ with $l_{1}\geq \ldots \geq l_{h}\geq 1$, and set $l_{j}=0$ for $h+1\leq j\leq k$ when $h<k$. 
Recall that for any two matrices $A,B$, we denote the inner product as $A \cdot B = {\rm Tr}(A^tB)$.
By the definition of $T_{1}$ (see Definition \ref{defn of translation operator}) and the recursive relation $(\ref{recursiveformulasforbessel})$,
\beas
T_{1}\tau_{k,Y}&=&(I_{i+j+2+l_{k-j}})_{i,j} \cdot (F_{ij})_{i,j} \\
&=& \sqrt{x}\Big(\frac{dI_{i+j+1+l_{k-j}}}{dx}\Big)_{i,j}
\cdot (F_{ij})_{i,j}
-\frac{1}{2\sqrt{x}} \Big((i+j+1+l_{k-j})I_{i+j+1+l_{k-j}}\Big)_{i,j}
\cdot (F_{ij})_{i,j}.
\eeas
Note that
\beas
&&\Big((i+j+1+l_{k-j})I_{i+j+1+l_{k-j}}\Big)_{i,j}
\cdot (F_{ij})_{i,j}\\
&=&\Big((j+1+l_{k-j})I_{i+j+1+l_{k-j}}\Big)_{i,j}
\cdot (F_{ij})_{i,j}+\Big(iI_{i+j+1+l_{k-j}}\Big)_{i,j}
\cdot (F_{ij})_{i,j}\\
&=&\Big(\frac{k^2+k}{2}+l+\frac{k^2-k}{2}\Big)\tau_{k,Y}=(k^2+l)\tau_{k,Y}.
\eeas
By the above, we have the claim in the proposition.
\end{proof}

We remark that the above technique of using the cofactor matrix to handle problems related to determinants is inspired by \cite[proof of Lemma 2.5]{kajiwara2001determinant} on giving determinantal formulae for general solutions of Toda equations, which are related to the theory of the Painlev\'{e} equations.

\begin{prop}
\label{initial values of tau}
Let $Y_{2,1}=(2)$ and $Y_{2,2}=(1,1)$ be Young diagrams. 
Then
\beas
\tau_{k,Y_{2,1}} &=& \frac{x}{2} \frac{d^2}{dx^2} \tau_k - \frac{k(k+2)}{2} \frac{d}{dx} \tau_k + \frac{k (k^3+4k^2+2k+4x)}{8x} \tau_k, \\
\tau_{k,Y_{2,2}} &=& \frac{x}{2} \frac{d^2}{dx^2} \tau_k - \frac{k(k-2)}{2} \frac{d}{dx} \tau_k + \frac{k (k^3-4k^2+2k-4x)}{8x} \tau_k.
\eeas
\end{prop}

\begin{proof}
Let $(F_{ij})_{i,j=0,\ldots,k-1}$ be the cofactor matrix of $(I_{i+j+1}(2\sqrt{x}))_{i,j=0,\ldots,k-1}$. Then
\beas
T_2 \tau_{k}(x) &=& (I_{i+j+3}(2\sqrt{x}))_{i,j} \cdot (F_{ij})_{i,j} \\
&=& (I_{i+j+1}(2\sqrt{x}))_{i,j} \cdot (F_{ij})_{i,j} - \frac{1}{\sqrt{x}} ((i+j+2)I_{i+j+2}(2\sqrt{x}))_{i,j} \cdot (F_{ij})_{i,j} \\
&=& k \tau_k - \frac{1}{\sqrt{x}} ((j+2)I_{i+j+2}(2\sqrt{x}))_{i,j} \cdot (F_{ij})_{i,j} - \frac{1}{\sqrt{x}} (i I_{i+j+2}(2\sqrt{x}))_{i,j} \cdot (F_{ij})_{i,j} \\
&=&  k\tau_k - \frac{2k}{\sqrt{x}} T_1 \tau_k.
\eeas
By Theorem \ref{gerenel linear systems} with $H_{k, \{X;Y\}} = \tau_{k,\{X;Y\}} $ and $X=\emptyset$, we obtain
\[
\begin{pmatrix}
\tau_{k,Y_{2,1}} \\
\tau_{k,Y_{2,2}}
\end{pmatrix} = 
\begin{pmatrix}
1/2 & 1/2 \\
1/2 & -1/2
\end{pmatrix}
\begin{pmatrix}
T_1 \tau_{k,Y_{1,1}} \\
T_2 \tau_{k}
\end{pmatrix},
\]
where $Y_{1,1}=(1)$.
By Proposition \ref{lemmaontranslationby1}, we obtain the claimed results.
\end{proof}

\begin{prop}
\label{recursive formula for tau0}
Let $k\geq 1$ and $l\geq 3$. Let $i,j$ with $1\leq i\leq l$ and $1\leq j\leq i$ be integers. Let $Y_{i,j}$ be a hook diagram. Let $B^{(l)}$, $C_1^{(l)}$ and $C_2^{(l)}$ be given in (\ref{definition of B}), (\ref{definitionofC1}) and (\ref{definitionofC2}), respectively.
Then
\beas
\begin{pmatrix}
\tau_{k,Y_{l,1}} \\
\vdots \\
\tau_{k,Y_{l,l}} 
\end{pmatrix}
&=&
-\, \sqrt{x} B^{(l)}
\begin{pmatrix}
\frac{d}{dx} \tau_{k,Y_{l-1,1}} \\
\vdots \\
\frac{d}{dx}\tau_{k,Y_{l-1,l-1}} \\
0
\end{pmatrix}
+\frac{k^2+l-1}{2\sqrt{x}} B^{(l)}\begin{pmatrix}
 \tau_{k,Y_{l-1,1}} \\
\vdots \\
\tau_{k,Y_{l-1,l-1}} \\
0
\end{pmatrix} \\
&& - \, \frac{1}{\sqrt{x}} C_1^{(l)}
\begin{pmatrix}
 \tau_{k,Y_{l-1,1}} \\
\vdots \\
\tau_{k,Y_{l-1,l-1}} 
\end{pmatrix}
+ C_2^{(l)}
\begin{pmatrix}
 \tau_{k,Y_{l-2,1}} \\
\vdots \\
\tau_{k,Y_{l-2,l-2}}
\end{pmatrix}.
\eeas
\end{prop}

\begin{proof}
Applying Theorem \ref{gerenel linear systems} with $H_{k, \{X;Y\}} = \tau_{k,(X;Y)}, X=\emptyset$ and recursive formula (\ref{recursiveformula2}),  we obtain
\bea
\begin{pmatrix}
\tau_{k,Y_{l,1}} \\
\vdots \\
\tau_{k,Y_{l,l}} 
\end{pmatrix}&=&
B^{(l)}\begin{pmatrix}
-T_1 \tau_{k,Y_{l-1,1}} \\
\vdots \\
-T_1 \tau_{k,Y_{l-1,l-1}} \\
0
\end{pmatrix}
-
\frac{1}{\sqrt{x}}B^{(l)}
\begin{pmatrix}
0 \\
\vdots\\
\sum_{h=2}^{j-1}(-1)^{h} S_{h-1} \tau_{k,Y_{l-h,j-h}} \\
\vdots \\
S_{l-1} \tau_k
\end{pmatrix}_{j=3,\ldots,l}\nonumber\\
&&+B^{(l)}
\begin{pmatrix}
0 \\
\vdots\\
\sum_{h=2}^{j-1}(-1)^{h} T_{h-2} \tau_{k,Y_{l-h,j-h}} \\
\vdots \\
T_{l-2} \tau_k
\end{pmatrix}_{j=3,\ldots,l}.
\label{sl-0}
\eea
Let $A^{(l-2)}=(a_{i,j})_{i,j=1\ldots l-2}$ be given in Lemma \ref{matrix} with $l$ replaced by $l-2$. By Theorem \ref{gerenel linear systems}, for $j=3,\ldots,l$,
\bea\label{Th-2}
\sum_{h=2}^{j-1} (-1)^h T_{h-2} \tau_{k,Y_{l-h,j-h}} &=& \sum_{h=0}^{j-3} (-1)^h T_{h} \tau_{k,Y_{l-2-h,j-2-h}} \nonumber\\
&=& k \tau_{k,Y_{l-2,j-2}} +  \sum_{h=1}^{j-3} (-1)^h T_{h} \tau_{k,Y_{l-2-h,j-2-h}}\nonumber \\
&=& k \tau_{k,Y_{l-2,j-2}} +  \sum_{q=1}^{l-2} a_{j-3,q} \tau_{k,Y_{l-2,q}}\nonumber \\
&=& \sum_{q=1}^{j-3} (-1)^{j-q} \tau_{k,Y_{l-2,q}} + (k-j+3) \tau_{k,Y_{l-2,j-2}}.
\eea
By Propositions \ref{513prop5} and \ref{513prop6}, we have 
\bea\label{sl-1}
S_{l-1} \tau_k =  \sum_{i=1}^{l-1} (-1)^{i-1} (2k-2i+l) \tau_{k,Y_{l-1,i}}, \quad 
T_{l-2}\tau_k = \sum_{i=1}^{l-2} (-1)^{i-1}\tau_{k,Y_{l-2,i}}.
\eea
Combining formulae (\ref{sl-0})-(\ref{sl-1}) and  Propositions \ref{411proponS} and \ref{lemmaontranslationby1}, we obtain the claimed result.
\end{proof}

The recursive formula in Proposition \ref{recursive formula for tau0} is for hook diagrams. This includes the special Young diagram that only has 1 row. We then can use induction on the number of rows to deduce $\tau_{k,Y}$ recursively for a general Young diagram $Y$. We analyze this briefly below, because, while it is not the main focus of this paper, it may be of independent interest in itself. Suppose we already know $\tau_{k,Y}$ for any $Y=(l_1,\ldots,l_s)$ with $s\geq 1$ rows. First, we show how to deduce $\tau_{k,(l_1,\ldots,l_s,1)}(x)$.
From the definition of $T_1$ (see Definition \ref{defn of translation operator}), we have
\[
T_1 \tau_{k,(l_1,\ldots,l_s)}(x)
=
\tau_{k,(l_1,\ldots,l_s,1)}(x) + \sum_{i=1}^s \tau_{k,(l_1,\ldots,l_i+1,\ldots,l_s)}(x) .
\]
By Proposition \ref{lemmaontranslationby1}, the left hand side is
\[
\sqrt{x} \frac{d \tau_{k,(l_1,\ldots,l_s)}(x)}{dx} - \frac{k^2+l_1+\cdots+l_s}{2\sqrt{x} } \tau_{k,(l_1,\ldots,l_s)}(x).
\]
From the above, we obtain $\tau_{k,(l_1,\ldots,l_s,1)}(x)$. By a similar argument to the proof of Theorem \ref{gerenel linear systems}, we have
\beas
\begin{pmatrix}
\tau_{k,(l_1,\ldots,l_s,Y_{l,1})} \\
\vdots \\
\vdots \\
\tau_{k,(l_1,\ldots,l_s,Y_{l,l})} \\
\end{pmatrix}
&=& B^{(l)}
\begin{pmatrix}
\vdots \\
\sum_{h=1}^{j-1} (-1)^h T_{h} \tau_{k,(l_1,\ldots,l_s,Y_{l-h,j-h})} \\
\vdots \\
T_{l}\tau_{k,(l_1,\ldots,l_s)} \\
\end{pmatrix}_{j=2,\ldots,l} \\
&& -\, B^{(l)}
\begin{pmatrix}
\vdots \\
\sum_{h=1}^{j-1} (-1)^h 
\sum_{i=1}^s \tau_{k,(l_1,\ldots,l_i+h,\ldots,l_s,Y_{l-h,j-h})} \\
\vdots  \\
\sum_{i=1}^s \tau_{k,(l_1,\ldots,l_i+l,\ldots,l_s)}
\end{pmatrix}_{j=2,\ldots,l}.
\eeas
By the inductive assumption on $\tau_{k,(l_1',\ldots,l_s',Y_{i,j})}$ with $i\leq l-1$, we can obtain the second part above. Regarding the first part, we can use a similar method to that of Proposition \ref{recursive formula for tau0}.

\section{Recursive formulae for $G_{k,Y (x,t)}$}
\label{LiftrfHsY}

In this section, we establish recursive formulae for determinants of Hankel matrices whose entries involve the generating function of I-Bessel function shifted by Young diagrams.

Let $\beta$ be a complex number, and
define
\bea\label{definition of g}
g_\beta(x,t) = \sum_{n=0}^\infty \frac{t^n}{n!} I_{2n+\beta}(2\sqrt{x}).
\eea
Let $h\geq 1$ and $Y=(l_{1},\ldots,l_{h})$ be a Young diagram with $l_{1}\geq \dots \geq l_{h}\geq 1$. As stated previously, when $h<k$, we set $l_{j}=0$ for $h+1\leq j\leq k$.
Define $G_{k,Y}(x,t)$ as
\bea\label{definition of G}
G_{k,Y}(x,t)=\det(g_{i+j+1+l_{k-j}}(x,t))_{i,j=0,\ldots,k-1}.
\eea
This is a special case of $H_{k,\{X,Y\}}$ obtained by setting $\alpha=1, X=\emptyset, a_\beta = g_\beta(x,t)$ in Definition \ref{definition of shifted Hankel}. In particular, when $t=0$, $G_{k,Y}=\tau_{k,Y}$ is defined in (\ref{definition of tau}).
When $Y=\emptyset$ is an empty Young diagram, we denote it as $G_k(x,t)$ for simplicity.
Namely,
\bea\label{definition of G1}
G_{k}(x,t)=\det(g_{i+j+1}(x,t))_{i,j=0,\ldots,k-1}.
\eea
For $i\geq 0$, define
\bea\label{definition of F}
F_{i}(x,t)=\frac{\partial^i}{\partial t^i}{G_{k}}(x,t).
\eea
By the recursive relations (\ref{recursiveformulasforbessel}) and (\ref{recursiveformula2}) satisfied by $I_{\beta}(x)$, 
we have recursive relations for $g_{\beta}(x,t)$ as follows.
\bea
\frac{\partial {g_{\beta}}}{\partial{x}} &=& \frac{1}{\sqrt{x}}g_{\beta+1}+\frac{\beta}{2x}g_{\beta}+\frac{t}{x}g_{\beta+2},
\label{412recursive formula for gbeta1} \\
\frac{\partial {g_{\beta}}}{\partial{x}} &=& \frac{1}{\sqrt{x}}g_{\beta-1}-\frac{\beta}{2x}g_{\beta}-\frac{t}{x}g_{\beta+2}, \label{412recursive formula for gbeta2} \\
g_{\beta+2} &=&g_{\beta}-\frac{\beta+1}{\sqrt{x}}g_{\beta+1}-\frac{2t}{\sqrt{x}}g_{\beta+3}.
\label{412recursive formula for gbeta}
\eea

The following is a fact obout derivatives of determinants.

\begin{lem}\label{NOTE4lem2}
Let $s\geq 0$, $k\geq 1$ be integers and $p_{i,j}(t)$ be $s$-times differential functions of $t$. Then
\[
\Big(\frac{d}{dt}\Big)^{s}\det(p_{i,j}(t))_{i,j=1,\ldots,k}=\sum_{\substack{l_1+\cdots+l_k=s \\ l_1\geq 0,\ldots,l_k\geq 0}}\binom{s}{l_1,\ldots,l_k}\det \Big(p_{i,j}^{(l_{i})}(t)\Big)_{i,j=1,\ldots,k},
\]
where $p_{i,j}^{(l_{i})}(t)$ means that we take the $l_{i}$-th derivative of $p_{i,j}(t)$.
\end{lem}

\begin{prop}
\label{derivative and translations}
Let $G_{k,Y}(x,t)$ be given in 
(\ref{definition of G}), then we have $\frac{\partial G_{k,Y}}{\partial t}=T_{2}G_{k,Y}.$
\end{prop}

\begin{proof}
Note that for any $s_j \geq 0$,
\beas
\frac{\partial^{s_j} g_\beta(x,t)}{\partial t^{s_j}}
&=& \sum_{n=s_j}^\infty \frac{t^{n-s_j}}{(n-s_j)!} I_{2n+\beta}(2\sqrt{x}) \\
&=& \sum_{n=0}^\infty \frac{t^{n}}{n!} I_{2n+\beta+2s_j}(2\sqrt{x}) \\
&=& g_{\beta+2s_j}(2\sqrt{x}).
\eeas
So by Lemma \ref{NOTE4lem2}, we have
\bea
\frac{\partial G_{k,Y}}{\partial t}
&=& \sum_{\substack{s_1+\cdots+s_k=1 \\s_1,\ldots,s_k \geq 0}} \det\left(\frac{\partial ^{s_i}}{\partial t^{s_i}} g_{i+j+1+l_{k-j}}(x,t)\right)_{i,j=0,\ldots,k-1} \nonumber \\
&=& \sum_{\substack{s_1+\cdots+s_k=1 \\s_1,\ldots,s_k \geq 0}} \det\left(g_{i+2s_i+j+1+l_{k-j}}(x,t)\right)_{i,j=0,\ldots,k-1}.
\label{eq1}
\eea
The summation above is over the set $\{(s_{1},\ldots,s_{k}):s_{i}=1, s_{j}=0~\text{for}~j\neq i,i=1,\ldots,k\}$. Due to the fact that the determinant is zero if there are two common columns or rows, the above determinant is nonzero if and only if $s_{k-1}=1$ or $s_k=1$. So (\ref{eq1}) equals
$G_{k,\{Y_{2,1};Y\}}-G_{k,\{Y_{2,2};Y\}}$, which is $T_{2}G_{k,Y}$ by Proposition \ref{513prop5}.
\end{proof}

\begin{prop}\label{412320prop3}
Let $Y$ be a Young diagram of length $m$, and let $G_{k,Y}$ be defined in (\ref{definition of G}). Then
\be \label{412320prop3:eq}
T_{1}G_{k,Y}=\sqrt{x}\frac{\partial G_{k,Y}}{\partial x}-\frac{k^2+m}{2\sqrt{x}}G_{k,Y}-\frac{t}{\sqrt{x}}\frac{\partial G_{k,Y}}{\partial t}.
\ee
\end{prop}

\begin{proof}
By a similar argument to the proof of Proposition \ref{lemmaontranslationby1},
we have
\[
T_{1}G_{k,Y}=\sqrt{x}\frac{\partial G_{k,Y}}{\partial x}-\frac{k^2+m}{2\sqrt{x}}G_{k,Y}-\frac{t}{\sqrt{x}} T_2 G_{k,Y}.
\]
The claim now follows from Proposition \ref{derivative and translations}.
\end{proof}

By recursive relation (\ref{412recursive formula for gbeta}) for $g_{\beta }(x,t)$, we have the following proposition. 

\begin{prop}\label{recursive formula for G about T}
Let $T_{i},S_{j}$ be the operators given in Definition \ref{defn of translation operator}, then
\[
T_{h}G_{k,Y} =T_{h-2} G_{k,Y} - \frac{S_{h-1} G_{k,Y}}{\sqrt{x}} - \frac{2t}{\sqrt{x}} T_{h+1} G_{k,Y}.
\]
\end{prop}

\begin{prop}
\label{initial values of G}
Let $Y_{1,1}=(1)$, $Y_{2,1}=(2)$ and $Y_{2,2}=(1,1)$ be Young diagrams. Let $G_{k,Y}, G_{k}$ and $F_{i}$ be defined as in (\ref{definition of G}), (\ref{definition of G1}) and (\ref{definition of F}), respectively. Then
\[
G_{k,Y_{1,1}}=\sqrt{x} \frac{\partial G_{k}}{\partial x}- \frac{1}{2\sqrt{x}} k^2G_{k} - \frac{t}{\sqrt{x}}F_{1},
\]
and
\beas
G_{k,Y_{2,1}} &=& \frac{kG_k}{2}- \frac{t}{\sqrt{x}} (G_{k,Y_{3,1}} - G_{k,Y_{3,2}} +G_{k,Y_{3,3}}  )  \\
&& 
\hspace{-.3cm}
+\, \frac{1}{2}
\Bigg(x \frac{\partial^2G_{k}}{\partial x^2}-(k^2+2k)\frac{\partial G_{k}}{\partial x}+\frac{k^4+4k^3+2k^2}{4x}G_{k}-2t\frac{\partial F_{1}}{\partial x}+\frac{k^2+2k+2}{x}tF_{1}+\frac{t^2}{x}F_{2}\Bigg),
\\
G_{k,Y_{2,2}} &=& -\frac{kG_k}{2} + \frac{t}{\sqrt{x}} (G_{k,Y_{3,1}} - G_{k,Y_{3,2}} +G_{k,Y_{3,3}}  )  \\
&& 
\hspace{-.3cm}
+\, \frac{1}{2}
\Bigg(x \frac{\partial^2G_{k}}{\partial x^2}-(k^2-2k)\frac{\partial G_{k}}{\partial x}+\frac{k^4-4k^3+2k^2}{4x}G_{k}-2t\frac{\partial F_{1}}{\partial x}+\frac{k^2-2k+2}{x}tF_{1}+\frac{t^2}{x}F_{2}\Bigg),
\eeas
where $Y_{3,1}=(3), Y_{3,2}=(2,1), Y_{3,3}=(1,1,1)$.
\end{prop}
\begin{proof}
The first claim comes from the fact that $G_{k,Y_{1,1}}=T_{1}G_{k}$ and Proposition \ref{412320prop3} with $Y=\emptyset$. For the second and third claims, using Theorem \ref{gerenel linear systems}, we have
\beas
\begin{pmatrix} \vspace{.12cm}
G_{k,Y_{2,1}} \\
G_{k,Y_{2,2}} \\
\end{pmatrix}
&=& \begin{pmatrix}
1/2 & 1/2 \\
1/2 & -1/2
\end{pmatrix}
\begin{pmatrix} \vspace{.12cm}
\displaystyle T_{1}G_{k,Y_{1,1}} \\
\displaystyle T_{2}G_{k} \\
\end{pmatrix}
\eeas
For $T_{1}G_{k,Y_{1,1}}$, we apply Proposition \ref{412320prop3} and the first claim of this proposition. For $T_{2}G_{k}$, we use Proposition \ref{recursive formula for G about T} to obtain
\beas
T_{2}G_{k} &=& k G_{k} - \frac{S_1 G_{k}}{\sqrt{x}} - \frac{2t}{\sqrt{x}} T_{3} G_{k} \\
&=& k G_{k} - \frac{2k G_{k,Y_{1,1}}}{\sqrt{x}} - \frac{2t}{\sqrt{x}} (G_{k,Y_{3,1}} - G_{k,Y_{3,2}} +G_{k,Y_{3,3}}) \\
&=& kG_k - 2k \frac{\partial G_k}{\partial x} + \frac{k^3}{x} G_k + \frac{2kt}{x} \frac{\partial G_k}{\partial t}  - \frac{2t}{\sqrt{x}} (G_{k,Y_{3,1}} - G_{k,Y_{3,2}} +G_{k,Y_{3,3}}) .
\eeas
The second equality comes from Propositions \ref{513prop5} and \ref{513prop6}. The last equality comes from the first claim of this proposition. Putting all of this together, we obtain the last two claims in the proposition.
\end{proof}

In the above proof, if we replace $T_2G_k$ by $\frac{\partial G_k}{\partial t}$ by Proposition \ref{derivative and translations}, we obtain
\bea
G_{k,Y_{2,1}}&=&\frac{F_{1}}{2}+\frac{1}{2}
\Bigg(x \frac{d^2G_{k}}{dx^2}-k^2\frac{dG_{k}}{dx}+\frac{2k^2+k^4}{4x}G_{k}-2t\frac{dF_{1}}{dx}+\frac{k^2+2}{x}tF_{1}+\frac{t^2}{x}F_{2}\Bigg), \nonumber \\
G_{k,Y_{2,2}}&=&-\frac{F_{1}}{2}+\frac{1}{2}
\Bigg(x \frac{d^2G_{k}}{dx^2}-k^2\frac{dG_{k}}{dx}+\frac{2k^2+k^4}{4x}G_{k}-2t\frac{dF_{1}}{dx}+\frac{k^2+2}{x}tF_{1}+\frac{t^2}{x}F_{2}\Bigg),
\label{recursive formula for Gk}
\eea
which plays an important role in the proof of Theorem \ref{theorem2 in note 4}.

\begin{lem}\label{F10}
Let $F_{1}(x,t)$ be given as in (\ref{definition of F}) and $\tau_{k,Y}(x)$ be given as in (\ref{definition of tau}). Then 
\[F_{1}(x,0)=-2k\frac{d\tau_{k}(x)}{dx}+(k+\frac{k^3}{x})\tau_{k}(x).\]
\end{lem}

\begin{proof}
By Propositions \ref{derivative and translations} and \ref{recursive formula for G about T},
\[F_{1}(x,t)=T_{2}G_{k}(x,t)=T_{0}G_{k}-\frac{S_{1}G_{k}}{\sqrt{x}}-\frac{2t}{\sqrt{x}}T_{3}G_{k}.\]
So
\[F_{1}(x,0)=k\tau_{k}(x)-\frac{2kT_{1}\tau_{k}(x)}{\sqrt{x}}.\]
By Proposition \ref{lemmaontranslationby1},
\[T_{1}\tau_{k}=\sqrt{x}\frac{d\tau_{k}}{dx}-\frac{k^2}{2\sqrt{x}}\tau_{k},\]
so we have the claimed result for $F_{1}(x,0)$.
\end{proof}

% Setting $t=0$ in Proposition \ref{initial values of G}, and by Lemma \ref{F10}, we have the following result for $\tau_{k}(x)$.
% \begin{prop}
% Let $k\geq 1$ and $Y$ be a Young diagram of length $m$. Let $\tau_{k,Y}$ be given in (\ref{definition of tau}). Let $Y_{1,1}=(1)$, $Y_{2,1}=(2)$ and $Y_{2,2}=(1,1)$ be Young diagrams.
% Then 
% \beas
% \tau_{k,Y_{2,1}}(x)&=&\frac{x}{2}\frac{d^2\tau_{k}}{dx^2}-\frac{k(k+2)}{2}\frac{d\tau_{k}}{dx}
%  +\frac{k(k^3+4k^2+2k+4x)}{8x}\tau_{k},\\
% \tau_{k,Y_{2,1}}(x)&=&\frac{x}{2}\frac{d^2\tau_{k}}{dx^2}-\frac{k(k-2)}{2}\frac{d\tau_{k}}{dx}
%  +\frac{k(k^3-4k^2+2k-4x)}{8x}\tau_{k}.
% \eeas
% \end{prop}

% \begin{lem}\label{F20}
% Let $F_{2}(x,t)$ be given in (\ref{definition of F}) and $\tau_{k,Y}(x)$ be given in (\ref{definition of tau}). Then 
% \bea\label{F20formula}
% F_{2}(x,0)&=&(4k^2+2)\frac{d^2\tau_{k}}{dx^2}
% -\frac{4}{x}(k^4+(x+\frac{7}{2})k^2+\frac{x}{2})\frac{d\tau_{k}}{dx}\nonumber\\
% &&+\frac{k^2\tau_{k}}{x^2}\Big(
% k^4+(2x+\frac{17}{2})k^2+x^2+7x+1\Big).
% \eea
% \end{lem}

Next we deduce a recursive formula for $G_{k,Y_{l,i}}(x,t)$ for $l\geq 3$ and $i=1,\ldots,l$.

\begin{prop}
\label{prop:iteraive formula 2}
Let $k\geq 1$ and $l\geq 3$.  Let $i,j$ with $1\leq i\leq l$ and $1\leq j\leq i$ be integers. Let $Y_{i,j}$ be a hook diagram. Let $B^{(l)}$, $C_1^{(l)}$, $C_2^{(l)}$ and $C_3^{(l)}$ be given in (\ref{definition of B}), (\ref{definitionofC1}),  (\ref{definitionofC2}) and (\ref{definition of C3}), respectively. 
Then
\bea\label{recursive formula for G_{k}}
\begin{pmatrix}
G_{k,Y_{l,1}} \\
\vdots \\
G_{k,Y_{l,l}} 
\end{pmatrix}
&=&
-B^{(l)}(\sqrt{x}\frac{\partial}{\partial x}-\frac{k^2+l-1}{2\sqrt{x}}+\frac{t}{\sqrt{x}}\frac{\partial}{\partial t})
\begin{pmatrix}
G_{k,Y_{l-1,1}} \\
\vdots \\
G_{k,Y_{l-1,l-1}} \\
0
\end{pmatrix} \nonumber\\
&&
-\frac{1}{\sqrt{x}}C_1^{(l)}
\begin{pmatrix}
 G_{k,Y_{l-1,1}} \\
\vdots \\
G_{k,Y_{l-1,l-1}} 
\end{pmatrix} 
+ C_2^{(l)}
\begin{pmatrix}
 G_{k,Y_{l-2,1}} \\
\vdots \\
G_{k,Y_{l-2,l-2}}
\end{pmatrix}
+\frac{2t}{\sqrt{x}}C_3^{(l)}
\begin{pmatrix}
 G_{k,Y_{l+1,1}} \\
\vdots \\
G_{k,Y_{l+1,l+1}}
\end{pmatrix}\nonumber \\
&&+B^{(l)}(2t\frac{\partial}{\partial x}-\frac{t}{x}(k^2+l)-\frac{2t^2}{x}\frac{\partial}{\partial t})
\begin{pmatrix}
G_{k,Y_{l,2}} \\
\vdots \\
G_{k,Y_{l,l}} \\
0
\end{pmatrix}.
\eea
\end{prop}
\begin{proof}
By Theorem \ref{gerenel linear systems} and Proposition \ref{recursive formula for G about T},
\beas
A^{(l)} \begin{pmatrix}
G_{k,Y_{l,1}} \\
G_{k,Y_{l,2}} \\
\vdots \\
G_{k,Y_{l,l}} \\
\end{pmatrix} 
&=&
\begin{pmatrix}
-T_1 G_{k,Y_{l-1,1}} \\
-T_1 G_{k,Y_{l-1,2}} \\
\vdots  \\
-T_1 G_{k,Y_{l-1,l-1}} \\
0 
\end{pmatrix}
+\begin{pmatrix}
0 \\
\vdots\\
\sum_{h=2}^{j-1} (-1)^h T_{h-2} G_{k,Y_{l-h,j-h}} \\
\vdots  \\
T_{l-2} G_{k}
\end{pmatrix}_{j=3,\ldots,l}  \\ 
&& - \frac{1}{\sqrt{x}} 
\begin{pmatrix}
0 \\
\vdots\\
\sum_{h=2}^{j-1} (-1)^h S_{h-1} G_{k,Y_{l-h,j-h}} \\
\vdots  \\
S_{l-1} G_{k}
\end{pmatrix}_{j=3,\ldots,l} \\
&&- \frac{2t}{\sqrt{x}} \begin{pmatrix}
0 \\
\vdots\\
\sum_{h=2}^{j-1} (-1)^h T_{h+1} G_{k,Y_{l-h,j-h}} \\
\vdots  \\
T_{l+1} G_{k}
\end{pmatrix}_{j=3,\ldots,l}\\
&=:&J_{1}+J_{2}+J_{3}+J_{4}
\eeas
For the terms $J_{1},J_{2},J_{3}$, we use a similar argument to that of 
Proposition \ref{recursive formula for tau0}. Now we handle $J_{4}$. By adding and subtracting a common term,
\bea\label{J4}
J_{4}&=& \frac{2t}{\sqrt{x}} 
\begin{pmatrix}
\vdots\\
\sum_{h=1}^{j} (-1)^h T_{h} G_{k,Y_{l+1-h,j+1-h}} \\
\vdots  \\
-T_{l+1} G_{k}
\end{pmatrix}_{j=2,\ldots,l} \nonumber\\
&& 
+ \, \frac{2t}{\sqrt{x}} \begin{pmatrix}
T_1 G_{k, Y_{l,2}} \\
T_{1} G_{k,Y_{l,3}} \\
\vdots  \\
T_{1} G_{k,Y_{l,l}} \\
0
\end{pmatrix} 
-  \frac{2t}{\sqrt{x}} \begin{pmatrix}
T_2 G_{k, Y_{l-1,1}} \\
T_{2} G_{k,Y_{l-1,2}} \\
\vdots  \\
T_{2} G_{k,Y_{l-1,l-1}} \\
0
\end{pmatrix} 
\eea
For the first term in $J_{4}$, by Theorem \ref{gerenel linear systems}, for $j=2,\ldots,l$,
\bea\label{412ev1}
\sum_{h=1}^{j} (-1)^{h} T_h G_{k, Y_{l+1-h,j+1-h}}=\sum_{q=1}^{l+1} (A^{(l+1)})_{j,q} G_{k,Y_{l+1,q}},
\eea
where matrix $A^{(l+1)}$ is defined by (\ref{definition of A}) with $l$ replaced by $l+1$.
By Proposition \ref{513prop5},
\bea\label{412ev4}
T_{l+1}G_{k} = \sum_{q=1}^{l+1} (-1)^{q-1}G_{k,Y_{l+1,q}}.
\eea
By (\ref{412ev1})-(\ref{412ev4}), we have
\beas
B^{(l)}\begin{pmatrix}
\vdots\\
\sum_{h=1}^{j} (-1)^h T_{h} G_{k,Y_{l+1-h,j+1-h}} \\
\vdots  \\
-T_{l+1} G_{k}
\end{pmatrix}_{j=2,\ldots,l}
=
C_3^{(l)}
\begin{pmatrix}
 G_{k,Y_{l+1,1}} \\
\vdots \\
G_{k,Y_{l+1,l+1}}
\end{pmatrix}.
\eeas
We apply Propositions \ref{derivative and translations} and \ref{412320prop3}  to the second and the third terms in  (\ref{J4}), respectively. Putting everything together, we obtain the claim in the proposition.
\end{proof}

\begin{lem}\label{413lemma1}
Let $G_{k}(x,t)$ be given in  (\ref{definition of G1}). Then for any integer $l\geq 1$, we have
\bea
\frac{\partial G_{k}}{\partial t}(x,t)=\sum_{i=0}^{l-1}(-1)^{i+1}d_{i}(\frac{2t}{\sqrt{x}})^{i}+d_{l}(\frac{2t}{\sqrt{x}})^{l},
\label{413lemma1:eq}
\eea
where $d_{i}=\frac{1}{\sqrt{x}}S_{i+1}G_{k} -T_{i}G_{k}$ for $i=0,\ldots,l-1$ and $d_{l}=(-1)^{l}T_{l+2}G_{k}$.
\end{lem}

\begin{proof}
By Proposition \ref{derivative and translations}, and repeatedly using Proposition \ref{recursive formula for G about T},
\beas
\frac{\partial G_{k}}{\partial t}(x,t)&=&T_{2}G_{k}=k G_{k} - \frac{S_1 G_{k}}{\sqrt{x}} - \frac{2t}{\sqrt{x}} (T_{1}G_{k} - \frac{1}{\sqrt{x}} S_2 G_{k} - \frac{2t}{\sqrt{x}} T_{4}G_{k}) \\
&=& k G_{k} - \frac{S_1 G_{k}}{\sqrt{x}} - \frac{2t}{\sqrt{x}} (T_{1}G_{k} - \frac{1}{\sqrt{x}} S_2 G_{k})\\
&&
+ \left( \frac{2t}{\sqrt{x}} \right)^2 (T_{2}G_{k} - \frac{1}{\sqrt{x}} S_3 G_{k} - \frac{2t}{\sqrt{x}} T_{5}G_{k}) \\
&=& k G_{k} - \frac{S_1 G_{k}}{\sqrt{x}} - \frac{2t}{\sqrt{x}} (T_{1}G_{k} - \frac{1}{\sqrt{x}} S_2 G_{k})\\
&&+ \sum_{i=2}^{l-1} (-1)^{i+1} \left( \frac{2t}{\sqrt{x}} \right)^{i} (\frac{S_{i+1}G_{k}}{\sqrt{x}}-T_{i}G_{k}) + (-1)^l \left( \frac{2t}{\sqrt{x}} \right)^l T_{l+2} G_{k}.
\eeas
This completes the proof.
\end{proof}

\begin{rem}
\label{newremark}
{\rm 
Lemma \ref{413lemma1} provides an expansion of $T_2G_k$ in powers of $t$ up to degree $l$ with coefficients related to $S_{i+1} G_k, T_iG_k$. This also holds for $T_2G_{k,Y}$ with $G_k$ replaced by $G_{k,Y}$ in (\ref{413lemma1:eq}). Moreover, via a similar process to 
further handle $T_iG_{k,Y}$, we can write 
\be
T_2G_{k,Y}= \sum_{i=0}^{l-1} \tilde{d_i} (\frac{2t}{\sqrt{x}})^i + \tilde{d_l} (\frac{2t}{\sqrt{x}})^l,
\label{413lemma1:eq-new}
\ee
where for $0\leq i \leq l-1$, $\tilde{d_i}$ is some linear combination of $S_j G_{k,Y}, j=0,1,\ldots,i+1$.
Taking the derivative with respect to $t$ on both sides of (\ref{413lemma1:eq-new}) produces $\frac{\partial }{\partial t}S_j G_{k,Y}$. By Proposition \ref{320eveprop6}, $S_j G_{k,Y}$ is a linear combination of $G_{k,Y_s}$ for some $Y_s$ with $|Y_s| = |Y|+j$. So by Proposition \ref{derivative and translations}, $\frac{\partial }{\partial t}S_j G_{k,Y} = T_2(S_j G_{k,Y} )$. 
By a similar argument to (\ref{413lemma1:eq-new}), one can further expand $T_2(S_j G_{k,Y} )$ as powers of $t$ up to degree $l-j$. Continuing the above process (i.e., taking derivatives, expanding as powers of $t$), we can show that $f_l$ (which equals $\frac{\partial^{l-1} T_2G_k}{\partial t^{l-1}} |_{t=0}$ by (\ref{aboutfi}) below) is a linear combination of $\tau_{k,Y}$ with $|Y|=l$.
Compared with (\ref{410definition of f_{l}}) where $f_l$ is a linear combination of $\tau_{k,Y}$ with $|Y|=2l$, now $f_l$ can be expressed as a linear combination of $\tau_{k,Y}$ with $|Y|=l$.
}
\end{rem}

\section{Proofs of Theorems \ref{theorem2 in note 4}, \ref{theorem 1 in note 4} and Proposition \ref{intro:prop}}\label{proof of main results}

\begin{proof}[Proof of Theorem \ref{theorem 1 in note 4}]
By Lemma \ref{NOTE4lem2} and Proposition \ref{derivative and translations},
\bea\label{aboutfi}
f_i(x)=F_{i}(x,0)=\frac{\partial^{i-1}T_2G_{k}}{\partial t^{i-1}}\Big|_{t=0}.
\eea
Let $1\leq j\leq i$, and let $Y_{i,j}$ be a hook diagram.
Let
\bea\label{definition of fjqi}
f_{j,q}^{(i)} = \frac{\partial^i}{\partial t^{i}}G_{k,Y_{j,q}}\Big|_{t=0}.
\eea
By Proposition \ref{derivative and translations}, (\ref{aboutfi}) and (\ref{definition of fjqi}), we can write $f_{i+1}(x)$ as in (\ref{expression of fi}).
Taking the $i$-th derivative with respect to $t$ on both sides of equation (\ref{recursive formula for G_{k}}) and letting $t=0$, we have the recursive relation (\ref{recursive formula for fjqi}) for $f_{j,q}^{(i)}$.
By Proposition \ref{initial values of G}, we have the initial conditions for the recursive formula (\ref{recursive formula for fjqi})
as in (\ref{initial values for the recursive formulas}).
\end{proof}

\begin{proof}[Proof of Theorem \ref{theorem2 in note 4}]
Using the recursive relations in Theorem \ref{theorem 1 in note 4}, in the following we apply the induction method to show that for any $i\geq 0, l \geq 1$ and  $1\leq q\leq l$,
\bea
f_{l,q}^{(i)}(x) &=& x^{-i-l/2} \sum_{s=0}^{l} \frac{d^s \tau_k}{dx^s} x^s \sum_{j=0}^{i+\lfloor\frac{l-s}{2}\rfloor} a_{j,l,q,s}^{(i)} (k) x^j 
+ x^{-i-l/2}  \sum_{s=l+1}^{i+l} \frac{d^s \tau_k}{dx^s} x^s \sum_{j=0}^{i+l-s} b_{j,l,q,s}^{(i)}(k) x^j , \label{exforflq} \\
f_{i+1}(x)&=&\frac{1}{x^{i+1}} \sum_{m=0}^{i+1} \frac{d^m \tau_k}{dx^m} x^m\sum_{j=0}^{i+1-m} c_{j,m}^{(i+1)} (k) x^j,
\label{expression for fl}
\eea
where $a_{j,l,q,s}^{(i)}(k)$ and $b_{j,l,q,s}^{(i)}(k)$ are polynomials in $k$ of degree at most $i+2(i+l-s-j)$ with coefficients depending on $j, l, q, s, i$, and $c_{j,m}^{(i+1)}(k)$ are polynomials in $k$ of degree at most $i+1+2(i+1-m-j)$ with coefficients depending on $j, i, m$.

By Lemma \ref{F10}, $f_1(x) = -2k \frac{d \tau_k}{dx} + (k+ \frac{k^3}{x}) \tau_k$. We have that (\ref{expression for fl}) holds for $i=1$. Assume inductively that, for any $i_{0}\leq i$ with some $i\geq 0$,  formulae (\ref{exforflq}) and (\ref{expression for fl}) hold for $f_{i_{0}+1}(x)$ $f_{l,q}^{(i_{0}-1)}(x)$ for any $l \geq 1$, $1\leq q\leq l$. In the above expression, when $i_{0}= 0$, we set $f_{l,q}^{(i_{0}-1)}(x)=0$.
In the following, we first deduce the formula for $f_{l,q}^{(i)}(x)$ for any $l\geq 3$ and $1\leq q\leq l$, and then deduce the formula for $f_{i+2}(x)$.

For a given $i$, one can use induction on $l$ to prove that for any $l\geq 3$ and $1\leq q\leq l$, $f_{l,q}^{(i)}(x)$ can be expressed as (\ref{exforflq}). Using $f^{(i)}_{1,1}$ in (\ref{initial values for the recursive formulas}), we have that (\ref{exforflq}) holds for $l=1$ by induction. In (\ref{recursive formula for Gk}), we take the $i$-th derivative with respect to $t$ at $t=0$, then
\beas
f_{2,1}^{(i)}(x)&=&\frac{1}{2}f_{i+1}+
\frac{1}{2}\Big(x \frac{d^2}{dx^2} f_i - (k^2+2i) \frac{d}{dx} f_i + \frac{2k^2+k^4+4k^2i+4i+4i^2}{4x} f_i\Big) , \nonumber \\
f_{2,2}^{(i)}(x)&=&-\frac{1}{2}f_{i+1}+
\frac{1}{2}\Big(x \frac{d^2}{dx^2} f_i - (k^2+2i) \frac{d}{dx} f_i + \frac{2k^2+k^4+4k^2i+4i+4i^2}{4x} f_i\Big).
\label{old recursive formula}
\eeas
So by induction, (\ref{exforflq}) holds for $l=2$.
Assume inductively it holds for all $l_{0}\leq l-1$ for some $l\geq 3$. Then computing $\frac{d}{dx} f_{l-1,q}^{(i)}$, $\frac{d}{dx} f_{l,q}^{(i-1)}$ in (\ref{recursive formula for fjqi}), we have that (\ref{exforflq}) holds for all $l$. 

Next, we use this information to show that $f_{i+2}(x)$ can be expressed in the form of (\ref{expression for fl}). 
By induction, we now have $f_{i+1}$ in the form of (\ref{expression for fl}). By (\ref{initial values for the recursive formulas}), we obtain $f^{(i+1)}_{1,1}$.
Thus,
\beas\label{intermedian}
kf_{i+1}(x)-\frac{2k}{\sqrt{x}}f_{1,1}^{(i+1)}(x)&=&x^{-i-2}\sum_{s=0}^{i+2} \frac{d^s \tau_k}{dx^s} x^s \sum_{j=0}^{i+2-s} \tilde{a}_{j,i,s}(k) x^j,
\eeas
where $\tilde{a}_{j,i,s}(k)$ are polynomials in $k$ with degree at most $i+2+2(i+2-s-j)$ with coefficients depending on $j,i,s$. Moreover, $\tilde{a}_{i+2,i,0}(k)=kc^{(i+1)}_{i+1,0}(k)$.

By Propositions \ref{513prop5} and \ref{513prop6}, we obtain
\be
\label{5.19}
\frac{S_{j+1}G_{k}}{\sqrt{x}}-T_{j}G_{k}=
\frac{1}{\sqrt{x}}\sum_{q=1}^{ j+1}(-1)^{q-1}(2k-2q+j+2)G_{k,Y_{j+1,q}}-
\sum_{q=1}^{j}(-1)^{q-1}G_{k,Y_{j,q}}.
\ee
Substituting (\ref{5.19}) into (\ref{413lemma1:eq}) and taking the $(i+1)$-th derivative with respect to $t$ at $t=0$ on both sides, recalling (\ref{definition of fjqi}) and (\ref{aboutfi}), we obtain
\beas
f_{i+2}(x) &=& kf_{i+1}(x)-\frac{2k}{\sqrt{x}}f_{1,1}^{(i+1)}(x)-\sum_{j=1}^{i+1}(-1)^{j+1}j!\binom{i+1}{j}(\frac{2}{\sqrt{x}})^{j}\sum_{q=1}^{j}(-1)^{q-1}f_{j,q}^{(i+1-j)}(x) \nonumber\\
&& + \, \frac{1}{2}\sum_{j=1}^{i+1}(-1)^{j+1}j!\binom{i+1}{j}(\frac{2}{\sqrt{x}})^{j+1}\Big(\sum_{q=1}^{j+1}(-1)^{q-1}(2k-2q+j+2)f_{j+1,q}^{(i+1-j)}(x)\Big).
\eeas
By induction on $f_{j,q}^{(i+1-j)}(x)$ and $f_{j+1,q}^{(i+1-j)}(x)$, we conclude that (\ref{expression for fl}) holds for $i+2$. 
This completes the proof.
\end{proof}

\begin{proof}[Proof of Proposition \ref{intro:prop}]
From (\ref{definition of tauk}),  $\tau_k(x)$ is a multiplier of $x^{k^2/2}$. Moreover, when $x$ tends to 0, we have
\be
\lim_{x\rightarrow 0} 
\frac{\tau_k(x)}{x^{k^2/2}}
= \det\left(\frac{1}{(i+j-1)!}\right)_{i,j=1,\ldots,k}
= (-1)^{\frac{k(k-1)}{2}} \prod_{j=0}^{k-1} \frac{j!}{(j+k)!}
=(-1)^{\frac{k(k-1)}{2}} \frac{G^2(k+1)}{G(2k+1)}.
\label{tauk at 0}
\ee
In the above, the second equality follows from (e.g., \cite[(4.39)-(4.41) with $m_i=0$]{conrey2006moments}). The last equality follows from the definition of Barnes G-function. So we can express $\tau_k(x)$ as
\beas
\tau_k(x) = (-1)^{\frac{k(k-1)}{2}} \frac{G^2(k+1)}{G(2k+1)} x^{\frac{k^2}{2}} e^{\frac{x}{2}} \sum_{j=0}^\infty d_j(k) x^j,
\eeas
where $e^{x/2}$ is added intentionally to cancel $e^{-x/2}$ in $F_1(M,k)$ defined in (\ref{F1Mk}), and $d_j(k)$ are some coefficients depending on $k$. 
Substituting $\tau_k(x)$ into $F_1(M,k)$ yields
\be
\label{F1mk2}
F_1(M,k) = (-1)^{M} \frac{G^2(k+1)}{G(2k+1)} (2M)! d_{2M}(k).
\ee
By \cite{winn2012derivative}, (\ref{F1Mk}) holds for any $M\geq 0$ with $2M$ an integer. Furthermore, by \cite{dehaye2008joint}, for such $M$
\[
F_1(M,k) = \frac{G^2(k+1)}{G(2k+1)} \frac{X_{M}(k)}{Y_{M}(k)},
\]
where $X_{M}, Y_{M}$ are polynomials in $k$ with explicit expressions that can be computed in a combinatorial way, and $Y_M(k)$ has no zeros when ${\rm Re}(k)>M-1/2$.  Together with (\ref{F1mk2}), for any $m\geq 0$, $d_{m}(k)$ is a rational function in $k$, and is analytic when ${\rm Re}(k)>(m-1)/2$.

We now substitute $\tau_k(x)$ into the expression (\ref{expression of fl}) for $f_l(x)$, leading to
\beas
f_l(x) &=& (-1)^{\frac{k(k-1)}{2}} \frac{G^2(k+1)}{G(2k+1)} x^{\frac{k^2}{2}-l}  \sum_{j=0}^\infty \Bigg( \sum_{q=0}^{\min(l,j)}
\sum_{i=0}^q \sum_{m=0}^{l-q}
(1/2)^i c_{q-i,m+i}^{(l)} \\
&& \times \,
 \binom{m+i}{i} \Big(\prod_{s=0}^{m-1} (j-q+\frac{k^2}{2}-s)\Big) d_{j-q}(k) \Bigg) x^j.
\eeas
When $m=0$, the product $\prod_{s=0}^{m-1} (j -q+\frac{k^2}{2}-s)$ is viewed as 1. We will make this convention throughout the paper.
According to (\ref{410definition of f_{l}}), $f_l(x)$ is a multiplier of $x^{\frac{k^2}{2}+l}$, so the summation over $j$ in the above formula starts from $2l$. Namely,
\beas
x^{-\frac{k^2}{2}-l} f_l(x) &=& (-1)^{\frac{k(k-1)}{2}} \frac{G^2(k+1)}{G(2k+1)}   \sum_{j=0}^\infty \Bigg( \sum_{q=0}^{l}
\sum_{i=0}^q \sum_{m=0}^{l-q} (1/2)^i c_{q-i,m+i}^{(l)} \\
&& \times \,
\binom{m+i}{i} \Big(\prod_{s=0}^{m-1} (2l+j-q+\frac{k^2}{2}-s)\Big) d_{2l+j-q}(k) \Bigg) x^{j}.
\eeas
Substituting this into (\ref{expression of bkm}) 
by changing variables $l$ to $h:=2M-l$
leads to
\beas
F_2(M,k) &=& \frac{G^2(k+1)}{G(2k+1)} \sum_{h=0}^{2M} \binom{2M}{h}
\sum_{j=0}^{2h} \sum_{q=0}^{2M-h} 
\sum_{i=0}^q \binom{2h}{j} 
(-1/2)^j (2h-j)!  \\
&& \times \,  \sum_{m=0}^{2M-h-q} (1/2)^i \binom{m+i}{i} \Big(\prod_{s=0}^{m-1} (4M-j-q+\frac{k^2}{2}-s) \Big) c_{q-i, m+i}^{(2M-h)}  d_{4M-j-q}(k) .
\eeas
By Theorem \ref{theorem2 in note 4}, all $c_{q-i, m+i}^{(2M-h)}$ are polynomials in $k$. From the above analysis, $ d_{4M-j-q}(k)$ is rational and analytic when ${\rm Re}(k)>2M-\frac{1}{2}$. So $F_2(M,k)$ is $\frac{G^2(k+1)}{G(2k+1)}$ multiplying a rational function, which is analytic at least when ${\rm Re}(k)>2M-\frac{1}{2}$.
Finally, by Lemma \ref{lem:Rational functions} in Appendix \ref{app}, we know that $F_2(M,k)$ can be expressed by a formula that is analytic when ${\rm Re}(k) > M-\frac{1}{2}$. Therefore, $\frac{G^2(k+1)}{G(2k+1)}$ multiplying the rational function is analytic when ${\rm Re}(k)>M-\frac{1}{2}$. 
\end{proof}

Regarding the computation of $F_2(M,k)$, we have several methods. One can either use the method discussed in the above proof or Lemma \ref{lem:Rational functions} in Appendix \ref{app}. In the first method,  $d_{m}(k)$ produced in $F_{1}(m/2,k)$ can be obtained recursively by a connection of it to a solution of the $\sigma$-Painlev\'{e} III$'$ equation (e.g., see \cite[Theorem 2]{basor2019representation} or \cite[Section 5]{forrester2006boundary}) and $c_{q,m}^{(l)}$ involved in the expression of $F_{2}(M,k)$ can be recursively obtained by Theorem \ref{theorem 1 in note 4}.

%Another method is as follows. Note that $\tau_k(x)$ has the expression (\ref{expression of tauk}). Substituting this into (\ref{expression of fl}), then by (\ref{expression of bkm}) we have
% \beas
% F_2(M,k) &=& \frac{G^2(k+1)}{G(2k+1)}  \sum_{h=0}^{2M} \binom{2M}{h} \sum_{j=0}^{2h} \binom{2h}{j}(-\frac{1}{2})^j (2h-j)! \\
% && \times \, \left( \sum_{m=0}^{2M-h} \, \sum_{q=0}^{2M-h-m} c^{(2M-h)}_{q,m} a_{4M-j-q} \prod_{s=0}^{m-1} (4M-j-q+\frac{k^2}{2}-s) \right).
% \eeas
% In the above, $c^{(2M-h)}_{q,m}$ are recursively given by Theorem \ref{theorem 1 in note 4}, and $a_0, a_{1},\ldots,a_{4M}$ are the first $(4M+1)$ Taylor coefficients of $e^{\frac{x}{2}} \exp( - \sum_{n=1}^\infty c_n \frac{(4x)^{2n}}{2n})$ at $x=0$, where $c_n$ are obtained from \cite[Proposition 5.1]{forrester2006boundary}.

\section{Truncated case and the proof of Theorem \ref{intro:thm5}}
\label{section:Truncated case}

For any fixed $k$, Theorem \ref{theorem 1 in note 4} provides an iterative approach to computing the $2k$-th moment of $Z_A''(1)$. The matrices involved have size $l$. From the analysis below Theorem \ref{theorem 1 in note 4}, $l$ can be as large as $2k$. Moreover, when the order of the derivative of $Z_A$ increases, $l$ can be much larger than $2k$.
For example, as analyzed in Section \ref{section:Generalization}, $l$ can be as large as $4k$ for the third-order derivative.
We can modify the iterative approach using truncated matrices of size $k$ to obtain a more effective approach from the viewpoint of computation.  

Let $l\geq 3,k\geq 1$ be integers. Denote $l_{0}=\min(l,k), l_{1}=\min(l-1,k),l_{2}=\min(l-2,k), l_{3}=\min(l+1,k).$
Let $\widetilde{C}_1^{(l)}= (\tilde{c}_{i,j}^{(1)})_{\substack{i=1,\ldots,l_{0}\\j=1,\ldots,l_{1}}}$ be the $l_{0}\times l_{1}$ matrix satisfying 
\[
\tilde{c}^{(1)}_{i,j} = 
\begin{cases}
c^{(1)}_{i,j}  & \text{if } j \leq l_0-1; \\
(-1)^{i+j} \frac{2k-2j+l}{l_0} & \text{if } l_0 \leq j \leq l_1.
\end{cases}
\]
Let $\widetilde{C}_2^{(l)}= (\tilde{c}_{i,j}^{(2)})_{\substack{i=1,\ldots,l_{0}\\j=1,\ldots,l_{2}}}$ be the $l_{0}\times l_{2}$ matrix satisfying 
\[
\tilde{c}^{(2)}_{i,j} = 
\begin{cases}
c^{(2)}_{i,j}  & \text{if } j \leq l_0-2; \\
(-1)^{i+j}/l_0 & \text{if } l_0-1 \leq j \leq l_2.
\end{cases}
\]
Let $\widetilde{C}_3^{(l)}= (\tilde{c}_{i,j}^{(3)})_{\substack{i=1,\ldots,l_{0}\\j=1,\ldots,l_{3}}}$ be the truncated $l_{0}\times l_{3}$ matrix by preserving the first $l_0$ rows and $l_3$ columns of ${C}_3^{(l)}$. 
The proof of the following result is similar to that of Theorem \ref{theorem 1 in note 4}. All the determinants shifted by hook diagrams $Y_{l,j}$ involved are restricted to truncated cases $j=1,\ldots,\min(l,k)$. For example, the summations over $j$ in Propositions \ref{513prop5} and \ref{513prop6} are restricted to $j=1,\ldots,\min(l,k)$. The ranges of $j$ in
Theorems \ref{thm1onHankel deteminants}, Proposition \ref{411proponS} are restricted to $j=2,\ldots,\min(l,k)$ and $j=3,\ldots,\min(l,k)$ respectively. The vectors in Theorem \ref{gerenel linear systems} are truncated to dimension $\min(l,k)$.

\begin{prop}
\label{prop33}
Using the same notation as above, let $k\geq 1, l\geq 3$ be integers, and let $f_{l}(x)$ be given as in (\ref{410definition of f_{l}}). 
Let $k\geq 1, l\geq 3$ be integers. 
For any $m,s$, define
\[
\f_{m,s}^{(i)} = \begin{pmatrix}
f_{m,1}^{(i)}(x) &
\cdots &
f_{m,s}^{(i)}(x) 
\end{pmatrix}^T, \quad
\hat{\f}_{m,s}^{(i)} = \begin{pmatrix}
f_{m,2}^{(i)}(x) &
\cdots &
f_{m,s}^{(i)}(x) 
\end{pmatrix}^T.
\]
Then, for $i\geq 0$,
\be
f_{i+1}(x) = 
\sum_{q=1}^{\min(k,2)} (-1)^{q-1} f^{(i)}_{2,q}(x),
\ee
where $f^{(i)}_{2,1}(x), f^{(i)}_{2,2}(x)$ satisfy the following recursive relation
\bea
\f_{l,l_0}^{(i)}
&=&
-\, B^{(l_0)} \Big(\sqrt{x}\frac{d}{d x}-\frac{k^2+l-1-2i}{2\sqrt{x}} \Big)
\begin{pmatrix}
\f_{l-1,l_0-1}^{(i)} \\
0
\end{pmatrix}  
- \frac{1}{\sqrt{x}} \widetilde{C}_1^{(l)}
\f_{l-1,l_1}^{(i)}
+ \widetilde{C}_2^{(l)}
\f_{l-2,l_2}^{(i)} +\frac{2i}{\sqrt{x}}\widetilde{C}_3^{(l)}
\f_{l+1,l_3}^{(i-1)}  \nonumber \\
&&
+ \, B^{(l_0)} \Big(2i\frac{d}{dx}-\frac{i(k^2+l)+2i(i-1)}{x}\Big)
\begin{pmatrix}
\hat{\f}_{l,l_0}^{(i-1)} \\
0
\end{pmatrix}.
\eea
The initial conditions for the above recursive formula are as follows.
\beas
f_{0}(x)&=&\tau_{k}(x), \nonumber \\
f_{1,1}^{(i)}(x)&=&
\sqrt{x} \frac{d}{dx} f_i - \frac{1}{2\sqrt{x}} k^2 f_i - \frac{i}{\sqrt{x}} f_i, \nonumber \\
f_{2,1}^{(i)}(x)&=&\frac{1}{2}
kf_i   - \frac{i}{\sqrt{x}} 
\sum_{q=1}^{\min(k,3)} (-1)^{q-1} f^{(i-1)}_{3,q} \nonumber  \\
&& +\, 
\frac{1}{2}\Big(x \frac{d^2}{dx^2} f_i - (k^2+2k+2i) \frac{d}{dx} f_i + \frac{(k^2+2i)(k^2+4k+2i+2)}{4x} f_i\Big), \nonumber \\
f_{2,2}^{(i)}(x)&=& -\frac{1}{2}
kf_i   + \frac{i}{\sqrt{x}} \sum_{q=1}^{\min(k,3)} (-1)^{q-1} f^{(i-1)}_{3,q}  \nonumber \\
&& +\,
\frac{1}{2}\Big(x \frac{d^2}{dx^2} f_i - (k^2-2k+2i) \frac{d}{dx} f_i + \frac{(k^2+2i)(k^2-4k+2i+2)}{4x} f_i\Big). 
\eeas
\end{prop}

Recall that our goal is to compute the coefficient of the main term of $\int_{\U(N)} |Z_A''(1)|^{2k} dA_N$, which is
\be
\label{main coefficient}
(-1)^{\frac{k(k-1)}{2}} \sum_{h=0}^{2k}\binom{2k}{h}( \frac{d}{dx})^{2h} \Big(e^{-x/2} x^{-\frac{k^2}{2}+h-2k}f_{2k-h}(x)\Big)\Bigg|_{x=0}
\ee
by Proposition \ref{intro:prop1}.
From Proposition \ref{prop33}, by a similar argument to that of
Theorem \ref{theorem2 in note 4}, for any fixed $k\geq 1$, let $l \geq 1$, then $f_l(x)$ has the following expression
\be
f_{l}(x) = \frac{1}{x^l} \sum_{m=0}^l  x^mP_{m}(x) \frac{d^m \tau_k(x)}{dx^m},
\label{82eq}
\ee
where $P_{m}(x)=\sum_{j=0}^{l-m} c_{j,m}^{(l)} x^{j}$ and $c_{j,m}^{(l)}$ are constants depending on $j, l, m$. So to compute (\ref{main coefficient}), it suffices to find the Taylor series of $\tau_k(x)$ at $x=0$.

According to \cite[(5.7), (5.8), (5.11)-(5.15)]{forrester2006boundary}, we can express $\tau_k(x)$ as
\bea
\tau_k(x) = (-1)^{\frac{k(k-1)}{2}} \frac{G^2(k+1)}{G(2k+1)} x^{\frac{k^2}{2}} e^{\frac{x}{2}} \exp\left(-\sum_{n=1}^kc_{2n} \frac{(4x)^{2n}}{2n} - \sum_{n=2k+1}^\infty c_n \frac{(4x)^{n}}{n} \right ).
\label{tauk-nw formula}
\eea
Let $c_0=-k^2, c_1=0$. Denote $\eta(s) = \sum_{n=0}^\infty c_n s^n$, which satisfies the following differential equation (e.g., see \cite[(5.15)]{forrester2006boundary})
\be
\label{equation for eta}
(s \eta'')^2 + 4 ((\eta')^2 - \frac{1}{64} ) (\eta-s \eta') - \frac{k^2}{16} = 0.
\ee
As a result, we can deduce that $c_2 = \frac{1}{64(4k^2-1)}$, and  for any $q \geq 3$
\bea
\frac{(q-1-2k)(q-1+2k)}{16(4k^2-1)} c_q 
&=& -\sum_{l=1}^{q-3}(l+1)(l+2)(q-l)(q-l-1) c_{l+2} c_{q-l} \nonumber \\ 
&& \hspace{-2cm} + \, 4 \sum_{l=2}^{q-1} (l-1) c_l E_{q-l}
+4k^2 \sum_{l=2}^{q-2} (l+1) (q-l+1) c_{l+1} c_{q-l+1}, \label{eta}
\eea
where $E_q = \sum_{l=0}^q (l+1) (q-l+1) c_{l+1} c_{q-l+1}$. From (\ref{eta}), we have $c_{2m+1}=0$ for $m\leq k-1$. When $q=2k+1$, the left-hand side of (\ref{eta}) vanishes, so we cannot use it to determine $c_{2k+1}$. 
This implies that the differential equation (\ref{equation for eta}) has a one-parameter family of solutions, corresponding to different values of $c_{2k+1}$. Here we remark that for $\tau_k(x)$ we cannot impose $c_{2k+1}=0$, as Forrester and Witte did in \cite[(5.16)]{forrester2006boundary}. 
For example, when $k=1$, $c_{2k+1}=-1/3072 \neq 0$.
Thus we need a new method to determine $c_{2k+1}$.

First, we explain why we need information about $c_{2k+1}$.
To compute the $2k$-th moment of the first order derivative of $Z_A$, according to
\cite[(4.38)-(4.41)]{basor2019representation}, the coefficients $c_1,c_2,\ldots,c_{2k}$ are enough. However, in the situation of computing $2k$-th moment of the second order derivative of $Z_A$, we need the first $4k$ coefficients $c_1,c_2,\ldots,c_{4k}$. We  explain this below in detail.

By Taylor expanding, we may rewrite $\tau_k(x)$ in the form of (\ref{expression of tauk}).  
Then by (\ref{expression of fl})
\beas
f_l(x) = (-1)^{\frac{k(k-1)}{2}} \frac{G^2(k+1)}{G(2k+1)} x^{\frac{k^2}{2}-l}
 \sum_{i=0}^\infty \left( \sum_{q=0}^{\min(i, l)}
 \sum_{m=0}^{l-q} \Big( \prod_{s=0}^{m-1} (i-q+\frac{k^2}{2}-s)\Big) c_{q,m}^{(l)} a_{i-q} 
\right)x^i.
\eeas
According to the expression (\ref{410definition of f_{l}}), $f_l(x)$ is a multiplier of $x^{\frac{k^2}{2}+l}$, so the summation over $i$ in the above formula starts from $2l$. Hence
\[
x^{-\frac{k^2}{2}-l} f_l(x)
=(-1)^{\frac{k(k-1)}{2}} \frac{G^2(k+1)}{G(2k+1)} 
\sum_{i=0}^\infty \left( \sum_{q=0}^{l}
 \sum_{m=0}^{l-q} \Big(\prod_{s=0}^{m-1} (i+2l-q+\frac{k^2}{2}-s)\Big) c_{q,m}^{(l)} a_{i+2l-q} 
\right)x^i.
\]
Substituting this into (\ref{main coefficient}), we obtain that
\beas
(\ref{main coefficient}) &=& \frac{G^2(k+1)}{G(2k+1)} \sum_{i=0}^{2k} \binom{2k}{i} \sum_{j=0}^{2i} \binom{2i}{j} (-1/2)^j (2i-j)! \\
&& \times \, \sum_{q=0}^{2k-i} \sum_{m=0}^{2k-i-q} \Big(\prod_{s=0}^{m-1} (4k-j-q+\frac{k^2}{2}-s) \Big) c_{q,m}^{(2k-i)}  a_{4k-j-q} .
\eeas
As one can see, this depends on the coefficients $a_1, \ldots, a_{4k}$. Note that
\be
\label{relation a and c}
e^{\frac{x}{2}} \exp\left(-\sum_{n=1}^kc_{2n} \frac{(4x)^{2n}}{2n} - \sum_{n=2k+1}^\infty c_n \frac{(4x)^{n}}{n} \right ) = \sum_{j=0}^\infty a_j x^j,
\ee
so (\ref{main coefficient}) depends on coefficients $c_1, \ldots, c_{4k}$. This explains why we need the first $4k$ coefficients $c_1,c_2,\ldots,c_{4k}$ to compute the $2k$-th moment of the second order derivative of $Z_A$.

As explained below (\ref{82eq}), to compute the $2k$-th moment, it suffices to find the Taylor expansion of $\tau_k(x)$ at $x=0$. From (\ref{expression of tauk}), it suffices to compute $a_0,a_1,\ldots,a_{4k}$. We deduce these $a_j$ as an application of Proposition \ref{recursive formula for tau0}. Note that by (\ref{relation a and c}), it is not hard to show that $a_{2k+1} + \frac{4^{2k+1}}{2k+1} c_{2k+1}$ is the coefficient of $x^{2k+1}$ in the polynomial 
\[
\left(\sum_{i=0}^{2k+1} \frac{(x/2)^i}{i!} \right) \left( \sum_{i=0}^{k} \frac{1}{i!}\left(-\sum_{n=1}^k \frac{c_{2n} (4x)^{2n}}{2n}\right)^i  \right).
\]
Consequently, if we can determine $a_{2k+1}$, we can also determine $c_{2k+1}$.

\begin{lem}
\label{lem:Yk,k-1}
Let $Y_{k,k}, Y_{k-1,k-1}$ be hook diagrams, then
\be
\label{lem:Yk,k-1-eq}
\sqrt{x} \frac{d}{dx} \tau_{k,Y_{k,k}}(x) + \frac{k^2+k}{2\sqrt{x}} \tau_{k,Y_{k,k}}(x) 
-\tau_{k,Y_{k-1,k-1}}(x) = 0.
\ee
\end{lem}

\begin{proof}
Note that $\tau_{k,Y_{k,k}}(x) = \det (I_{i+j+\alpha}(2\sqrt{x}))_{i,j=0,\ldots,k-1}$, where $\alpha=2$.
Let $(F_{i,j})_{i,j=0,\ldots,k-1}$ be the cofactor matrix of $(I_{i+j+\alpha}(2\sqrt{x}))_{i,j=0,\ldots,k-1}$. Then
$\tau_{k,Y_{k-1,k-1}}(x) = (I_{i+j+\alpha-1}(2\sqrt{x})) \cdot (F_{i,j})$.
Since
\[
I_{i+j+\alpha-1}(2\sqrt{x})
=\sqrt{x} \frac{d}{dx} I_{i+j+\alpha}(2\sqrt{x}) + \frac{\alpha+i+j}{2\sqrt{x}} I_{i+j+\alpha}(2\sqrt{x}),
\]
we have
\beas
&& \tau_{k,Y_{k-1,k-1}}(x) \\
&=& \sqrt{x}( \frac{d}{dx} I_{i+j+\alpha}(2\sqrt{x})) \cdot (F_{i,j}) + 
(\frac{\alpha+j}{2\sqrt{x}} I_{i+j+\alpha}(2\sqrt{x})) \cdot (F_{i,j})
+
(\frac{i}{2\sqrt{x}}I_{i+j+\alpha}(2\sqrt{x})) \cdot (F_{i,j}) \\
&=& \sqrt{x} \frac{d}{dx} \tau_{k,Y_{k,k}}(x) + \frac{k^2+k}{2\sqrt{x}} \tau_{k,Y_{k,k}}(x) ,
\eeas
as claimed.
\end{proof}

For a fixed $k$, by (\ref{definition of fjqi}) and (\ref{exforflq}) with $i=0$, we may assume that for any $l \geq 1$,
\be
\label{taukYss}
\tau_{k,Y_{l,l}} (x) = x^{-l/2} \sum_{m=0}^l  x^m \left(\sum_{j=0}^{\lfloor \frac{l-m}{2}\rfloor} b_{j,m}^{(l)} x^j \right) \frac{d^m \tau_k(x)}{dx^m}.
\ee
By Proposition \ref{recursive formula for tau0}, 
\[
\tau_{k,Y_{k+1,k+1}}(x) = \frac{1}{k+1} \left(\sqrt{x} \frac{d\tau_{k,Y_{k,k}}(x)}{dx} - \frac{k^2+k}{2\sqrt{x}} \tau_{k,Y_{k,k}}(x) \right) + \frac{k}{\sqrt{x}} \tau_{k,Y_{k,k}}(x)-\frac{1}{k+1} \tau_{k,Y_{k-1,k-1}}(x)  .
\]
By Lemma \ref{lem:Yk,k-1}, $\tau_{k,Y_{k+1,k+1}}(x)\equiv 0$, which defines a differential equation. For example, if $k=2$, then we have the following differential equation:
\be
\label{diff equation}
\frac{x^{3} \left(\frac{d^{3}\tau_2(x)}{d x^{3}}  \right) +4x^2 \left(\frac{d^{2}\tau_2(x)}{d x^{2}} \right)  -2x \left(2x+1 \right) \left(\frac{d\tau_2(x)}{d x}
\right)-2 \left(x +2\right)\tau_2(x)}{6 x^{\frac{3}{2}}} = 0.
\ee
More generally, for any given $k$, by 
(\ref{taukYss}), we have
\bea\label{an identity}
x^{-\frac{k+1}{2}}\sum_{m=0}^{k+1}x^m \left(\sum_{j=0}^{\lfloor \frac{k+1-m}{2}\rfloor} b_{j,m}^{(k+1)} x^j \right) \frac{d^m\tau_k(x)}{dx^m} \equiv 0.
\eea
Now we substitute the expression (\ref{expression of tauk}) for $\tau_k(x)$ into (\ref{an identity}) to obtain
\[
(-1)^{\frac{k(k-1)}{2}} \frac{G^2(k+1)}{G(2k+1)} x^{\frac{k^2-k-1}{2}} \sum_{i=0}^\infty \left( \sum_{q=0}^{\min(i, \lfloor \frac{k+1}{2} \rfloor)}
 \sum_{m=0}^{k+1-2q} \Big( \prod_{s=0}^{m-1} (i-q+\frac{k^2}{2}-s) \Big) b_{q,m}^{(k+1)} a_{i-q} 
\right)x^i=0.
\]
Hence for any $i\geq 1$,
\bea\label{find ai}
&& a_{i}  \sum_{m=0}^{k+1} \Big(\prod_{s=0}^{m-1} (i+\frac{k^2}{2}-s) \Big) b_{0,m}^{(k+1)} \nonumber \\
&=&
-\sum_{q=1}^{\min(i,\lfloor \frac{k+1}{2} \rfloor)}  a_{i-q} \sum_{m=0}^{k+1-2q} \Big(\prod_{s=0}^{m-1} (i-q+\frac{k^2}{2}-s) \Big) b_{q,m}^{(k+1)}.
\eea

In the following, we aim to give a more concise recursive formula for $a_i$, i.e., we aim to prove Theorem \ref{intro:thm5}. The following result shows how to compute the coefficients of $a_i$ in (\ref{find ai}).

\begin{prop}\label{recursive for Taylor series of tau} Let $k\geq 1$ be a given integer.  Let $b_{q,m}^{(l)}$ be the coefficients in  (\ref{taukYss}). For any $l \geq 1$ and any integer $n$, define
\bea\label{Dlq}
D_{l,q}(n) = \sum_{m=0}^{l-2q} \Big(\prod_{s=0}^{m-1} (n-q+\frac{k^2}{2}-s)\Big) b_{q,m}^{(l)}
\eea
if $0\leq q \leq \lfloor l/2\rfloor $, and define $D_{l,q}(n)=0$ otherwise.
Then $D_{l,q}(n)$ satisfies the following recursive relation when $l\geq 3$
\bea
D_{l,q}(n) &=& \frac{n+(l-1)(2k-l+1)}{l} D_{l-1,q}(n) + \frac{l-k-2}{l} D_{l-2,q-1}(n-1), 
\label{recursive formula of Dlq-1} \\
D_{2,0}(n) &=& \frac{n(2k-1+n)}{2}, \quad 
D_{2,1}(n) = -\frac{k}{2}, \quad D_{1,0}(n) = n. \label{recursive formula of Dlq-2} 
\eea
In particular,
\[
D_{l,0}(n) =  \frac{n(2k-1+n)}{2} \prod_{s=3}^l\frac{(2k-s+1)(s-1)+n}{s}.
\]
\end{prop}

\begin{proof}
By Proposition \ref{recursive formula for tau0}, for any $l \geq 3$, we have 
\beas
\tau_{k,Y_{l,l}} &=&  \frac{1}{l} \left(\sqrt{x} \frac{d\tau_{k,Y_{l-1,l-1}}}{dx} - \frac{k^2+l-1}{2\sqrt{x}} \tau_{k,Y_{l-1,l-1}} \right)  \\
&& +\,  \frac{(l-1) (2k-l+2) }{l\sqrt{x}} \tau_{k,Y_{l-1,l-1}} + \frac{l-k-2}{l} \tau_{k,Y_{l-2,l-2}} .
\eeas
By (\ref{taukYss}), for $q=0,1,\ldots,\lfloor l/2 \rfloor$, 
\[
b_{q,l-2q}^{(l)} = \frac{1}{l} b_{q,l-2q-1}^{(l-1)}+\frac{l-k-2}{l}b_{q-1,l-2q}^{(l-2)},
\]
and
for $0\leq m\leq l-2q-1$,
\[
b_{q,m}^{(l)} =
\frac{2(q+m)-k^2+2(l-1)(2k-l+1)}{2l} b_{q,m}^{(l-1)} + \frac{1}{l} b_{q,m-1}^{(l-1)}+\frac{l-k-2}{l}b_{q-1,m}^{(l-2)}.
\]
In the above, when $m=0$, $b_{q,m-1}^{(l-1)} := 0$ and when $q=0$, $b_{q-1,m}^{(l-2)}:=0$.
Denote 
$B_{q,m}^{(l)}(n)= b_{q,m}^{(l)}  \prod_{s=0}^{m-1} (n+\frac{k^2}{2}-s-q)$, then 
\[B_{q,l-2q}^{(l)} (n)=\frac{n+\frac{k^2}{2}-q-(l-2q-1)}{l} B_{q,l-2q-1}^{(l-1)} (n) +\frac{l-k-2}{l}B_{q-1,l-2q}^{(l-2)}(n-1),\]
and for $0\leq m \leq l-2q-1$,
\beas
B_{q,m}^{(l)}(n)
&=&
\frac{2(q+m)-k^2+2(l-1)(2k-l+1)}{2l} B_{q,m}^{(l-1)}(n) \\
&& +\, \frac{n+\frac{k^2}{2}-q-m+1}{l} B_{q,m-1}^{(l-1)} (n)+\frac{l-k-2}{l}B_{q-1,m}^{(l-2)}(n-1).
\eeas
So
\[
D_{l,q}(n) = \sum_{m=0}^{l-2q}B_{q,m}^{(l)}(n)=\frac{n+(l-1)(2k-l+1)}{l}D_{l-1,q}(n)+\frac{l-k-2}{l}D_{l-2,q-1}(n-1).
\]
In particular,
\[D_{l,0}(n)=\frac{n+(l-1)(2k-l+1)}{l}D_{l-1,0}(n).\]
By Proposition \ref{initial values of tau}, $D_{2,0}(n)= \frac{(2k-1+n)n}{2}, D_{2,1}(n) = -\frac{k}{2}$.

So we have
\[
D_{l,0}(n) =  \frac{n(2k-1+n)}{2} \prod_{s=3}^l \frac{(2k-s+1)(s-1)+n}{s}.
\]
By Proposition \ref{lemmaontranslationby1}, $D_{1,0}(n) = n$.
\end{proof}

\begin{proof}[Proof of Theorem \ref{intro:thm5}]
Equation (\ref{thm:recursive formula for ai-eq1}) follows directly from (\ref{find ai}). Comparing (\ref{expression of tauk}), (\ref{tauk at 0}),  we have $a_0=1$. By Proposition \ref{recursive for Taylor series of tau},
\[
D_{k+1,0}(i) = \frac{i(2k-1+i)}{2} \prod_{s=3}^{k+1}\frac{(2k-s+1)(s-1)+i}{s}.
\]
Note that the solutions of $(2k-s+1)(s-1)+i=0$ are $s=k +1+\sqrt{k^{2}+i}, k +1-\sqrt{k^{2}+i}$, so $D_{k+1,0}(i) \neq 0$.
\end{proof}

\section{Generalizations}
\label{section:Generalization}

In this section, we consider the generalization of the recursive formula in Theorem \ref{theorem 1 in note 4}  to higher-order derivatives. We will discuss the case of the third-order derivative in detail below, then briefly explain how to generalize the approach to the higher-order case. From \cite[Theorem 24]{jon-fei}, 
\bea
F_{3}(M,k)&:=&\lim_{N\rightarrow \infty}\frac{\int_{\U(N)} |Z_A^{(3)} (1)|^{2M} |Z_A(1)|^{2k-2M} dA_N}{N^{k^2+6M}}\nonumber\\
&=&(-1)^{3M+\frac{k(k-1)}{2}}N^{k^2+6M}
\sum_{\substack{n_1+n_2+n_3= 2M}}\binom{2M}{n_1,n_2,n_3}
\frac{6^{2M}}{ 
2^{n_1}
3^{n_2}
6^{n_3}
} \nonumber\\
&& \times \, \left( \frac{d}{dx} \right)^{n_1+3n_3} \left( e^{-\frac{x}{2}} x^{-\frac{k^2}{2} - 3M + \frac{1}{2} (n_1+3n_3)} F(n_1,n_2) \right) \Bigg|_{x=0}
\label{the third term},
\eea
where
\[
F(n_1,n_2) = \sum_{\substack{ \sum_{j=0}^{k-1} h_{1,j}=n_1 \\ \sum_{j=0}^{k-1} h_{2,j}=n_2}} 
\hspace{-.3cm}
\binom{n_1}{h_{1,0},\ldots,h_{1,k-1}} \binom{n_2}{h_{2,0},\ldots,h_{2,k-1}} \det\Big( I_{ 2h_{1,j}+3h_{2,j}+i+j+1 } (2\sqrt{x}) \Big)_{i,j=0,\ldots,k-1}.
\]
Compared with the second order case, we now need to introduce two variables $t_1, t_2$. Below, every result coincides with that in the second order case when $t_2=0$.
Similar to (\ref{definition of g}), we now define
\[
g_\beta(x,t_1,t_2) = \sum_{n=0}^\infty \sum_{m=0}^\infty \frac{t_1^n t_2^m}{n!m!} I_{2n+3m+\beta} (2\sqrt{x}).
\]
This is (\ref{definition of g}) when $t_2=0$. By Lemma \ref{NOTE4lem2},
\[
F(n_1,n_2)= \left( \frac{\partial}{\partial t_1} \right)^{n_1} \left( \frac{\partial}{\partial t_2} \right)^{n_2} \det\left( g_{i+j+1}(x,t_1,t_2) \right)_{i,j=0,\ldots,k-1} \Bigg|_{t_1=t_2=0} .
\]
Our goal is to give a recursive formula for $F(n_1,n_2)$.

Let $Y=(l_1,\ldots,l_s)$ be a Young diagram with $1\leq s \leq k$. Set $l_{s+1}=\cdots=l_k=0$. Similar to (\ref{definition of G}) and (\ref{definition of F}), let
\[
G_{k,Y}(x,t_1,t_2)=\det\left( g_{i+j+1+l_{k-j}}(x,t_1,t_2) \right)_{i,j=0,\ldots,k-1}.
\]
Similar to (\ref{definition of fjqi}), let
\bea\label{definition of fjqin1n2}
f_{j,q}^{(n_{1},n_{2})}:=\left( \frac{\partial}{\partial t_1} \right)^{n_1}  \left( \frac{\partial}{\partial t_2} \right)^{n_2} G_{k,Y_{j,q}}\Big|_{t_{1}=t_{2}=0},
\eea
where $Y_{j,q}$ is a hook diagram.
Then we have the following recursive relations for $g_\beta$, which correspond to (\ref{412recursive formula for gbeta1})-(\ref{412recursive formula for gbeta})
\bea
\frac{\partial}{ \partial x}g_\beta&=& \frac{1}{\sqrt{x}}  g_{\beta+1} + \frac{\beta}{2x} g_\beta + \frac{t_1}{x} g_{\beta+2} + \frac{3t_2}{2x} g_{\beta+3}  \nonumber \\
 &=& \frac{1}{\sqrt{x}}  g_{\beta-1} - \frac{\beta}{2x} g_\beta - \frac{t_1}{x} g_{\beta+2} - \frac{3t_2}{2x} g_{\beta+3}, \nonumber \\
g_{\beta+2} &=&
g_{\beta} - \frac{\beta+1}{\sqrt{x}} g_{\beta+1} - \frac{2t_1}{\sqrt{x}} g_{\beta+3} -\frac{3t_2}{\sqrt{x}} g_{\beta+4} . \label{g beta 2}
\eea
So similar to the proofs of Propositions \ref{derivative and translations} and \ref{412320prop3}, we have the following relations on the partial derivatives and translations of $G_{k,Y}$:
\bea
\frac{\partial}{\partial t_1} G_{k,Y} &=& T_2 G_{k,Y}, 
\label{partial derivative 1} \\
\frac{\partial}{\partial t_2} G_{k,Y} &=& T_3 G_{k,Y},
\label{partial derivative 2} \\
T_1G_{k,Y} &=& \sqrt{x} \frac{\partial G_{k,Y}}{\partial x} - \frac{(k^2+l) G_{k,Y}}{2\sqrt{x}} - \frac{t_1}{\sqrt{x}} \frac{\partial G_{k,Y}}{\partial t_1} - \frac{3t_2}{2\sqrt{x}} \frac{\partial G_{k,Y}}{\partial t_2},
\label{partial derivative 0}
\eea
where $l$ is the length of the Young diagram $Y$.

As usual, when $Y=\emptyset$ is an empty Young diagram, we denote $G_{k,Y}$ by $G_k$.
Similar to the proof of Lemma \ref{413lemma1}, by (\ref{g beta 2}) and (\ref{partial derivative 1}), for any $m\geq 1$ we have
\bea
\frac{\partial}{\partial t_1} G_k &=& k G_k - \frac{S_1G_k }{\sqrt{x}} 
+ \sum_{i=1}^{m-1} (-1)^i \sum_{j=0}^i \binom{i}{j} \left(\frac{3t_2}{\sqrt{x}}\right)^j \left(\frac{2t_1}{\sqrt{x}}\right)^{i-j} \left(T_{i+j} G_k - \frac{S_{i+j+1} G_k}{\sqrt{x}}\right) \nonumber \\
&& + \, (-1)^{m} \sum_{j=0}^{m} \binom{m}{j} \left(\frac{3t_2}{\sqrt{x}}\right)^j \left(\frac{2t_1}{\sqrt{x}}\right)^{m-j} T_{m+2+j} G_k .
\label{partial t1 of G}
\eea
Setting $m=n_1+n_2$,  
\bea
F(n_1,n_2) &=&  k F(n_1-1,n_2) +\sum_{j=0}^{n_2} \left(\frac{3}{\sqrt{x}}\right)^j \binom{n_2}{j} \sum_{i=0}^{n_1-1} (-1)^{i+j}
(i+j)! \left(\frac{2}{\sqrt{x}} \right)^i  \binom{n_1-1}{i} 
\nonumber \\
&& \hspace{-2.6cm} \times\,  
\Bigg(
\sum_{s=1}^{i+2j}(-1)^{s-1}f_{i+2j,s}^{(n_{1}-1-i,n_{2}-j)}
- \frac{1}{\sqrt{x}}\sum_{s=1}^{i+2j+1}(-1)^{s-1}(2k-2s+i+2j+2)f_{i+2j+1,s}^{(n_{1}-1-i,n_{2}-j)} \Bigg).
\label{recursive formula for F}
\eea
For $n_1=0$, the expression for $F(0,n_2)$ is obtained by (\ref{g beta 2}) and (\ref{partial derivative 2}).
More precisely, by (\ref{partial derivative 2}), we can expand $\frac{\partial }{\partial t_2} G_k$ in powers of $t_{2}$, 
\bea
\frac{\partial }{\partial t_2} G_k &=& 
T_1G_k - \frac{S_2G_{k}}{\sqrt{x}}  + \sum_{i=1}^{l_2-1} (-1)^i \left( T_{2i+1} G_k - \frac{S_{2i+2}G_k}{\sqrt{x}}  \right) \left( \frac{3t_2}{\sqrt{x}}\right)^i \nonumber \\
&& +\, (-1)^{l_2} \left( \frac{3t_2}{\sqrt{x}}\right)^{l_2} T_{2l_2+3} G_k + t_1 H(x,t_1,t_2),
\label{partial t2 of G}
\eea
where $H(x,t_{1},t_{2})$
is some function of $x,t_{1},t_{2}$ and has no singularity at  
$t_{1}=0$ or $t_{2}=0$. Note that in (\ref{partial t1 of G}), we expanded $\frac{\partial }{\partial t_1} G_k$ with respect to powers of $t_1, t_2$. But in (\ref{partial t2 of G}), we only expanded it in terms of powers of $t_2$ because our purpose is to express $F(0,n_2)$.
Hence
\bea
F(0,n_2) &=& f_{1,1}^{(0,n_2-1)} - \frac{1}{\sqrt{x}} \Big( (2k+1) f_{2,1}^{(0,n_2-1)} - (2k-1) f_{2,2}^{(0,n_2-1)} \Big)
\nonumber \\
&& + \, \sum_{i=1}^{n_2-1} (-1)^i i! \left( \frac{3}{\sqrt{x}}\right)^i \binom{n_2-1}{i} \Bigg( \sum_{s=1}^{2i+1} (-1)^{s-1} f_{2i+1,s}^{(0,n_2-1-i)} \nonumber \\
&&
-\, \frac{1}{\sqrt{x}} \sum_{s=1}^{2i+2} (-1)^{s-1} (2k-2s+2i+3) f_{2i+2,s}^{(0,n_2-1-i)} \Bigg)\label{F0}.
\eea
When $n_2=0$, we have $F(l,0)=f_l(x)$ defined in (\ref{410definition of f_{l}}).
To compute $F(n_1,n_2)$, we need to compute $f^{(n_1,n_2)}_{j,q}$. Similar to Theorem \ref{theorem 1 in note 4}, we have the following result on the recursive formula for $f^{(n_1,n_2)}_{j,q}$.

\begin{prop}
\label{final prop}
For $1\leq i \leq m$, denote
\[
\f^{(n_1,n_2)}_{m,i} = \begin{pmatrix}
f^{(n_1,n_2)}_{m,i} \\
\vdots \\
f^{(n_1,n_2)}_{m,m} \\
\end{pmatrix},
\]
then we have the following recursive formula
\bea
\f^{(n_1,n_2)}_{m,1} &=& - \, \sqrt{x}B^{(m)}(  \frac{d}{dx} -  \frac{k^2+m-1-2n_1-3n_2}{2x})
\begin{pmatrix}
\f^{(n_1,n_2)}_{m-1,1} \\
0
\end{pmatrix} \nonumber \\
&& -\, \frac{1}{\sqrt{x}} C_1^{(m)} \f^{(n_1,n_2)}_{m-1,1} + C_2^{(m)} \f^{(n_1,n_2)}_{m-2,1}  + \frac{2n_1}{\sqrt{x}} C_3^{(m)} \f^{(n_1-1,n_2)}_{m+1,1} - \frac{3n_2}{\sqrt{x}} C_4^{(m)} \f^{(n_1,n_2-1)}_{m+2,1} \nonumber \\
&& + \, 2n_1B^{(m)}(  \frac{d}{dx}- \frac{k^2+m+2n_1+3n_2-2}{2x})
\begin{pmatrix}
\f^{(n_1-1,n_2)}_{m,2} \\
0
\end{pmatrix} \nonumber \\
&& + \, \frac{3 n_2}{\sqrt{x}} B^{(m)}
\begin{pmatrix}
\f^{(n_1+1,n_2-1)}_{m,2} \\
0
\end{pmatrix}  - 3 n_2B^{(m)} ( \frac{d}{dx} - \frac{k^2+m+1}{2x} )
\begin{pmatrix}
\f^{(n_1,n_2-1)}_{m+1,3} \\
0
\end{pmatrix} \nonumber \\
&& + \, \frac{3n_1n_2}{x} B^{(m)}\begin{pmatrix}
\f^{(n_1-1,n_2-1)}_{m+1,3} \\
0
\end{pmatrix} + \frac{9n_2(n_2-1)}{2x} B^{(m)} \begin{pmatrix}
\f^{(n_1,n_2-2)}_{m+1,3} \\
0
\end{pmatrix},
\label{final prop:eq}
\eea
where $B^{(m)}, C_1^{(m)}, C_2^{(m)}, C_3^{(m)}$ are defined in (\ref{definition of B}), 
 (\ref{definitionofC1}), (\ref{definitionofC2}), (\ref{definition of C3}), and
$C_4^{(m)} = (c_{i,j}^{(4)})_{\substack{i=1,\ldots,m \\j=1,\ldots,m+2}}$ is an $m \times (m+2)$-matrix satisfying
\be\label{definition of C4}
c_{i,j}^{(4)}
=\begin{cases}
(-1)^{j-1} & i=1, j=1,2,3; \\
(-1)^{i-j-1} \frac{2}{(j-2)(j-3)} & 
j>i+2  ;\\
\frac{i+2}{i} & j=i+2, i \neq 1; \\
0 & j \leq i+1, i \neq 1.
\end{cases}
\ee
\end{prop}

In particular, when $n_1=0$, we have
\bea
\f^{(0,n_2)}_{m,1} &=& -\, \sqrt{x} B^{(m)}(  \frac{d}{dx} -  \frac{k^2+m-1-3n_2}{2x})
\begin{pmatrix}
\f^{(0,n_2)}_{m-1,1} \\
0
\end{pmatrix} \nonumber \\
&& -\, \frac{1}{\sqrt{x}} C_1^{(m)} \f^{(0,n_2)}_{m-1,1} + C_2^{(m)} \f^{(0,n_2)}_{m-2,1}  - \frac{3n_2}{\sqrt{x}} C_4^{(m)} \f^{(0,n_2-1)}_{m+2,1} \nonumber \\
&& + \, \frac{3 n_2}{\sqrt{x}} B^{(m)}
\begin{pmatrix}
\f^{(1,n_2-1)}_{m,2} \\
0
\end{pmatrix}  - 3 n_2B^{(m)} ( \frac{d}{dx} - \frac{k^2+m+1}{2x} )
\begin{pmatrix}
\f^{(0,n_2-1)}_{m+1,3} \\
0
\end{pmatrix} \nonumber \\
&& + \, \frac{9n_2(n_2-1)}{2x} B^{(m)} \begin{pmatrix}
\f^{(0,n_2-2)}_{m+1,3} \\
0
\end{pmatrix}.
\label{recursive formula i=0}
\eea
To use the above recursive formulae, we need some initial conditions.  These are given as follows
\bea
F(0,0) &=& \tau_k,\label{final cond0} \\
F(1,0)&=&-2k \frac{d \tau_k}{dx} + (k+ \frac{k^3}{x}) \tau_k\\
f_{1,1}^{(n_1,n_2)}
&=& \left(\sqrt{x} \frac{d}{dx} - \frac{k^2}{2\sqrt{x}} - \frac{n_1}{\sqrt{x}} - \frac{3n_2}{2\sqrt{x}} \right) F(n_1,n_2), 
\label{final cond1} \\
f_{2,1}^{(n_1,n_2)} &=& \frac{1}{2} 
F(n_1+1,n_2)+ \frac{1}{2} \Big(x \frac{d^2}{dx^2} - (k^2+2n_1+ 3n_2) \frac{d}{dx} + \frac{(k^2+2)n_1}{x}  + \frac{k^2(k^2+2)}{4x} \nonumber \\
&& +\, \frac{n_1(n_1-1)}{x} + \frac{(6k^2+15)n_2}{4x} + \frac{9n_2(n_2-1)}{4x} + \frac{3n_1n_2}{x} \Big) F(n_1,n_2), \label{final cond2}  \\
f_{2,2}^{(n_1,n_2)} &=& -\,\frac{1}{2} F(n_1+1,n_2) + \frac{1}{2} \Big(x \frac{d^2}{dx^2} - (k^2+2n_1+ 3n_2) \frac{d}{dx} + \frac{(k^2+2)n_1}{x}  + \frac{k^2(k^2+2)}{4x} \nonumber \\
&& +\, \frac{n_1(n_1-1)}{x} + \frac{(6k^2+15)n_2}{4x} + \frac{9n_2(n_2-1)}{4x} + \frac{3n_1n_2}{x} \Big) F(n_1,n_2) .
\label{final cond3} 
\eea
% \bea
% F(0,0) &=& \tau_k, \\
% f_{1,1}^{(n_1,n_2)}
% &=& \left(\sqrt{x} \frac{d}{dx} - \frac{k^2}{2\sqrt{x}} - \frac{n_1}{\sqrt{x}} - \frac{3n_2}{2\sqrt{x}} \right) F(n_1,n_2), 
% \label{final cond1} \\
% f_{2,1}^{(n_1,n_2)} &=& \frac{1}{2} 
% \Big( k F(n_1,n_2) - \frac{2k}{\sqrt{x}} f_{1,1}^{(n_1,n_2)} - \frac{2n_1}{\sqrt{x}} \Big(f_{3,1}^{(n_1-1,n_2)} - f_{3,2}^{(n_1-1,n_2)} + f_{3,3}^{(n_1-1,n_2)} \Big) \\ 
% && - \frac{3n_2}{\sqrt{x}} \Big(f_{4,1}^{(n_1,n_2-1)} - f_{4,2}^{(n_1,n_2-1)} + f_{4,3}^{(n_1,n_2-1)} -f_{4,4}^{(n_1,n_2-1)}  \Big) 
% \Big) \\
% && + \frac{1}{2} \Big(x \frac{d^2}{dx^2} - (k^2+2n_1+ 3n_2) \frac{d}{dx} + \frac{(k^2+2)n_1}{x}  + \frac{k^2(k^2+2)}{4x} \nonumber \\
% && +\, \frac{n_1(n_1-1)}{x} + \frac{(6k^2+15)n_2}{4x} + \frac{9n_2(n_2-1)}{4x} + \frac{3n_1n_2}{x} \Big) F(n_1,n_2), \label{final cond2}  \\
% f_{2,2}^{(n_1,n_2)} &=& -\,\frac{1}{2} F(n_1+1,n_2) + \frac{1}{2} \Big(x \frac{d^2}{dx^2} - (k^2+2n_1+ 3n_2) \frac{d}{dx} + \frac{(k^2+2)n_1}{x}  + \frac{k^2(k^2+2)}{4x} \nonumber \\
% && +\, \frac{n_1(n_1-1)}{x} + \frac{(6k^2+15)n_2}{4x} + \frac{9n_2(n_2-1)}{4x} + \frac{3n_1n_2}{x} \Big) F(n_1,n_2) .
% \label{final cond3} 
% \eea
The above initial conditions are obtained in a similar way to that of (\ref{recursive formula for Gk}).

We defer the proof of Proposition \ref{final prop} to the end of this section. We  state next how the above recursive formulae are used when computing $F(n_1,n_2)$.

Firstly, from (\ref{the third term}),
we need to compute $F(n_{1},n_{2})$ for $0\leq n_{1}\leq 2M$ and $0\leq n_{2}\leq 2M-n_{1}$, so from (\ref{recursive formula for F}) and (\ref{F0}), it suffices to compute $\f^{(i',j')}_{m',1}$ for  $0\leq i' \leq n_1-1, 0 \leq j' \leq n_2$ and $m'=n_1+2n_2-i'-2j'$, for all $0\leq n_{2}\leq 2M$, $0\leq n_{1} \leq 2M-n_{2}$.
Secondly, to use the recursive formula (\ref{final prop:eq}) to compute $\f^{(i,j)}_{m,1}$, we start from the case $j=0$. 
When $j=0$, $\f^{(i,0)}_{m,1}$ corresponds to $\f^{(i)}_{m}$ in Theorem \ref{theorem 1 in note 4}, so we can compute $\f^{(i,0)}_{m,1}$.
When $j\geq 1$, we shall use (\ref{final prop:eq}) to compute $\f^{(i,j)}_{m,1}$ recursively with respect to $j$. To be more precise, suppose we already have $\f^{(i,0)}_{m,1}, \ldots, \f^{(i,j-1)}_{m,1}$ for any $i,m$, then we shall use (\ref{final prop:eq}) recursively to compute $\f^{(0,j)}_{m,1}, \cdots, \f^{(i,j)}_{m,1}$ for any $m$. 
Here for $\f^{(0,j)}_{m,1}$ we use the recursive formula (\ref{recursive formula i=0}) with initial conditions (\ref{final cond1})-(\ref{final cond3}). 
Now suppose we already have $\f^{(0,j)}_{m,1}, \cdots, \f^{(i-1,j)}_{m,1}$ for any $m$. 
For $\f^{(i,j)}_{m,1}$, when using (\ref{final prop:eq}) initially, we know all terms except the first three. With the initial conditions (\ref{final cond0})-(\ref{final cond3}), we can use this recursive formula with respect to $m$. In the above process, we indeed only need to compute a finite number of $\f^{(i,j)}_{m,1}$ with $i\leq 2M-j $ and $ m\leq 4M-i-2j+1$.

Theorem \ref{theorem2 in note 4} represents $x^{l}f_l(x)$ as derivatives of $\tau_k$. The highest order of the derivative is $l$. Here, for $x^{n_{1}+\frac{3}{2}n_2}F(n_1,n_2)$ we have a similar structure with the highest order of derivative equals $n_1+2n_2$. Using this result and a similar argument to that of Proposition \ref{intro:prop},
for any given integers $k\geq 1$ and any integer $M$ with $0\leq M\leq k$, we have
\beas
F_3(M,k) = \frac{G^2(k+1)}{G(2k+1)} R_{3,M}(k),
\eeas
where $R_{3,M}(k)$ is a rational function which is analytic when ${\rm Re}(k)>M-1/2$.

\begin{proof}[Proof of Proposition \ref{final prop}]
Firstly, we use a similar argument to that of Proposition \ref{prop:iteraive formula 2},
\beas
A_m \begin{pmatrix}
    G_{k,Y_{m,1}} \\
    \vdots \\
    G_{k,Y_{m,m}} \\
\end{pmatrix}
&=& \begin{pmatrix}
    -T G_{k,Y_{m-1,1}} \\
    \vdots \\
    -T G_{k,Y_{m-1,m-1}} \\
    0
\end{pmatrix}
+\begin{pmatrix}
    0 \\
    \sum_{h=2}^{j-1}(-1)^h T_{h-2} G_{k,Y_{m-h,j-h}} \\
    \vdots \\
    T_{m-2}G_k
\end{pmatrix}_{j=3,\ldots,m} \\
&& \hspace{-3cm} - \frac{1}{\sqrt{x}} \begin{pmatrix}
    0 \\
    \sum_{h=2}^{j-1}(-1)^h S_{h-1} G_{k,Y_{m-h,j-h}} \\
    \vdots \\
    S_{m-1}G_k
\end{pmatrix}_{j=3,\ldots,m} - \frac{2t_1}{\sqrt{x}} \begin{pmatrix}
    0 \\
    \sum_{h=2}^{j-1}(-1)^h T_{h+1} G_{k,Y_{m-h,j-h}} \\
    \vdots \\
    T_{m+1}G_k
\end{pmatrix}_{j=3,\ldots,m} \\
&& \hspace{-3cm} - \frac{3t_2}{\sqrt{x}} \begin{pmatrix}
    0 \\
    \sum_{h=2}^{j-1}(-1)^h T_{h+2} G_{k,Y_{m-h,j-h}} \\
    \vdots \\
    T_{m+2}G_k
\end{pmatrix}_{j=3,\ldots,m}
\eeas
The first four terms can be handled in a similar way to that of Proposition \ref{prop:iteraive formula 2}.
The difference is that in (\ref{J4}), $T_1G_{k,Y_{l,q}}$ should be replaced with (\ref{partial derivative 0}) rather than (\ref{412320prop3:eq}) because of the extra variable $t_2$.
Regarding the last term, we have
\beas
\begin{pmatrix}
    0 \\
    \sum_{h=2}^{j-1}(-1)^h T_{h+2} G_{k,Y_{m-h,j-h}} \\
    \vdots \\
    T_{m+2}G_k
\end{pmatrix}_{j=3,\ldots,m}
&=& \begin{pmatrix}
    \vdots \\
    \sum_{h=1}^{j-1}(-1)^h T_{h} G_{k,Y_{m+2-h,j-h}} \\
    \vdots \\
    T_{m+2}G_k
\end{pmatrix}_{j=4,\ldots,m+2} \\
&& 
\hspace{-2cm}
-\, \begin{pmatrix}
    -T_{3} G_{k,Y_{m-1,1}} \\
    \vdots \\
    -T_{3} G_{k,Y_{m-1,m-1}} \\
    0
\end{pmatrix} 
- \begin{pmatrix}
    T_{2} G_{k,Y_{m,2}} \\
    \vdots \\
    T_{2} G_{k,Y_{m,m}} \\
    0
\end{pmatrix}
- \begin{pmatrix}
    -T_{1} G_{k,Y_{m+1,3}} \\
    \vdots \\
    -T_{1} G_{k,Y_{m+1,m+1}} \\
    0
\end{pmatrix}.
\eeas
To simplify the notation, we denote
\[
\G_{m,i}:=
\begin{pmatrix}
    G_{k,Y_{m,i}} \\
    \vdots \\
    G_{k,Y_{m,m}} \\
\end{pmatrix}.
\]
By (\ref{partial derivative 2})-(\ref{partial derivative 0}), we then obtain
\beas
\G_{m,1} &=& B^{(m)}(- \sqrt{x} \frac{d}{dx} +  \frac{k^2+m-1}{2\sqrt{x}} -  \frac{t_1}{\sqrt{x}}\frac{\partial }{\partial t_1} -  \frac{3t_2}{2\sqrt{x}}\frac{\partial }{\partial t_2})\begin{pmatrix}
    \G_{m-1,1} \\
    0
\end{pmatrix} \\
&& + C_2^{(m)} \G_{m-2,1} - \frac{1}{\sqrt{x}} C_1^{(m)} \G_{m-1,1}+ \frac{2t_1}{\sqrt{x}} C_3^{(m)} \G_{m+1,1} - \frac{3t_2}{\sqrt{x}} C_4^{(m)} \G_{m+2,1} \\
&& + B^{(m)}( 2t_1 \frac{d}{dx}  -  \frac{t_1(k^2+m)}{x} -  (\frac{2t_1^2}{x}-\frac{3t_2}{\sqrt{x}}) \frac{\partial }{\partial t_1} 
 - \frac{3t_1t_2}{x} \frac{\partial }{\partial t_2} )
\begin{pmatrix}
   \G_{m,2} \\
   0
\end{pmatrix} \\
&& - 3t_2B^{(m)} ( \frac{d}{dx} - \frac{k^2+m+1+2t_1+3t_2}{2x} ) \begin{pmatrix}
   \G_{m+1,3} \\
   0
\end{pmatrix}.
\eeas

Secondly, we take the $n_1$-th and $n_2$-th derivatives with respect to $t_1$ and $t_2$ on both sides of the above equation. This leads to the claimed recursive formula in the proposition.
\end{proof}

In the $(d+1)$-th order derivative case, we can use a similar idea to define
\[
g_\beta(x,t_1,t_2,\ldots,t_{d}) = \sum_{n_1=0}^\infty \cdots \sum_{n_d=0}^\infty \frac{t_1^{n_1}\cdots t_d^{n_d}}{n_1!\cdots n_d!} I_{\beta+ \sum_{i=1}^{d} (i+1)n_i} (2\sqrt{x}),
\]
$G_{k,Y}(x,t_1,\ldots,t_d)$, $F(n_1,\ldots,n_d)$,  and $f^{(n_1,\ldots,n_d)}_{j,q}$, etc. We can also obtain recursive formulae for computing $F(n_1,\ldots,n_d)$ and $f^{(n_1,\ldots,n_d)}_{j,q}$.

\section{Conflicts of interest}

All authors certify that there are no conflicts of interest for this work.

\section{Data availability}

There is no data created in this work.

\begin{appendices}

\section{One lemma}
\label{app}

In this appendix, we establish a result for proving Proposition \ref{intro:prop}. The following lemma gives an expression of $F_2(M,k)$ for integers $k, M$. Moreover, the expression is analytic when ${\rm Re}(k) > M-1/2$. This lemma comes from a similar argument to \cite[Chapter 6]{hughes2001characteristic}. 

\begin{lem} 
\label{lem:Rational functions}
For any given integer $k\geq 1$ and any integer $M$ with $0\leq M \leq k$, 
\beas
F_2(M,k) 
&=& \frac{G^2(k+1)}{G(2k+1)} N^{k^2+4M} \sum_{\substack{n_1,n_2 \geq 0 \\ n_1+n_2 \leq 2M}}  (-2)^{n_1} \binom{2M}{n_1} \binom{2M-n_1}{n_2} \\
&& \times \, \sum_{m=0}^{4M} (-1)^m \frac{1}{m!}  (\frac{n_1}{2}+n_2)^m Z_{k,n_1,n_2}^{(4M-m)},
\eeas
which is analytic when ${\rm Re}(k) > M-1/2$.
Here
$Z_{k,n_1,n_2}^{(4M-m)}$ is determined by the following equality
\[
\det(d_{i,j})_{i,j=1,\ldots,n} =  \sum_{j=0}^{4M} Z_{k,n_1,n_2}^{(j)} (\i N\beta)^j + O_N(\beta^{4M+1}),
\]
where $n=n_1+n_2$. In addition, when $i=1,\ldots,n_1$
\[
d_{ij} = \sum_{m=0}^{4M} \frac{(2k-n+i-1)!}{(2k-n+i-1+m)!} \binom{i+k-n-1+m}{m} \binom{i+m-1}{j-1} (\i N\beta)^m,
\]
when $i=n_1+1,\ldots,n$
\beas
d_{ij} &=& \sum_{m=0}^{4M} \frac{(2k-n+i-1)!}{(2k-n+i-1+m)!} \binom{i+k-n-1+m}{m} \binom{i+m-1}{j-1} \\ 
&& \times \, \left(\sum_{l=0}^m \frac{(i-1)! m!}{(i-1+l)! (m-l)!} \binom{i-n_1-1+l}{i-n_1-1} \right) (\i N\beta)^m,
\eeas
where for any $m,n$
\[
\binom{m}{n}:= \frac{m (m-1) \cdots (m-n+1)}{n!}.
\]
\end{lem}

\begin{proof}
Define $\mathcal{Z}_A(\theta) = \overline{\Lambda_A(e^{\i \theta})}, V_A(\theta) = \overline{Z_A(e^{\i \theta})}$. Note that $V_{A}(\theta)$ is real when $\theta$ is real.
By definition
\beas
&& \int_{\U(N)} V_A''(0)^{2M} V_A(0)^{2k-2M} dA_N \\
&=&
\lim_{\beta \rightarrow 0} \frac{1}{\beta^{4M}}
\int_{\U(N)} (V_A(2\beta) -2V_A(\beta)+V_A(0))^{2M} V_A(0)^{2k-2M} dA_N \\
&=& \sum_{\substack{n_1,n_2 \geq 0 \\ n_1+n_2 \leq 2M}} (-2)^{n_1} \binom{2M}{n_1} \binom{2M-n_1}{n_2} \lim_{\beta \rightarrow 0} \frac{1}{\beta^{4M}} \int_{\U(N)} V_A(\beta)^{n_1} V_A(2\beta)^{n_2}V_A(0)^{2k-n_1-n_2} dA_N .
\eeas
The above integral has the same main term as that of $\int_{\U(N)} |Z_A''(1)|^{2M} |Z_A(1)|^{2k-2M} dA_N$.
It is not hard to show that
\[
V_A(\beta)^{n_1} V_A(2\beta)^{n_2}V_A(0)^{2k-n_1-n_2}
=e^{-\i N \frac{\beta}{2} n_1} e^{-\i N \beta n_2} 
\mathcal{Z}_A(0)^k \overline{\mathcal{Z}_A(0)}^{k-n_1-n_2} \overline{\mathcal{Z}_A(\beta)}^{n_1} \overline{\mathcal{Z}_A(2\beta)}^{n_2} .
\]

We next use a similar method to \cite[(6.19)]{hughes2001characteristic} to estimate 
\be
\int_{\U(N)} 
\mathcal{Z}_A(0)^k \overline{\mathcal{Z}_A(0)}^{k-n_1-n_2} \overline{\mathcal{Z}_A(\beta)}^{n_1} \overline{\mathcal{Z}_A(2\beta)}^{n_2}
d A_N.
\label{integral}
\ee
Similar to \cite[(6.19)]{hughes2001characteristic}, by the Heine identity, we obtain that $(\ref{integral}) = D_N[f]$, where $D_N[f]$ is the Toeplitz determinant with symbol
\[
f(\theta) = (-1)^k e^{-\i k\theta} \prod_{j=1}^{2k} (e^{\i \theta} - e^{\i \alpha_j}).
\]
Here $\alpha_1=\cdots=\alpha_{2k-n}=0, \alpha_{2k-n+1}=\cdots=\alpha_{2k-n+n_1} = \beta, \alpha_{2k-n+n_1+1}=\cdots=\alpha_{2k} = 2 \beta$. We then use the trick of \cite{basor1994formulas} to compute this Toeplitz determinant.
Finally, we obtain
\[
(\ref{integral})
= M_N(2k) \det(s_{i,j})_{i,j=1,\ldots,n},
\]
where 
\[
M_N(2k) = \frac{G^2(1+k) G(N+1) G(N+1+2k)}{G(1+2k) G^2(N+1+k)}.
\]
In addition, if $1\leq i \leq n_1$,
\[
s_{i,j} = \sum_{m=0}^{N+j-i} \frac{(2k-n+i-1)!}{(2k-n+i-1+m)!} \binom{i+k-n-1+m}{m} \binom{i+m-1}{j-1} \frac{N!}{(N+j-i-m)!} (e^{\i \beta}-1)^m,
\]
and if $n_1+1\leq i \leq n$,
\beas
s_{i,j} 
&=& \sum_{m=0}^{N+j-i} \frac{(2k-n+i-1)!}{(2k-n+i-1+m)!} \binom{i+k-n-1+m}{m} \binom{i+m-1}{j-1} \frac{N!}{(N+j-i-m)!} \\
&& \times \, \sum_{l=0}^m \frac{(i-1)! m!}{(i-1+l)! (m-l)!} \binom{i-n_1-1+l}{i-n_1-1}  (e^{\i \beta}-1)^{m-l} (e^{2\i \beta}-e^{\i \beta})^{l} .
\eeas

Note that $M_N(2k)=\frac{G^2(k+1)}{G(2k+1)} N^{k^2} + O(N^{k^2-1})$. Moreover,
if $1\leq i \leq n_1$,
\beas
s_{i,j} &\sim& \sum_{m=0}^{4M} \frac{(2k-n+i-1)!}{(2k-n+i-1+m)!} \binom{i+k-n-1+m}{m} \binom{i+m-1}{j-1} \\ 
&& \times \, N^{m+i-j} ((\i \beta)^m + O(\beta^{m+1})) + O_N(\beta^{4M+1}).
\eeas
If $n_1+1\leq i \leq n$,
\beas
s_{i,j} 
&\sim& \sum_{m=0}^{4M} \frac{(2k-n+i-1)!}{(2k-n+i-1+m)!} \binom{i+k-n-1+m}{m} \binom{i+m-1}{j-1}  \\
&& \left( \sum_{l=0}^m \frac{(i-1)! m!}{(i-1+l)! (m-l)!} \binom{i-n_1-1+l}{i-n_1-1}  \right) N^{m+i-j} ((\i \beta)^m + O(\beta^{m+1})) + O_N(\beta^{4M+1}).
\eeas
For the matrix $(s_{ij})_{i,j=1,\ldots,n}$,
multiplying the $i$-th row by $1/N^i$ and the $j$-th column by $N^j$ for each $i,j$, we then have
$\det(s_{i,j})_{i,j=1,\ldots,n} \sim \det(d_{i,j})_{i,j=1,\ldots,n}$.

In the expression of $d_{ij}$, for $m\geq 1$
\[
\frac{(2k-n+i-1)!}{(2k-n+i-1+m)!}=
\frac{1}{(2k-n+i-1+m) (2k-n+i-1+m-1) \cdots (2k-n+i)}.
\]
Note that $n\leq 2M$ and $i\geq 1$, so it is analytic when ${\rm Re}(k) > M-1/2$.
\end{proof}

\section{Numerical data}

The following are expressions $F_2(k,k)$ for $k=1,\ldots,9$:
\beas
&\displaystyle \frac{1}{2^4\cdot 5}& \\
&\displaystyle \frac{17}{2^{10} \cdot3^{3} \cdot5 \cdot7 \cdot11}& \\
&\displaystyle \frac{11593}{2^{18} \cdot3^{7} \cdot5^{2} \cdot7^{3} \cdot11 \cdot13 \cdot17}& \\
&\displaystyle \frac{103 \cdot413129}{2^{28} \cdot3^{12} \cdot5^{5} \cdot7^{3} \cdot11^{2} \cdot13^{2} \cdot17 \cdot19 \cdot23}& \\
&\displaystyle \frac{2616269 \cdot322433}{2^{40} \cdot3^{17} \cdot5^{8} \cdot7^{5} \cdot11^{4} \cdot13^{3} \cdot17^{2} \cdot19 \cdot23 \cdot29}& \\
&\displaystyle \frac{53 \cdot5830411 \cdot94098709}{2^{54} \cdot3^{24} \cdot5^{13} \cdot7^{8} \cdot11^{4} \cdot13^{4} \cdot17^{3} \cdot19^{2} \cdot23 \cdot29 \cdot31}& \\
&\displaystyle \frac{896318952226585228351}{2^{70} \cdot3^{32} \cdot5^{16} \cdot7^{10} \cdot11^{6} \cdot13^{4} \cdot17^{4} \cdot19^{3} \cdot23^{2} \cdot29 \cdot31 \cdot37 \cdot41}& \\ 
&\displaystyle \frac{103 \cdot167 \cdot64283 \cdot71225030041520923}{2^{88} \cdot3^{42} \cdot5^{20} \cdot7^{13} \cdot11^{6} \cdot13^{6} \cdot17^{5} \cdot19^{4} \cdot23^{2} \cdot29^{2} \cdot31 \cdot37 \cdot41 \cdot43 \cdot47}& \\
&\displaystyle  \frac{ 109 \cdot 9335580613 \cdot 845744949032889042779 }{2^{108}  \cdot 3^{52} \cdot 5^{25}  \cdot  7 ^{17}  \cdot 11^{9}  \cdot 13 ^{8} \cdot 17^{5}  \cdot 19^{5}  \cdot 23^{4}  \cdot 29^{2}  \cdot 31^{2}  \cdot 37  \cdot 41  \cdot 43  \cdot 47  \cdot 53}&
\eeas

\end{appendices}

\bibliographystyle{plain}
\bibliography{main}

\end{document}